\documentclass[letterpaper]{article} %
\usepackage{ifthen}
\provideboolean{extended-version}
\setboolean{extended-version}{true}
\ifthenelse{\boolean{extended-version}}{
\usepackage[draft]{aaai2026}  %
}{
\usepackage{aaai2026}  %
}
\usepackage{times}  %
\usepackage{helvet}  %
\usepackage{courier}  %
\usepackage[hyphens]{url}  %
\usepackage{graphicx} %
\usepackage{natbib}  %
\usepackage{caption} %
\usepackage{newfloat}
\usepackage{listings}
\DeclareCaptionStyle{ruled}{labelfont=normalfont,labelsep=colon,strut=off} %
\usepackage{amssymb}
\usepackage{mathtools}
\usepackage{pgfplots}

\usepackage{ifthen}

\provideboolean{detectedSTOC}
\provideboolean{detectedFOCS}
\provideboolean{detectedElsevier}
\provideboolean{detectedNOW}
\provideboolean{detectedLMCS}
\provideboolean{detectedIEEE}
\provideboolean{detectedPoster}
\provideboolean{detectedSIAM}
\provideboolean{detectedLNCS}
\provideboolean{detectedACM}
\provideboolean{detectedACMconf}
\provideboolean{detectedSigplanconf}
\provideboolean{detectedToC}
\provideboolean{detectedLIPIcs}
\provideboolean{detectedAAAI}
\provideboolean{detectedIJCAI}
\provideboolean{detectedCompCplx}
\provideboolean{detectedEasyChair}
\provideboolean{detectedJAIR}
\provideboolean{detectedArticle}
\provideboolean{detectedReport}
\provideboolean{detectedThesis}

\makeatletter

\@ifclassloaded{sig-alternate}
{\setboolean{detectedSTOC}{true}}
{\setboolean{detectedSTOC}{false}}

\@ifclassloaded{elsarticle}
{\setboolean{detectedElsevier}{true}}
{\setboolean{detectedElsevier}{false}}

\@ifclassloaded{now}
{\setboolean{detectedNOW}{true}}
{\setboolean{detectedNOW}{false}}

\@ifclassloaded{lmcs}
{\setboolean{detectedLMCS}{true}}
{\setboolean{detectedLMCS}{false}}

\@ifclassloaded{IEEEtran} {
  \setboolean{detectedIEEE}{true}
  \ifCLASSOPTIONconference {
    \setboolean{detectedFOCS}{true}
  }
  \else {
    \setboolean{detectedFOCS}{false}
  }
  \fi
}
{
  \setboolean{detectedFOCS}{false}
  \setboolean{detectedIEEE}{false}
}

\@ifclassloaded{siamart171218}
{\setboolean{detectedSIAM}{true}}
{\setboolean{detectedSIAM}{false}}

\@ifclassloaded{llncs}
{\setboolean{detectedLNCS}{true}}
{\setboolean{detectedLNCS}{false}}

\@ifclassloaded{acmsmall}
{\setboolean{detectedACM}{true}}
{\setboolean{detectedACM}{false}}

\@ifclassloaded{acmart}
{\setboolean{detectedACMconf}{true}
 \setboolean{detectedACM}{true}}
{\setboolean{detectedACMconf}{false}}

\@ifclassloaded{sigplanconf}
{\setboolean{detectedSigplanconf}{true}}
{\setboolean{detectedSigplanconf}{false}}

\@ifclassloaded{toc}
{\setboolean{detectedToC}{true}}
{\setboolean{detectedToC}{false}}

\@ifclassloaded{lipics}
{\setboolean{detectedLIPIcs}{true}}
{\@ifclassloaded{lipics-v2019}
  {\setboolean{detectedLIPIcs}{true}}
  {\@ifclassloaded{oasics-v2019}
    {\setboolean{detectedLIPIcs}{true}}
    {\@ifclassloaded{lipics-v2021}
      {\setboolean{detectedLIPIcs}{true}}
      {\setboolean{detectedLIPIcs}{false}}
    }
  }
}

\@ifclassloaded{cc}
{\setboolean{detectedCompCplx}{true}}
{\setboolean{detectedCompCplx}{false}}

\@ifclassloaded{easychair}
{\setboolean{detectedEasyChair}{true}}
{\setboolean{detectedEasyChair}{false}}

\@ifpackageloaded{jair}
{\setboolean{detectedJAIR}{true}}        
{\setboolean{detectedJAIR}{false}}        

\@ifpackageloaded{aaai}
{\setboolean{detectedAAAI}{true}}        
{\@ifpackageloaded{aaai18}
  {\setboolean{detectedAAAI}{true}}
  {\@ifpackageloaded{aaai20}       
    {\setboolean{detectedAAAI}{true}}
    {\setboolean{detectedAAAI}{false}}}}

\@ifpackageloaded{ijcai18}
{\setboolean{detectedIJCAI}{true}}        
{\@ifpackageloaded{ijcai19}
  {\setboolean{detectedIJCAI}{true}}        
  {\setboolean{detectedIJCAI}{false}}}

\@ifclassloaded{sciposter}
{\setboolean{detectedPoster}{true}}
{\setboolean{detectedPoster}{false}}

\@ifclassloaded{article}
{\setboolean{detectedArticle}{true}}
{\setboolean{detectedArticle}{false}}

\makeatother

\ifthenelse{\not \isundefined{\examen} 
  \and \not \isundefined{\disputationsdatum} 
  \and \not \isundefined{\disputationslokal}}   
  {\setboolean{detectedThesis}{true}}
  {\setboolean{detectedThesis}{false}}

\ifthenelse{\boolean{detectedArticle} \or \boolean{detectedThesis}
  \or \boolean{detectedSTOC}    \or \boolean{detectedFOCS}
  \or \boolean{detectedSIAM}    \or \boolean{detectedIEEE}
  \or \boolean{detectedACMconf} \or \boolean{detectedACM}
  \or \boolean{detectedPoster}}
{\setboolean{detectedReport}{false}}
{\setboolean{detectedReport}{true}}

\setboolean{detectedAAAI}{true}

\usepackage{ifthen}
\usepackage{xspace}
\ifthenelse
{\boolean{detectedElsevier} \or \boolean{detectedSIAM} 
  \or \boolean{detectedLIPIcs}}
{}
{\usepackage{varioref}}

\DeclareMathAlphabet{\mathsfsl}{OT1}{cmss}{m}{sl}

\DeclareRobustCommand{\BibTeX}{%
  {\normalfont B\kern-.05em{\scshape i\kern-.025em b}\kern-.08em \TeX}%
}

\newcommand{\bigoh}[1]{\mathrm{O} ( #1 )}

\ifthenelse{\boolean{detectedToC}}{}
{

  \newcommand{\Nplus}     {\mathbb{N}^{+}}

}

\newcommand{\ceiling}[1]{\lceil #1 \rceil}

\newcommand{\MAXOFEXPR}[2][]{\max_{#1} \left\{ #2 \right\}}
\newcommand{\MINOFEXPR}[2][]{\min_{#1} \left\{ #2 \right\}}
\newcommand{\Maxofexpr}[2][]{\max_{#1} \bigl\{ #2 \bigr\}}
\newcommand{\Minofexpr}[2][]{\min_{#1} \bigl\{ #2 \bigr\}}

\newcommand{\MAXOFSET}[3][:]%
     {\ifthenelse{\equal{#1}{;}}%
     {\MAXOFEXPR{ #2 \,;\, #3 }}
     {\ifthenelse{\equal{#1}{:}}%
     {\MAXOFEXPR{ #2 \,:\, #3 }}
     {\max \twincommandJN{\left\{}{#2}{\left#1}{\right}{\,#3}{\right\}}}}}
\newcommand{\MINOFSET}[3][:]%
     {\ifthenelse{\equal{#1}{;}}%
     {\MINOFEXPR{ #2 \,;\, #3 }}
     {\ifthenelse{\equal{#1}{:}}%
     {\MINOFEXPR{ #2 \,:\, #3 }}
     {\min \twincommandJN{\left\{}{#2}{\left#1}{\right}{\,#3}{\right\}}}}}

\newcommand{\Maxofset}[3][:]%
     {\ifthenelse{\equal{#1}{;}}%
     {\Maxofexpr{ #2 \,;\, #3 }}
     {\ifthenelse{\equal{#1}{:}}%
     {\Maxofexpr{ #2 \,:\, #3 }}
     {\max \twincommandJN{\bigl\{}{#2}{\bigl#1}{\bigr}{\,#3}{\bigr\}}}}}
\newcommand{\Minofset}[3][:]%
     {\ifthenelse{\equal{#1}{;}}%
     {\Minofexpr{ #2 \,;\, #3 }}
     {\ifthenelse{\equal{#1}{:}}%
     {\Minofexpr{ #2 \,:\, #3 }}
     {\min \twincommandJN{\bigl\{}{#2}{\bigl#1}{\bigr}{\,#3}{\bigr\}}}}}

\DeclareMathOperator{\Expop}{E}

\ifthenelse{\boolean{detectedLMCS}}
{}
{}

\newcommand{\twincommandJN}[6]%
    {#1#2#3\vphantom{#2#5}\mspace{-2.05mu}#4.#5#6}

\newcommand{\CondExp}[2]%
    {\Expop\twincommandJN{\bigl[}{#1}{\bigl|}{\bigr}{\,#2}{\bigr]}}
\newcommand{\CONDEXP}[2]%
     {\Expop\twincommandJN{\left[}{#1}{\left|}{\right}{\,#2}{\right]}}

\newcommand{\Condprob}[3][]%
    {\Pr_{#1}\twincommandJN{\bigl[}{#2}{\bigl|}{\bigr}{\,#3}{\bigr]}}
\newcommand{\CONDPROB}[3][]%
    {\Pr_{#1}\twincommandJN{\left[}{#2}{\left|}{\right}{\,#3}{\right]}}

\newcommand{\set}[1]{\{ #1 \}}
\newcommand{\Set}[1]{\bigl\{ #1 \bigr\}}

\newcommand{\Setdescr}[3][|]%
     {\ifthenelse{\equal{#1}{;}}%
     {\Set{ #2 \,;\, #3 }}
     {\ifthenelse{\equal{#1}{:}}%
     {\Set{ #2 \,:\, #3 }}
     {\twincommandJN{\bigl\{}{#2\,}{\bigl#1}{\bigr}{\,#3}{\bigr\}}}}}
\newcommand{\SETDESCR}[3][|]%
     {\twincommandJN{\left\{}{#2\,}{\left#1}{\right}{\,#3}{\right\}}}

\newcommand{\Setdescrbrackets}[3][|]%
     {\twincommandJN{\bigl[}{#2}{\bigl#1}{\bigr}{\,#3}{\bigr]}}
\newcommand{\SETDESCRBRACKETS}[3][|]%
     {\twincommandJN{\left[}{#2}{\left#1}{\right}{\,#3}{\right]}}

\newcommand{\setsize}[1]{\lvert#1\rvert}

\newcommand{\union}{\cup}

\newcommand{\olnot}[1]{\overline{#1}}

\newcommand{\synteq}{\doteq}

\newcommand{\nvar}{n}

\newcommand{\nclause}{m}

\newcommand{\clwidth}{k}

\newcommand{\randkcnfnclwrepl}[3][\clwidth]%
        {\ensuremath{\mathcal{F}^{#2, #3}_{#1}}}
\newcommand{\randkcnfnclwreplstd}%
        {\randkcnfnclwrepl{\clwidth}{\nvar}{\nclause}}

\newcommand{\complclassformat}[1]%
        {\textrm{\upshape{\textsf{#1}}}\xspace}
\newcommand{\cocomplclass}[1]%
        {\textrm{\upshape{\textsf{co#1}}}\xspace}

\newcommand{\DTIMEadviceclass}[2]%
    {\ensuremath{\complclassformat{DTIME}\bigl(#1\bigr)/{#2}}}

\newcommand{\PCPalph}[5]%
    {\ensuremath{\complclassformat{PCP}_{{#1},{#2}}[{#3}, {#4}, {#5}]}}
\newcommand{\PCP}[4]%
    {\ensuremath{\complclassformat{PCP}_{{#1},{#2}}[{#3}, {#4}]}}

\newcommand{\eqperiod}{\enspace .}
\newcommand{\eqcomma}{\enspace ,}
\renewcommand{\eqperiod}{\, .}
\renewcommand{\eqcomma}{\, ,}

\newcommand{\egabbrev}{e.g.,\ }

\ifthenelse{\boolean{detectedToC}}{}
  { %
    \newcommand{\Eg}{For instance\xspace}}
\ifthenelse{\boolean{detectedToC}}{}{
\newcommand{\ie}{i.e.,\ }

}

\ifthenelse{\boolean{detectedLIPIcs} \or \boolean{detectedIJCAI}}
{\renewcommand{\st}{\errmessage{Please do not use st}}}
{\ifthenelse{\isundefined{\st}}
  {\newcommand{\st}{such that\xspace}}
  {}}

\newcommand{\wolog}{without loss of generality\xspace}
\newcommand{\Wolog}{Without loss of generality\xspace}

\ifthenelse{\boolean{detectedIEEE}}{}{}

\newcommand{\refsec}[1]{Section~\ref{#1}}

\newcommand{\reftwosecs}[2]{Sections~\ref{#1} and~\ref{#2}}

\newcommand{\refapp}[1]{Appendix~\ref{#1}}

\ifthenelse
{\isundefined{\refeq}}
{\newcommand{\refeq}[1]{\eqref{#1}}}
{\renewcommand{\refeq}[1]{\eqref{#1}}}

\newcommand{\derives}{\vdash}

\newcommand{\formf}{\ensuremath{F}}
\newcommand{\formg}{\ensuremath{G}}

\newcommand{\varx}{\ensuremath{x}}
\ifthenelse{\boolean{detectedACMconf}}
{\providecommand{\vary}{\ensuremath{y}}}
{\newcommand{\vary}{\ensuremath{y}}}
\newcommand{\varz}{\ensuremath{z}}

\ifthenelse{\boolean{detectedACMconf}}
{
 }
{
 }

\newcommand{\constrc}{\ensuremath{C}}

\newcommand{\SETSOFVARSORLIT}[2]%
        {\mathit{#1}\left({#2}\right)}
\newcommand{\setsofvarsorlit}[2]%
        {\mathit{#1}({#2})}
\newcommand{\Setsofvarsorlit}[2]%
        {\mathit{#1}\bigl({#2}\bigr)}

\newcommand{\derivabbrev}[2]{\bigl( #1 \vdash #2 \bigr)}
\newcommand{\derivabbrevsmall}[2]{( #1 \vdash #2 )}
\newcommand{\derivabbrevcompact}[2]{\bigl( #1 \vdash #2 \bigr)}

\newcommand{\refutabbrevsmall}[1]{\derivabbrevsmall{#1}{\!\bot}}
\newcommand{\refutabbrevcompact}[1]{\derivabbrevcompact{#1}{\!\bot}}

\newcommand{\genericrefsmall}[3]%
    {{\mathit{#1}}_{#2}\refutabbrevsmall{#3}}
\newcommand{\genericrefcompact}[3]%
    {{\mathit{#1}}_{#2}\refutabbrevcompact{#3}}
\newcommand{\genericderiv}[4]%
    {{\mathit{#1}}_{#2}\derivabbrev{#3}{#4}}
\newcommand{\genericderivsmall}[4]%
    {{\mathit{#1}}_{#2}\derivabbrevsmall{#3}{#4}}
\newcommand{\genericderivcompact}[4]%
    {{\mathit{#1}}_{#2}\derivabbrevcompact{#3}{#4}}
\newcommand{\generictaut}[3]%
    {{\mathit{#1}}_{#2}\derivabbrev{}{#3}}
\newcommand{\generictautcompact}[3]%
    {{\mathit{#1}}_{#2}\derivabbrevcompact{}{#3}}
\newcommand{\generictautsmall}[3]%
    {{\mathit{#1}}_{#2}\derivabbrevsmall{}{#3}}

\newcommand{\formulaformat}[1]{\mathit{#1}}

\newcommand{\extendedversion}[1]{\widetilde{#1}}

\newcommand{\epopnot}[1]%
    {\extendedversion{\formulaformat{POP}}_{#1}}

\newcommand{\elopnot}[1]%
    {\extendedversion{\formulaformat{LOP}}_{#1}}

\newcommand{\ephpnot}[2]%
    {\vphantom{\extendedversion{\formulaformat{PHP}}}
      {\smash{\extendedversion{\formulaformat{PHP}}}
        \vphantom{\formulaformat{PHP}}}^{#1}_{#2}}

\newcommand{\efphpnot}[2]%
    {\vphantom{\extendedversion{\formulaformat{FPHP}}}
      {\smash{\extendedversion{\formulaformat{FPHP}}}
        \vphantom{\formulaformat{FPHP}}}^{#1}_{#2}}

\newcommand{\ontophpnot}[2]%
    {\formulaformat{Onto}\text{-}\formulaformat{PHP}^{#1}_{#2}}
\newcommand{\ontofphpnot}[2]%
    {\formulaformat{Onto}\text{-}\formulaformat{FPHP}^{#1}_{#2}}

\newcommand{\graphontophpnot}[1][G]%
    {\text{$\formulaformat{Onto}$-$\formulaformat{PHP}$}({#1})}

\newcommand{\perfectmatchingnot}[1][G]%
    {\formulaformat{PM}({#1})}

\usepackage[utf8]{inputenc}
\usepackage{tikz}
\usepackage{tabularx}
\usepackage{mathtools}
\usepackage{verbatim}
\usepackage{adjustbox}
\usepackage{array}
\usepackage{multirow}
\usepackage{pifont}
\usepackage[linesnumbered,ruled,commentsnumbered,longend]{algorithm2e}
\usetikzlibrary{shapes.misc, positioning}
\usetikzlibrary{shapes,fit}
\usepackage[edges]{forest}
\usetikzlibrary{shapes,arrows,positioning,calc,decorations.pathmorphing,arrows.meta,decorations.pathreplacing,calligraphy}
\usepackage{amsmath,amsthm}
\usepackage{todonotes}

\newtheorem{theorem}{Theorem}
\newtheorem{lemma}{Lemma}

\newtheorem{remark}{Remark}

\theoremstyle{definition}
\newtheorem{definition}{Definition}

\usepackage{mathptmx}

\providecommand{\zeroone}{$0$--$1$\xspace}

\DeclareMathOperator{\var}{var}
\DeclareMathOperator{\supp}{supp}  %
\DeclareMathOperator{\img}{image}
\newcommand{\uh}{\!\!\upharpoonright}
\newcommand\subst[2]{#1\uh_{#2}}
\newcommand{\core}{\mathcal{C}}
\newcommand{\der}{\mathcal{D}}
\newcommand{\ord}{\mathcal{O}_{\preceq}}
\newcommand{\spec}{\mathcal{S}_{\preceq}}
\providecommand{\specification}{\mathcal{S}}

\providecommand{\trivialorder}{\mathcal{O}_\top}
\providecommand{\varset}{\va}
\providecommand{\assmnta}{\alpha}
\providecommand{\assmntaprime}{\alpha^\prime}
\providecommand{\assmntb}{\beta}
\providecommand{\assmntc}{\gamma}
\providecommand{\assmntp}{\varrho}

\providecommand{\assmntone}{\assmnta}
\providecommand{\assmnttwo}{\assmntaprime}

\providecommand{\witness}{\omega}
\newcommand{\conf}{(\core, \der, \ord, \spec, \vec{z}, \vec{a}, v)}
\newcommand{\confdec}{(\core, \der, \ord, \spec, \vec{z}, \vec{a})}
\providecommand{\lexorder}{\preceq_{\mathrm{lex}}}

\providecommand{\toolformat}[1]{\textsc{#1}\xspace}
\providecommand{\satsuma}{\toolformat{satsuma}}
\providecommand{\veripb}{\toolformat{VeriPB}}
\providecommand{\cakepb}{\toolformat{CakePB}}

\providecommand{\cnfgen}{\toolformat{CNFgen}}
\providecommand{\drat}{\toolformat{DRAT}}

\providecommand{\m}[1]{\ensuremath{#1}}
\providecommand{\varx}{\m{x}}
\providecommand{\vary}{\m{y}}
\providecommand{\varz}{\m{z}}
\providecommand{\auxa}{\m{a}}
\providecommand{\auxd}{\m{d}}
\providecommand{\auxb}{\m{b}}
\providecommand{\auxe}{\m{e}}
\providecommand{\auxc}{\m{c}}
\providecommand{\auxf}{\m{f}}
\providecommand{\coeffapp}{\m{p}}
\providecommand{\degapp}{\m{P}}
\providecommand{\auxs}{\m{s}}
\providecommand{\auxt}{\m{t}}
\providecommand{\lexorder}{\preceq_{\mathrm{lex}}}
\providecommand{\nvar}{\m{n}}
\providecommand{\idxi}{\m{i}}
\providecommand{\suppsize}{\m{k}}
\providecommand{\subidx}{\m{l}}
\providecommand{\symmetry}{\sigma}
\providecommand{\suppsymmetry}{\supp(\symmetry)}
\providecommand{\suppsymmetryidxs}{S}

\providecommand{\formg}{\m{G}} %

\newcommand{\va}{\vec{\auxa}}
\newcommand{\vb}{\vec{\auxb}}
\newcommand{\vc}{\vec{\auxc}}
\newcommand{\vd}{\vec{\auxd}}
\newcommand{\ve}{\vec{\auxe}}
\newcommand{\vf}{\vec{\auxf}}
\newcommand{\vecs}{\vec{\auxs}}
\newcommand{\vu}{\vec{u}}
\newcommand{\vv}{\vec{v}}
\newcommand{\vx}{\vec{\varx}}
\newcommand{\vy}{\vec{\vary}}
\newcommand{\vz}{\vec{\varz}}

\newcommand{\listt}{list\xspace}    %
\newcommand{\lists}{lists\xspace}  %

\newcommand{\lexord}{\mathcal{O}_{\lexorder}}

\newcommand{\original}{original\xspace}    %
\newcommand{\extended}{extended\xspace}    %
\newcommand{\old}{old\xspace}              %
\newcommand{\new}{new\xspace}              %

\newcommand{\slack}[2]{\mathit{slack(#1; #2)}}

\usepackage{listings}
\PassOptionsToPackage{usenames,dvipsnames}{color}
\usepackage{xcolor}

\definecolor{comment}{HTML}{7F7F7F}
\definecolor{rule}{HTML}{95133B}
\definecolor{environment}{HTML}{246272}
\definecolor{constant}{HTML}{934E11}
\definecolor{operator}{HTML}{95133B}
\definecolor{separator}{HTML}{000000}
\definecolor{identifier}{HTML}{000000}
\definecolor{integers}{HTML}{5B2B7C}

\lstdefinelanguage{veripb}{
  backgroundcolor=\color{white},
  commentstyle=\color{comment},
  keywordstyle=[1]{\color{integers}},
  keywordstyle=[2]{\color{rule}},
  keywordstyle=[3]{\color{environment}},
  keywordstyle=[4]{\color{constant}},
  identifierstyle=\itshape\color{identifier},
  basicstyle=\color{integers}\ttfamily\scriptsize,
  alsoletter={_\{\}\^},
  keywordsprefix=@,%
  breaklines=true,
  keepspaces=true,
  columns=fullflexible,
  captionpos=b,
  showstringspaces=false,
  showtabs=false,
  tabsize=4,
  frame=single,
  rulecolor=\color{black},
  morecomment=[l]{\%},
  morekeywords=[1]{},
  morekeywords=[2]{del,delc,pbc,red,dom,obju,output,conclusion,pol,rup,id,range,deld,load_order,core,new,diff,f,e,eobj,eord_def,eord_loaded,i,ia,strengthening_to_core,sol,soli,solx,start_time,end_time,is_deleted,fail,left,right,aux,fresh_right,fresh_aux_1,fresh_aux_2,setlvl,wiplvl,a},
  morekeywords=[3]{subproof,qed,proofgoal,proof,version,end,def_order,proof,vars,def,spec,transitivity,reflexivity,opb,scope},
  morekeywords=[4]{NONE,DERIVABLE,EQUISATISFIABLE,EQUIOPTIMAL,EQUIENUMERABLE,FILE,CONSTRAINTS,IMPLICIT,PERMUTATION,SAT,UNSAT,BOUNDS,on,off,INF,leq,geq},
  escapeinside={|}{|},
  literate=%
    *{pseudo-Boolean}{{{\color{environment}pseudo-Boolean}}}{14}
    {del-spec}{{{\color{rule}del spec}}}{8}
    {->}{{{\color{separator}->}}}{2}
    {,}{{{\color{separator},}}}{1}
    {:}{{{\color{separator}:}}}{1}
    {;}{{{\color{separator};}}}{1}
    {[}{{{\color{separator}[}}}{1}
    {]}{{{\color{separator}]}}}{1}
    {...}{{{\color{black}...}}}{3}
    {~}{{\raisebox{0.5ex}{\color{operator}\texttildelow}}}{1}
    {>=}{{{\color{operator}>=}}}{2}
    {<=}{{{\color{operator}<=}}}{2}
    {==>}{{{\color{operator}==>}}}{3}
    {<==}{{{\color{operator}<==}}}{3}
}

\title{Faster Certified Symmetry Breaking Using Orders With Auxiliary Variables
\ifthenelse{\boolean{extended-version}}{
  (Extended~Version Including Appendix)
}{}
}

\author{
    Markus Anders\textsuperscript{\rm 1},
    Bart Bogaerts\textsuperscript{\rm 2,3},
    Benjamin Bogø\textsuperscript{\rm 4,5},
    Arthur Gontier\textsuperscript{\rm 6},
    Wietze Koops\textsuperscript{\rm 5,4},
    Ciaran~McCreesh\textsuperscript{\rm 6},
    Magnus O. Myreen\textsuperscript{\rm 7,8},
    Jakob Nordström\textsuperscript{\rm 4,5},
    Andy Oertel\textsuperscript{\rm 5,4}, 
    Adrian~Rebola-Pardo\textsuperscript{\rm 9,10},
    Yong Kiam Tan\textsuperscript{\rm 11,12}
}
\affiliations{
    \textsuperscript{\rm 1}RPTU Kaiserslautern-Landau, Kaiserslautern, Germany\\
    \textsuperscript{\rm 2}KU Leuven, Dept.\ of Computer Science, Leuven, Belgium\\
    \textsuperscript{\rm 3}Vrije Universiteit Brussel, Dept.\ of Computer Science, Brussels, Belgium\\
    \textsuperscript{\rm 4}University of Copenhagen, Copenhagen, Denmark\\
    \textsuperscript{\rm 5}Lund University, Lund, Sweden\\
    \textsuperscript{\rm 6}University of Glasgow, Glasgow, Scotland \\ %
    \textsuperscript{\rm 7}Chalmers University of Technology, Gothenburg, Sweden\\
    \textsuperscript{\rm 8}University of Gothenburg, Gothenburg, Sweden\\
    \textsuperscript{\rm 9}Vienna University of Technology, Vienna, Austria\\
    \textsuperscript{\rm 10}Johannes Kepler University Linz, Linz, Austria\\
    \textsuperscript{\rm 11}Nanyang Technological University, Singapore\\
    \textsuperscript{\rm 12}Institute for Infocomm Research (I$^2$R), A*STAR, Singapore\\

    anders@cs.uni-kl.de, bart.bogaerts@kuleuven.be, bebo@di.ku.dk, Arthur.Gontier@glasgow.ac.uk, wietze.koops@cs.lth.se, Ciaran.McCreesh@glasgow.ac.uk, myreen@chalmers.se, jn@di.ku.dk, andy.oertel@cs.lth.se, adrian.rebola@tuwien.ac.at, yongkiam.tan@ntu.edu.sg
}

\begin{document}

\maketitle

\begin{abstract}
  Symmetry breaking is a crucial technique in modern combinatorial
  solving, but it is difficult to be sure it is implemented correctly.
  The most successful approach to deal with bugs is to make solvers
  \emph{certifying}, so that they output not just a solution,  but
  also a mathematical proof of correctness in a standard format, which
  can then be checked by a formally verified checker.
  This requires justifying symmetry reasoning within the proof, but
  developing efficient methods for this has remained a long-standing
  open challenge.
  A fully general approach was recently proposed by
  Bogaerts et al.~(2023),
  but it relies on encoding lexicographic orders with big integers,
  which quickly becomes infeasible for 
  large symmetries.
  In this work, 
  we develop a method for instead encoding orders with auxiliary variables.
  We show that this leads to orders-of-magnitude speed-ups in both
  theory and practice by running experiments on proof logging and checking
  for SAT symmetry breaking using the state-of-the-art
  \satsuma symmetry breaker and the \veripb proof checking toolchain.
\end{abstract}

\section{Introduction}
\label{sec:intro}

An important challenge in combinatorial solving is to avoid repeatedly
exploring different parts of the search space that are equivalent
under symmetries. 
In a wide range of combinatorial solving paradigms,
\emph{symmetry breaking} is deployed
as a default technique, including
mixed integer programming~\cite{AW13MIP,BBBCetal24SCIP} and
constraint programming~\cite{Walsh06GeneralSymmetry}. %
The importance of symmetry breaking is
supported by both theoretical considerations~\cite{Urquhart99Symmetry}
and experimental results~%
\cite{PR19Symmetry}.
A detailed discussion of symmetry breaking is given, \egabbrev  by
\citet{Sakallah21SymmetrySatisfiabilityPlusCrossref}. 

Symmetry breaking has not been adopted as mainstream in
Boolean satisfiability (SAT) solving, however, 
despite a body of work
\cite{AMS03ShatterEfficientSymmetry,%
  DBBD16ImprovedStatic,%
  AndersBR24}
showing the potential for speed-ups also in this setting.
One reason for this could perhaps be the higher cost, relatively
speaking, of symmetry breaking compared to low-level SAT reasoning,
but state-of-the art symmetry detection is efficient enough to use
by default without degrading performance~%
\cite{AndersBR24}.
A more important concern is that the SAT community places a strong
emphasis on provable correctness. For over a decade, SAT solvers
taking part in the annual SAT competitions have had to generate
machine-verifiable proofs for their results.
Such proofs are especially important for sophisticated techniques such as
symmetry breaking, which is notoriously difficult to implement
correctly.
However, except for some special
cases~\cite{HHW15SymmetryBreaking},
it has not been known how to generate proofs for symmetry breaking in
the \drat proof format~%
\cite{WHH14DRAT}
used in the competitions, or whether this is
even possible.

The way symmetry breaking is typically done in SAT solving is by
introducing
\emph{lex-leader} constraints, which are 
encoded as the %
clauses
\begin{align}
  \label{eq:intro-lex-clauses}
  &
  \auxs_1 \lor \olnot{\varx}_1
  && \auxs_1 \lor \vary_1
  && \vary_1 \lor \olnot{\varx}_{1}
  \\
&  \auxs_{\idxi+1} \lor \olnot{\auxs}_\idxi \lor
  \olnot{\varx}_{\idxi+1}
 &&
    \auxs_{\idxi+1} \lor \olnot{\auxs}_\idxi \lor \vary_{\idxi+1}
  && \olnot{\auxs}_{\idxi} \lor \vary_{\idxi+1} \lor
     \olnot{\varx}_{\idxi+1}
     \nonumber
\end{align}
that can be thought of as encoding a circuit enforcing
$(\varx_1, \dots, \varx_n) \lexorder (\vary_1, \dots,
\vary_n)$---here,
$\auxs_i$ are fresh auxiliary variables encoding that the 
$\varx$- and
$\vary$-variables are equal up to position~$i$; using these, we enforce that
$x_i$~is false and
$y_i$~true the first time this does not hold.
Such clauses are clearly not implied by the original formula,
and the problem is how to prove that they can be added
without changing the satisfiability of the input.
Although the RAT rule~\cite{JHB12Inprocessing}
in \drat can handle a single symmetry~\cite{KT24Dominance},
once the first symmetry is broken it is not known how or even if the
other symmetries found by the symmetry breaker could be proven correct
using \drat.

\citet{BGMN23Dominance} finally resolved this
long-standing open
problem by introducing a
stronger proof
format, which operates with
\emph{pseudo-Boolean} (i.e., \zeroone linear)
inequalities rather than clauses, and reasons in terms of 
\emph{dominance}~\cite{CS15Dominance}
to support fully general symmetry breaking without any limitations on
the number of symmetries that can be handled.
One benefit of this richer format is that a single inequality
\begin{equation}
  \label{eq:intro-lex-pb}
  2^{n-1} x_1 + \dots
  + 2 x_{n-1}
  +  x_n
  \leq
  2^{n-1} y_1 + \dots
  + 2 y_{n-1}
  + y_n
  \eqcomma
\end{equation}
can be used in the proof to encode lexicographic order, and from this
constraint it is straightforward to derive the
clauses~\eqref{eq:intro-lex-clauses}
used by the solver.
However, at least $n^2$~bits are needed to represent the coefficients
in~\eqref{eq:intro-lex-pb}, while the representation
of~\eqref{eq:intro-lex-clauses}
scales linearly with~$n$.
This means that proof generation incurs a linear overhead
compared to solving.
Also, the algorithm by~\citet{AndersBR24} can break a symmetry in
quasi-linear time measured in the
number of variables~$k$ remapped by the symmetry,
which introduces yet another asymptotic slowdown
in proof generation if $k \ll n$.
Furthermore, the exponentially growing integer coefficients in~%
\refeq{eq:intro-lex-pb}
require expensive arbitrary-precision arithmetic, which
slows down proof checking.
All of these problems combine to make the proof logging approach
proposed by \citet{BGMN23Dominance}
infeasible for large-scale problems requiring non-trivial symmetry breaking.

In this work, we present an asymptotically faster method for
generating and checking proofs of correctness for symmetry breaking.
The main new technical idea is to use
\emph{auxiliary variables}
to encode the lexicographic order used for the dominance reasoning, 
similar to the clausal encoding in~\eqref{eq:intro-lex-clauses}.
Unfortunately,
this breaks the fundamental invariant of
\citet{BGMN23Dominance}
that all low-level proofs should be implicational.
When one needs to prove that a symmetry-breaking constraint respects
lexicographical order, the encoding of this order will contain
auxiliary variables that are not mentioned in the premises,
and so this property cannot possibly be implied.
We therefore need to make a substantial redesign of the proof system
of~\citet{BGMN23Dominance} to work with auxiliary variables. Very
briefly, our key technical twist is to split the encoding of the order
into two parts, putting one part into the premises, so that the
property of implicational low-level proofs can be maintained.
Our redesigned proof system supports fully general symmetry breaking
in a similar fashion to \citet{BGMN23Dominance},
but is significantly more efficient. 
Specifically, we prove that our approach leads to asymptotic gains for
proof logging and checking for symmetry breaking
by at least a linear factor in the size~$n$ of
the lexicographic order used.

We have implemented support for our new proof system in the
proof checker
\veripb~\cite{BGMN23Dominance,GN21CertifyingParity,Gocht22Thesis}
with its formally verified backend
\cakepb~\cite{GMMNOT24Subgraph}. 
Together, these yield an efficient, end-to-end verified proof
checking toolchain for symmetry breaking proofs.
We have also enhanced the 
state-of-the-art SAT symmetry breaker
\satsuma~\cite{AndersBR24} %
to generate proofs of correctness in our new format as well as that of
\citet{BGMN23Dominance}
for a comparative evaluation of performance.
Our experimental findings match our theoretical results and show that  
only a constant overhead in running time is required for proof
logging with our new method.
Proof checking performance is also vastly better
compared to \citet{BGMN23Dominance},
although here there might be room for further improvements.

Our paper is organized as follows.
After reviewing preliminaries in \refsec{sec:prelim},  
we present our new proof logging system in
\refsec{sec:extending-proof-system}.  
\reftwosecs{sec:satsumalogging}{sec:implementationdetails} 
discuss how proof logging and checking
can be improved asymptotically using our new method, 
which is confirmed by our
experiments
in
\refsec{sec:experiments}.
We conclude with a brief discussion of future work 
in \refsec{sec:conclusion}.
\ifthenelse{\boolean{extended-version}}{In this full-length version, we also provide five appendices. 
Appendix~\ref{app:cuttingplanes} gives an overview of the cutting planes proof system and the syntax used by \veripb for implicational reasoning. 
Appendix~\ref{app:proof_system_full} provides
further details on the extended proof system introduced in 
Sections~\ref{ssc:specifications} up to~\ref{ssc:redrule}, 
including all proofs. 
Appendix~\ref{app:satsumalogging} provides all details on how the proof logging is implemented in \satsuma.
Appendix~\ref{app:example} presents a concrete toy example of a \veripb proof generated by \satsuma.
Appendix~\ref{app:crafted} contains an overview of the crafted benchmarks that we used in our experimental evaluation.
}{
Further details, proofs, and a worked-out example 
\mbox{can be found in the full-length version.}
}

\section{Preliminaries}
\label{sec:prelim}

We start with a brief review of
pseudo-Boolean reasoning. For more details, we refer the reader to, \egabbrev \citet{BN21ProofCplxSATplusCrossref} or \citet{BGMN23Dominance}.
A \emph{Boolean variable} takes values $0$ or $1$.
A \emph{literal} over a Boolean variable $x$ is $x$ itself or its negation $\olnot{x} = 1 - x$.
A \emph{pseudo-Boolean (PB) constraint} $C$ is an integer linear inequality over literals 
\begin{align}
    C \synteq \sum\nolimits_i a_i \ell_i \geq A \eqcomma \label{eq:pb-constraint}
\end{align}
where we use $\synteq$ to denote syntactic equivalence. \Wolog the coefficients $a_i$ and the right-hand side $A$ are non-negative and the literals $\ell_i$ are over distinct variables.
The \emph{trivially false constraint} is $\bot \synteq 0 \geq 1$.
The \emph{negation} $\neg C$ of the pseudo-Boolean constraint $C$ in~\eqref{eq:pb-constraint} is the pseudo-Boolean constraint $\neg C \synteq \sum_i a_i \olnot{\ell}_i \geq \sum_i a_i - A + 1$.
A \emph{pseudo-Boolean formula} $\formf$ is a conjunction $\formf \synteq \bigwedge_i C_i$ or equivalently a set $\formf \synteq \bigcup_i \set{C_i}$ of pseudo-Boolean constraints $C_i$, whichever view is more convenient.
A \emph{(disjunctive) clause} $\bigvee_i \ell_i$ is equivalent to the pseudo-Boolean constraint \mbox{$\sum_i \ell_i \geq 1$}. 
Hence, formulas in \emph{conjunctive normal form (CNF)} are special cases of pseudo-Boolean formulas.

An \emph{assignment} is a function mapping from Boolean variables to $\set{0, 1}$.
\emph{Substitutions} 
(or \emph{witnesses})
generalize assignments by allowing variables to be mapped to literals, too.
A substitution $\omega$ is extended to literals by $\omega(\overline{x}) = \overline{\omega(x)}$, and to preserve truth values, \ie $\omega(0) = 0$ and $\omega(1) = 1$.
For a substitution $\omega$, the support $\supp(\omega)$ is the set of variables $x$ where $\omega(x) \neq x$.
A substitution $\assmnta$ can be composed with another substitution $\omega$ by applying $\omega$ first and then $\assmnta$, \ie 
$(\assmnta \circ \omega)(x) = \assmnta(\omega(x))$.
Applying a substitution $\omega$ to the pseudo-Boolean constraint $C$ in~\eqref{eq:pb-constraint} yields the pseudo-Boolean constraint $C\uh_\omega \synteq \sum_i a_i \omega(\ell_i) \geq A$. This is extended to formulas by defining 
$\formf\uh_\omega \synteq \bigwedge_i C_i\uh_\omega$.
The pseudo-Boolean constraint $C$ is satisfied by an assignment $\omega$ if 
\mbox{$\sum_{i:\omega(\ell_i) = 1} a_i \geq A$}. 
A pseudo-Boolean formula $\formf$ is satisfied by $\omega$ if $\omega$ satisfies every constraint in $\formf$.
If there is no assignment that satisfies $\formf$, then $\formf$ is \emph{unsatisfiable}.

We use the notation $\formf(\vec{x})$ to stress that the formula is defined over the list of variables $\vec{x} = x_1, \dots , x_n$, where we syntactically highlight a partitioning of the list of variables by writing $\formf(\vec{y}, \vec{z})$ or $\formf(\vec{a}, \vec{b}, \vec{c})$ %
meaning $\vec{x} = \vec{y}, \vec{z}$ or $\vec{x} = \vec{a}, \vec{b}, \vec{c}$, respectively (denoting concatenation of the lists of variables).
To apply a substitution $\omega$ element-wise to a list of literals we write 
$\vec{\ell}\uh_\omega = \omega(\ell_1), \dots , \omega(\ell_n)$.
For a formula $\formf(\vec{x})$ and a list of literals and truth values $\vec{y} = y_1, \dots, y_n$, the notation $\formf(\vec{y})$ is syntactic sugar for $\formf\uh_\omega$ with the implicitly defined substitution $\omega(x_i) = y_i$ for $i = 1, \dots , n$. Finally, we write $\var(\formf)$ for the set of variables in a formula $\formf$.

\subsection{The \veripb Proof System}
\label{ssc:original}

The proof system introduced by~\citet{BGMN23Dominance}
(which we will refer to as the \emph{\original system})
can prove optimal values
for \emph{optimization problems} $(\formf, f)$, where $\formf$ is a pseudo-Boolean formula,
and $f$ is an integer linear 
\emph{objective} 
function 
over literals to be minimized subject to satisfying $\formf$.
The \emph{satisfiability (SAT) problem} is a special case by having 
$f = 0$ and $\formf$ being a CNF formula. 
Proving the unsatisfiability of~$\formf$ then
corresponds to proving that $\infty$ is a lower bound for~$(\formf, f)$.
For clarity of exposition, we focus on decision problems, \ie problems with 
objective function
$f = 0$, but the results can
easily 
be extended to optimization problems as in~\citet{BGMN23Dominance}.

A proof in this proof system consists of a sequence of rule applications, each deriving a new constraint. For implicational reasoning, 
the \emph{cutting planes} proof system~\cite{CCT87ComplexityCP}
is used, 
which provides sound reasoning rules to derive pseudo-Boolean constraints implied by a pseudo-Boolean formula $\formf$, \egabbrev taking positive integer linear combinations or dividing by an integer and rounding up.
We write $\formf \derives \constrc$ if there is a cutting planes proof deriving $\constrc$ from $\formf$.
A set of constraints $\formf^\prime$ is derivable from another set $\formf$, denoted by $\formf \derives \formf^\prime$, if $\formf \derives \constrc$ for all $\constrc \in \formf^\prime$. 

The proof system also 
has
rules for deriving constraints which are not implied.
To do this, 
the \original system keeps track of two pseudo-Boolean formulas (\ie sets of constraints),
called \emph{core} $\core$ and \emph{derived} $\der$,
which \citet{JHB12Inprocessing} call irredundant and redundant clauses, respectively. 
In addition, we need a pseudo-Boolean formula $\ord(\vu, \vv)$,
encoding a preorder,
\ie a reflexive and transitive relation. %
This preorder is used to compare assignments $\assmnta, \assmntb$ over the literals in 
a
list~$\vz$
and we write  $\assmnta \preceq \assmntb$ if 
$\ord(\vz\uh_{\assmnta}, \vz\uh_{\assmntb})$
evaluates to true.
For a preorder $\preceq$, we define the strict order $\prec$ such that $\assmnta \prec \assmntb$ holds if $\assmnta \preceq \assmntb$ and $\assmntb \not\preceq \assmnta$.

We call the tuple $(\core, \der, \ord, \vz)$ a \emph{configuration}.
Formally, proof rules incrementally modify the configuration. 
To handle optimization problems, the configuration of \citet{BGMN23Dominance} also contains the current upper bound on $f$, which we can omit for decision problems. %

The proof system maintains two invariants: (1) $\core$ is satisfiable if $\formf$ is satisfiable, and (2) for any assignments $\assmnta$ satisfying $\core$
there exists an assignment $\assmntaprime$ satisfying $\core$, $\der$, and 
$\assmntaprime \preceq \assmnta$.
Starting with the configuration $(\formf, \emptyset, \emptyset, \emptyset)$,
any valid derivation of a configuration $(\core, \der, \ord, \vz)$ with $\bot \in \core \union \der$ proves that $\formf$ is unsatisfiable.

\paragraph{Proof Rules.}
We list the satisfiability version of the proof rules from~\citet{BGMN23Dominance} our work modifies;
all other rules in the \original proof system remain unchanged.
\begin{itemize}

\item The \emph{redundance-based strengthening rule} (or \emph{redundance rule} for short)
allows transitioning from $(\core, \der, \ord, \vz)$ 
to $(\core, \der \cup \{\constrc\}, \ord, \vz)$
if a substitution $\witness$ and cutting planes proofs are provided showing that
\begin{equation}
    \core \cup \der \cup \{\neg \constrc\} \derives
    (\core \union \der \union \{ \constrc \})\uh_{\witness}
    \cup \ord(\vz\uh_{\witness}, \vz) \eqperiod
    \label{eq:prelims_red}
\end{equation}

\item The \emph{dominance-based strengthening rule} (or \emph{dominance rule} for short)
allows transitioning from $(\core, \der, \ord, \vz)$ 
to $(\core, \der \cup \{\constrc\}, \ord, \vz)$
if a substitution $\witness$ and cutting planes proofs are provided showing that
\begin{align}
    \core \cup \der \cup \{\neg \constrc\} & \derives \core\uh_{\witness} \cup \ord(\vz\uh_{\witness}, \vz) \label{eq:prelims_dom1} \\
    \core \cup \der \cup \{\neg \constrc\}
    {} \cup \ord(\vz, \vz\uh_{\omega}) & \derives \bot \eqperiod \label{eq:prelims_dom2}
\end{align}
\end{itemize}

We now briefly explain why the redundance rule preserves the second invariant. 
Let $\assmnta$ be an assignment satisfying $\core$. 
Since the invariant holds for $(\core, \der, \ord, \vec{z})$, 
there exists an assignment $\assmntaprime$ satisfying $\core \cup \der$ 
and %
$\assmntaprime \preceq \assmnta$. 
If $\assmntaprime$ happens to 
satisfy $\constrc$, we are done. 
Otherwise, the derivation~\eqref{eq:prelims_red} guarantees that 
$\assmntaprime \circ \omega$ satisfies $\core \cup \der \cup \{\constrc\}$ 
and $\ord(\vec{z}\uh_{\assmntaprime \circ \omega}, \vec{z}\uh_{\assmntaprime})$, 
\ie $\assmntaprime \circ \omega \preceq \assmntaprime$.
By transitivity we get $\assmntaprime \circ \omega \preceq \assmnta$. 
For the dominance rule,
$\assmntaprime$ might have to be composed with $\omega$ repeatedly, but the process is guaranteed to eventually satisfy $C$, since
the composed assignment strictly decreases with respect to the order, which is encoded by~\eqref{eq:prelims_dom2}.

\paragraph{Preorders.} Before using a preorder $\ord$, it needs to be proven within the proof system that $\ord$ is indeed reflexive and transitive. For this, the \original system requires cutting planes proofs for
$\emptyset \derives \ord(\vec{u}, \vec{u})$ and $\ord(\vec{u}, \vec{v}) \union \ord(\vec{v}, \vec{w}) \derives \ord(\vec{u}, \vec{w})$,
where $\vec{w}$ is of the same size as $\vec{u}$ and $\vec{v}$.

\subsection{Symmetry Breaking}
\label{ssc:symbreak}

We 
briefly review symmetry breaking as it is used in practice, which we want to certify.
Typically, symmetry breaking considers
permutations $\sigma$ between literals with $\sigma(\overline{\ell}) = \overline{\sigma(\ell)}$
for all literals $\ell$, and finite support $\supp(\sigma)$.
Practical symmetry breaking algorithms only detect
\emph{syntactic} symmetries of a formula $\formf$, \ie permutations $\sigma$
with $\formf\uh_{\sigma} \synteq \formf$.

To encode an ordering of assignments, 
typically 
the \emph{lex-leader} constraint is used.
Let $S$ be a set of detected symmetries in a 
formula $F$, and $z_1,\dots,z_n$ be variables with
$\suppsymmetry \subseteq \{z_1,\dots,z_n\}$ for all 
$\sigma \in S$.
Then we 
define the 
\emph{lexicographic order} 
$\lexorder$ over assignments $\assmnta,\assmntb$
and the 
\emph{lex-leader constraint} 
$B_\symmetry$ for a symmetry $\symmetry \in S$ as
\begin{align}
\label{eqn:lexleq}
\alpha \lexorder \beta & \text{ iff } \sum\nolimits_{i=1}^n 2^{n-i} \alpha(z_i) \leq \sum\nolimits_{i=1}^n 2^{n-i} \beta(z_i) \\
\label{eqn:symbreak}
B_\sigma & {} \synteq \sum\nolimits_{i=1}^n 2^{n-i} (\symmetry(x_i) - x_i) \geq 0 \eqperiod
\end{align}
Intuitively, \eqref{eqn:symbreak} constrains assignments to be smaller w.r.t.\ the preorder in \eqref{eqn:lexleq}
than their symmetric counterpart.
Symmetry breaking introduces the constraints $B_\symmetry$ for $\symmetry \in S$
such that if $\formf$ is satisfiable, then $\formf \union \bigcup_{\symmetry \in S} B_\symmetry$ is satisfiable.

When doing proof logging for symmetry breaking, 
the dominance rule in the \original system can derive
the lex-leader constraints $B_\symmetry$ as follows. 
Suppose that the symmetry breaker detect symmetries 
$\symmetry_1,\dots,\symmetry_m$ of $\core$.
To log these symmetries, we use the order defined by
\begin{equation}
\label{eqn:lexord}
\lexord(\vx, \vy) = \left\{ \sum\nolimits_{i=1}^n 2^{n - i} (y_i - x_i) \geq 0 \right\} \eqperiod
\end{equation}
To add a constraint $B_{\symmetry_i}$ to $\der$, we use the dominance rule with witness $\symmetry_i$. The application of this rule is justified by $\core \derives \core\uh_{\sigma_i}$ (trivial, because $\sigma_i$ is a symmetry of $\core$), and the fact that $\neg B_{\sigma_i}$ implies both $\lexord(\vec{z}\uh_{\sigma_i}, \vec{z})$ and $\neg \lexord(\vec{z}, \vec{z}\uh_{\sigma_i})$.

When using symmetry breaking for the SAT problem, the symmetry breaker instead encodes $\vx \lexorder \symmetry(\vx)$ as the clauses
\begin{subequations}
\begin{align}
\auxs_1 +   \olnot{\varx_1}  \geq 1\eqcomma\quad
\auxs_{i+1} + \olnot{\auxs_i} + \olnot{\varx_{i+1}} \geq 1\eqcomma \label{eq:symbreak_s1} \\
\auxs_1 + \symmetry(\varx_1) \geq 1\eqcomma\quad
\auxs_{i+1} + \olnot{\auxs_i} + \symmetry(\varx_{i+1}) \geq 1\eqcomma 
\label{eq:symbreak_s2} \\
\symmetry(\varx_1) + \olnot{\varx_1} \geq 1 \eqcomma\quad
\olnot{\auxs_i} + \symmetry(\varx_{i+1}) + \olnot{\varx_{i+1}}\geq 1\eqcomma  \label{eq:symbreak_t}
\end{align}
\end{subequations}
where $\auxs_{i}$ encodes that $(\varx_1, \ldots, \varx_i) = (\symmetry(\varx_1), \ldots, \symmetry(\varx_i)) $. These clauses can be derived from~\eqref{eqn:symbreak} using redundance.

\section{Strengthening with Auxiliary Variables}
\label{sec:extending-proof-system}

While the method presented in \refsec{ssc:symbreak} enables proof logging for symmetry breaking,
encoding the coefficients in~\eqref{eqn:symbreak} and~\eqref{eqn:lexord} grows quadratically in the size of $\vz$, which often includes all variables in the formula, making proof logging for large symmetries infeasible in practice.
For proof checking, the situation is even more dire, as the proof checker has to reason internally with arbitrary-precision integer arithmetic to handle the coefficients in~\eqref{eqn:symbreak} and~\eqref{eqn:lexord}.

One way to avoid these big integers would be to represent the order as a set of clauses as in Equation \eqref{eq:symbreak_s1}--\eqref{eq:symbreak_t}, using a list of extension variables $\vecs$. However, this leads to challenges when defining the actual preorder $\preceq$. 
For an order without extension variables, we define $\assmnta \preceq \assmntb$ to hold if $\ord(\vz \uh_\assmnta, \vz\uh_\assmntb)$ is true.  However, for a formula $\ord(\vx, \vy, \vecs)$ containing extension variables~$\vecs$ this does not work, since
the variables~$\vecs$
in $\ord(\vz \uh_\assmnta, \vz\uh_\assmntb, \vecs)$ are unassigned and in general $\ord(\vz \uh_\assmnta, \vz\uh_\assmntb, \vecs)$ will not hold for all assignments to $\vecs$. 

Instead, what we are trying to capture is that $\assmnta \preceq \assmntb$ holds precisely when $\ord(\vz \uh_\assmnta, \vz\uh_\assmntb, \vecs)$ holds, \emph{provided that} the extension variables $\vecs$ are set in the right way. Equivalently, we want to say that $\assmnta \preceq \assmntb$  holds precisely when there exists an assignment $\assmntp$ to $\vecs$ such that $\ord(\vz \uh_\assmnta, \vz\uh_\assmntb, \vecs\uh_\assmntp)$ holds.

However, just adding extension variables to the proof obligations $\ord(\vz\uh_\omega,\vz)$ in the redundance and dominance rules would not work. 
What we would need to show is that \emph{some} assignment to $\vecs$ exists such that $\ord(\vz\uh_\witness,\vz,\vecs)$ holds, 
but the proof system cannot express existential quantification.  
While the proof rules could specify the value of all extension variables $\vecs$, this would be very cumbersome. However, in all 
applications
we have in mind, the preorder with extension variables already contains 
the information how to set the extension variables $\vecs$,
since the extension variables are \emph{defined} (functionally) in terms of the other variables. 

To make this precise, let $\spec(\vx, \vy, \vecs)$ be a 
\emph{definition} of $\vecs$ in terms of the other variables (\ie each assignment to the $\vx$ and $\vy$ can uniquely be extended to an assignment to $\vecs$ that satisfies $\spec(\vx, \vy, \vecs)$). 
We now redefine $\preceq$ such that $\assmnta \preceq \assmntb$ holds precisely when there exists an assignment $\assmntp$ to $\vecs$ such that $\spec(\vz \uh_\assmnta, \vz\uh_\assmntb, \vecs\uh_\assmntp)  \wedge \ord(\vz \uh_\assmnta, \vz\uh_\assmntb, \vecs\uh_\assmntp)$ holds. 

 In this case, whenever we need to show that $\assmnta\preceq\assmntb$, because of the definitional nature of $\spec$ we can \emph{assume} that $ \spec(\vz \uh_\assmnta, \vz\uh_\assmntb, \vecs\uh_\assmntp)$ holds and 
 derive $ \ord(\vz \uh_\assmnta, \vz\uh_\assmntb, \vecs\uh_\assmntp)$ from this, thereby completely eliminating the need for 
providing an assignment to $\vecs$
in every rule application. 
 In our actual proof system, 
 we %
 relax the condition that $\spec$ %
 is 
 definitional slightly, but 
 intuitively, 
 $\spec$ is best thought of as a circuit defining the value of $\vecs$ in terms of the other variables. 
 
 An important restriction for this to be sound is that the extension variables $\vecs$ in the preorder, which we call \emph{auxiliary variables}, do not appear outside the preorder.

We now formalize this. As mentioned in Section~\ref{ssc:original}, we focus on decision
problems%
\ifthenelse{\boolean{extended-version}}{%
, and refer to Appendix~\ref{app:proof_system_full} for the extension to optimization 
problems and proofs of all results. 
}{%
.
}

\subsection{Specifications}
\label{ssc:specifications}

Let $\varset$ be a \listt of variables. 
A pseudo-Boolean formula $\specification(\vx, \varset)$ is a \emph{specification over the variables $\varset$}, if it is derivable 
from the empty formula $\emptyset$ by the redundance rule, where each application only witnesses 
over variables in~$\varset$.

\begin{definition}
A formula $\specification(\vx, \varset) = \set{\constrc_1, \constrc_2, \dots{}, \constrc_n}$ is a \emph{specification over the variables $\varset$}, if there is a list 
\begin{equation*}
(\constrc_1, \witness_1), (\constrc_2, \witness_2), \dots{}, (\constrc_n, \witness_n)
\end{equation*}
which satisfies the following: 
\begin{enumerate}
    \item The constraint $\constrc_1$ can be obtained from the empty formula $\emptyset$ using the redundance rule with witness $\witness_1$.
    \item For each $i \in \{2, \dots{}, n\}$ we have that 
$C_i$ can be added by the redundance rule to
$\bigcup_{j=1}^{i-1} \{\constrc_j\} $ with the witness $\witness_i$. 
    \item For every witness $\witness_i$, 
    $\supp(\witness_i) \subseteq \varset$ holds.  
\end{enumerate}
\end{definition}

A crucial property of specifications is that we can recover an assignment of the auxiliary variables from 
the assignment of the non-auxiliary variables.
We state this property below.
\begin{lemma} Let $\specification(\vx, \va)$ 
    be a specification over $\va$.
    Let $\assmntone$ be any assignment of the variables $\vx$. 
    Then, $\assmntone$ can be extended to an assignment $\assmnttwo$, such that
    \begin{enumerate} 
        \item $\assmnttwo$ satisfies $\specification$, and 
        \item $\assmntone(x) = \assmnttwo(x)$ holds for every $x \in \vx$. 
    \end{enumerate}
    \label{lem:spec_behaves_nicely}
\end{lemma}

To explain why Lemma~\ref{lem:spec_behaves_nicely} holds, recall from Section~\ref{ssc:original} that the redundance rule %
satisfies the following: if a constraint $\constrc$ is added by redundance with witness $\witness$, then given an assignment $\assmntone$ satisfying all other constraints, either  $\assmntone$ or  $\assmntone \circ \witness$ also satisfies $\constrc$. Hence, defining $\assmnttwo$ by composing $\assmntone$ with the witnesses~$\witness_i$ corresponding to the constraints 
that do not already hold
ensures that $\assmnttwo$ satisfies $\specification$, and the fact that each witness~$\witness_i$ is the identity on $\vx$ (since $\supp(\witness_i) \subseteq \varset$) ensures that $\assmntone(x) = \assmnttwo(x)$ for $x \in \vx$.

\subsection{Orders with Auxiliary Variables} \label{subsec:orderswithauxvars}
Next, we explain how two pseudo-Boolean formulas $\ord(\vu, \vv, \va)$ and $\spec(\vu, \vv, \va)$, together with 
the two disjoint \lists of variables $\vz$ and $\va$, define a preorder.

\begin{definition} \label{def:induced-preorder}
Let $\vu$ and $\vv$ be disjoint \lists of variables of size~$\nvar$ and let 
$\va$ be a \listt of auxiliary variables. 
Let $\ord(\vu, \vv, \va)$ and $\spec(\vu, \vv, \va)$ be two pseudo-Boolean formulas such that $\spec$ is a specification over $\va$. 

Then we define the relation $\preceq$  over the domain of total assignments to a 
\listt of variables $\vz$ of size~$\nvar$ as follows:
For assignments $\assmnta, \assmntb$ we let $\assmnta \preceq \assmntb$ hold, 
if and only if there exists an assignment $\assmntp$ to the variables $\va$, such that 
\begin{equation*}
\spec(\vz\uh_{\assmnta},\vz\uh_{\assmntb}, \va\uh_{\assmntp}) \land
\ord(\vz\uh_{\assmnta},\vz\uh_{\assmntb}, \va\uh_{\assmntp})
\end{equation*}
evaluates to true.
\end{definition}

To ensure that $\ord$ and $\spec$ 
actually define a preorder,
we require 
cutting planes proofs
that
show reflexivity, \ie 
$\emptyset \derives \assmnta\preceq\assmnta$, 
and transitivity, \ie 
$\assmnta\preceq\assmntb \land 
\assmntb\preceq\assmntc\derives 
\assmnta\preceq \assmntc$. 
To write these proof obligations using the cutting planes proof system, 
which cannot handle an existentially quantified conclusion,
we can use the specification as a premise. 
The specification premise essentially tells us which auxiliary variables 
the existential quantifier should pick. 
In particular, for reflexivity, the proof obligation is 
\begin{equation}
\spec(\vx, \vx, \va) \vdash{} \ord(\vx, \vx, \va) \eqperiod
  \label{eq:preorder1}
\end{equation}
For transitivity, the proof obligation is
\begin{equation}
\begin{split}
\spec(\vx, \vy, \va) &\cup \ord(\vx, \vy, \va) \cup \spec(\vy, \vz, \vb) \\ &\cup 
  \ord(\vy, \vz, \vb) \cup \spec(\vx, \vz, \vc)   \vdash{} \ord(\vx, \vz, \vc) \eqperiod
  \label{eq:preorder2}
\end{split}
\end{equation}

Intuitively, \eqref{eq:preorder2} says that if the circuits defining the auxiliary variables are correctly evaluated, which is encoded by the premises $\spec(\vx, \vy, \va) \cup \spec(\vy, \vz, \vb) \cup \spec(\vx, \vz, \vc)$, then transitivity should hold, \ie 
$\ord(\vx, \vy, \va) \cup 
  \ord(\vy, \vz, \vb)  \vdash{} \ord(\vx, \vz, \vc)$. 
However, if the auxiliary variables are not correctly set, then no claims are made.

These proof obligations 
ensure that $\preceq$ is a preorder:

\begin{lemma} If $\ord$ and $\spec$ satisfy Equations~\eqref{eq:preorder1} and \eqref{eq:preorder2}, then 
    $\preceq$ as defined by $\ord$ and $\spec$ is a preorder.
\end{lemma}

\subsection{Validity} \label{sec:validity}
We extend the configurations of the proof system to $\confdec$.
In particular, we extend the notion of \emph{weak-$(\formf, f)$-validity} 
from \citet{BGMN23Dominance} to our new configurations, focusing on decision problems:

\begin{definition} \label{def:validity}
  A configuration $\confdec{}$ is \emph{weakly $\formf$-valid} if the following conditions hold:
  \begin{enumerate}
    \item If $\formf$ is satisfiable, then $\core$ is satisfiable.
    \item For every assignment~$\assmnta$ satisfying $\core$, there exists an assignment $\assmntaprime$ satisfying $\core \cup \der$ and $\assmntaprime \preceq \assmnta$.
  \end{enumerate}
  \label{def:valid}
\end{definition}
In the following, we furthermore assume that
for any configuration $\confdec{}$, the following hold:
\begin{enumerate}
\item $\ord$ and $\spec$ refer to formulas for which Equations~\eqref{eq:preorder1} and~\eqref{eq:preorder2} have been successfully proven.
\item $\spec$ is a specification over $\va$.
\item The variables $\va$ only occur in $\ord$ and $\spec$, 
      and these variables are disjoint %
      from $\vz$. 
\end{enumerate}

Observe that due to these invariants, satisfying assignments for $\core \cup \der$ 
do not need to assign the 
variables $\va$. 
In the following, we assume that such assignments are indeed defined only over the domain 
$\var(\core \cup \der)$.

\subsection{Dominance-Based Strengthening Rule}

As in the \original system, the dominance rule allows adding a constraint $\constrc$ to the derived set  using witness $\witness$ if from the premises $\core \cup \der \cup \{\neg \constrc\}$ we can derive $\core \uh_\witness$ and show that $\assmnta \circ \witness \prec \assmnta$ holds for all assignments $\assmnta$ satisfying $\core \cup \der \cup \{\neg \constrc\}$. To show  $\assmnta  \circ \witness \prec \assmnta$, we separately show that  $\assmnta  \circ \witness \preceq \assmnta$ and $\assmnta \not\preceq \assmnta \circ \witness$. To show that  $\assmnta  \circ \witness \preceq \assmnta$, we have to show that $\ord(\vz\uh_\witness, \vz, \va)$, assuming that the circuit defining the auxiliary variables $\va$ has been evaluated correctly, which is encoded by the specification $\spec(\vz\uh_\witness, \vz, \va)$.  This leads to the proof obligation
\begin{equation}
\core \cup \der \cup \{\neg \constrc\} \cup \spec(\vz\uh_\witness, \vz, \va) 
        \vdash \ord(\vz\uh_\witness, \vz, \va)\eqperiod
\end{equation}
To show that $\assmnta \not\preceq \assmnta \circ \witness$, we have to show that $\neg\ord(\vz, \vz\uh_\witness, \va)$ assuming $ \spec(\vz, \vz\uh_\witness,\va)$. However, since  $\neg\ord(\vz, \vz\uh_\witness, \va)$ is not necessarily a pseudo-Boolean formula (due to the negation), we instead show that we can derive contradiction from $\ord(\vz, \vz\uh_\witness, \va)$, leading to the proof obligation:
\begin{equation}
\core \cup \der \cup \{\neg \constrc\} \cup \spec(\vz, \vz\uh_\witness,\va) \cup \ord(\vz, \vz\uh_\witness, \va) \vdash \bot \eqperiod
\end{equation}
The following lemma shows that these proof obligations indeed imply $\assmnta \circ\witness \preceq \assmnta$ and $\assmnta \not\preceq \assmnta\circ\witness$, respectively: 
 
\begin{lemma} 
Let $\formg$ be a formula and $\witness$ a witness with 
$\supp(\witness) \subseteq \var(\formg)$. 
Furthermore, let $\va \cap \var(\formg) = \emptyset$ and $\spec$ be a specification over $\va$.
              Also, let $\ord$ and $\spec$ define a pre-order $\preceq$. Then the following hold:
\begin{enumerate}
    \item If $\formg \cup \spec(\vz\uh_\witness, \vz, \va) 
    \vdash \ord(\vz\uh_\witness, \vz, \va)$
    holds, then for each assignment $\assmnta$ satisfying $\formg$, 
    $\assmnta \circ\witness \preceq \assmnta$ holds.
    \item If $\formg \cup \spec(\vz, \vz\uh_\witness, \va) 
    \cup \ord(\vz, \vz\uh_\witness, \va) \vdash \bot$
    holds, then for each assignment~$\assmnta$ satisfying~$\formg$, 
    $\assmnta \not\preceq \assmnta\circ\witness$ holds. 
\end{enumerate} 
\label{lem:order_reasonable_both}
\end{lemma}

Hence, we define the dominance rule as follows:

\begin{definition}[Dominance-based strengthening with specification]
    We can transition from the configuration $\confdec$ to $(\core, \der \cup \{\constrc\}, \ord, \spec, \vz, \va)$ using the dominance rule 
    if the following conditions are met: 
    \begin{enumerate}
        \item The constraint $\constrc$ does not contain variables in $\va$.
        \item There is a witness $\witness$ for which $\img(\witness) \cap \va = \emptyset$ holds.
        \item We have cutting planes proofs for the following:
\end{enumerate}
\begin{align}
\core &\cup \der \cup \{\neg \constrc\} \cup \spec(\vz\uh_\witness, \vz, \va) 
        \vdash \core\uh_\witness \!\cup \ord(\vz\uh_\witness, \vz, \va) \label{eq:dom1}\\\core &\cup \der \cup \{\neg \constrc\} \cup \spec(\vz, \vz\uh_\witness,\va) \cup \ord(\vz, \vz\uh_\witness, \va) 
        \vdash \bot \eqperiod\label{eq:dom2}
\end{align}
\end{definition}

Using Lemma~\ref{lem:order_reasonable_both}, we can show that the dominance rule preserves the invariants required by weak $\formf$-validity:

\begin{lemma}
    If we can transition from $\confdec$ to $(\core, \der \cup \{\constrc\}, \ord, \spec, \vz, \va)$ 
    by the dominance rule, and $\confdec$ is weakly $\formf$-valid, 
    then the configuration~$(\core, \der \cup \{\constrc\}, \ord, \spec, \vz, \va)$ is also weakly $\formf$-valid.
\end{lemma}

\subsection{Redundance-Based Strengthening Rule}
\label{ssc:redrule}
We also modify the redundance rule to work in our \extended proof system.
Similarly to the dominance rule, we can use $\spec(\vz\uh_\witness, \vz, \va)$ as an 
extra
premise in our proof obligations. 

\begin{definition}[Redundance-based strengthening with specification]
    We can transition from the configuration $\confdec$ to $(\core, \der \cup \{\constrc\}, \ord, \spec, \vz, \va)$ using 
    the redundance rule if the following conditions are met: 
    \begin{enumerate}
        \item The constraint $\constrc$ does not contain variables in $\va$.
        \item There is a witness $\witness$ for which $\img(\witness) \cap \va = \emptyset$ holds.
        \item We have cutting planes proof that the following holds:
     \begin{equation}
     \begin{split}
        \core &\cup \der \cup \{\neg \constrc\} \cup \spec(\vz\uh_\witness, \vz, \va)  \\
         &\vdash (\core \union \der \union \{ \constrc \})\uh_{\witness} \cup \ord(\vz\uh_\witness, \vz, \va) \eqperiod \label{eq:red}
    \end{split}
    \end{equation}
\end{enumerate}
\end{definition}

The redundance rule preserves weak $\formf$-validity:

\begin{lemma}
    \label{lem:red_preserves_validity}
    If we can transition from $\confdec$ to 
    $(\core, \der \cup \{\constrc\}, \ord, \spec, \vz, \va)$ 
    by the redundance rule, and $\confdec$ is weakly $\formf$-valid, 
    then the configuration~$(\core, \der \cup \{\constrc\}, \ord, \spec, \vz, \va)$ 
    is also weakly $\formf$-valid.
\end{lemma}

\section{Efficient Proof Logging in \satsuma} \label{sec:satsumalogging}

Using our \extended proof system, we implement proof logging in the state-of-the-art symmetry breaker \satsuma\footnote{\satsuma code: \url{https://doi.org/10.5281/zenodo.17607863}}. 
\ifthenelse{\boolean{extended-version}}{
    The details can be found in Appendix~\ref{app:satsumalogging} and a worked-out example is given in Appendix~\ref{app:example}.
}{}
Like in the \original system (as explained in Section~\ref{ssc:symbreak}), the %
(negation of the)  %
symmetry breaking constraints can be used to show the proof obligations for the order. However, in the \extended system this is more complicated, because we need to relate two different sets of extension variables (those in the symmetry breaking constraints and the auxiliary variables in the specification).

Our \new method achieves an asymptotic speedup over the \old method. Defining 
the lexicographical order over $\nvar$ variables can be done in time $\bigoh{\nvar}$ with our \new method (for both checking and logging), while the \old method requires time $\bigoh{\nvar^2}$. 
Breaking
a symmetry $\symmetry$ over  $\suppsize = |\suppsymmetry|$ variables ($\suppsize \leq \nvar$) takes time $\bigoh{\suppsize}$ for logging and $\bigoh{\nvar}$ for checking with our \new method, while the \old method requires time  $\bigoh{\nvar\suppsize}$ for logging and time $\bigoh{\nvar^2+\nvar\suppsize^2}$ for checking. Therefore, the \new method is in each case asymptotically at least a factor $\nvar$
faster
for both logging and checking than the \old method used by \citet{BGMN23Dominance}.

\section{Proof Checker Implementation}
\label{sec:implementationdetails}

We implemented checking for our \extended proof system in the proof checker \veripb\footnote{\veripb code: \url{https://doi.org/10.5281/zenodo.17608873}} and the formally verified proof checker \cakepb\footnote{\cakepb code: \url{https://doi.org/10.5281/zenodo.17609070}}.
Several optimizations are necessary to handle orders with many specification constraints efficiently.

\paragraph{Lazy Constraint Loading and Evaluation.} When checking cutting planes derivations for the dominance or redundance rule (\ref{eq:dom1})--(\ref{eq:red}), the proof checkers load the specification constraints from $\spec$ only when they are used in the proof.
More specifically, the constraints in $\spec$ are not even computed until loaded explicitly, which improves the checking performance by a linear factor if the specification $\spec$ is not required for a cutting planes derivation.

\paragraph{Implicit Reflexivity Proof.} Since the loaded order is always proven to be reflexive %
\eqref{eq:preorder1}, %
the cutting planes derivation for $\ord(\vec{\varz}\uh_\omega, \vec{\varz}, \va)$ can be skipped for the redundance
rule
\eqref{eq:red} if the domain of the witness $\omega$ does not contain a variable in $\vec{z}$.
Requiring an explicit cutting planes derivation for $\ord(\vec{\varz}\uh_\omega, \vec{\varz}, \va)$ would again involve computation over all constraints in the specification $\spec$, which incurs a linear overhead for proof logging and checking.

\paragraph{Formal Verification.} We updated \cakepb in two phases, yielding the same end-to-end verification guarantees for proof checking as discussed in~\citet{GMMNOT24Subgraph}.
First, we formally verified soundness of all updates to the proof system (including Lemmas~\ref{lem:spec_behaves_nicely}--\ref{lem:red_preserves_validity}).
Second, we implemented and verified these changes in the \cakepb codebase, including soundness for the optimizations described above.

\section{Experimental Evaluation}
\label{sec:experiments}

\begin{figure*}[t]
    \begin{center}
    \includegraphics*{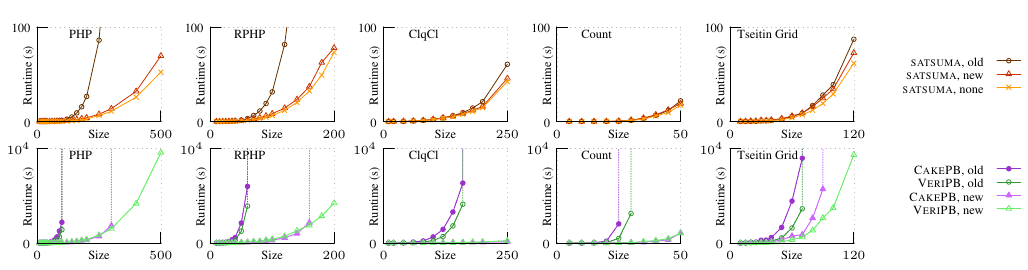}
    \end{center}
    \caption{On top, the cost of running \satsuma with or without proof logging, on crafted benchmark instances, as the instance
    size grows. In each case logging with the \new method scales similarly to not doing logging, whilst the \old method exhibits
    worse scaling for several families. On the bottom, the cost of checking these proofs: the \new method exhibits
    better scaling.}
    \label{figure:crafted-scalability}
\end{figure*}

\begin{figure*}[t]
    \begin{center}
        \includegraphics*{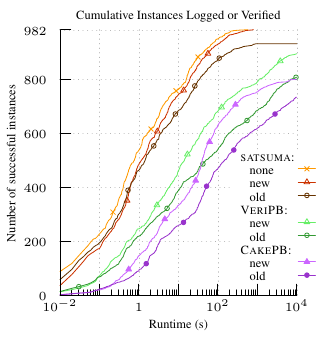}
        \includegraphics*{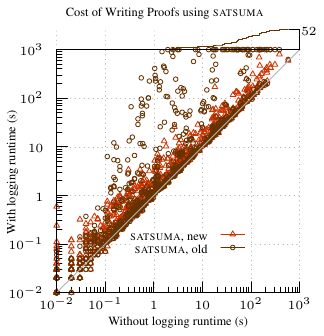}
        \includegraphics*{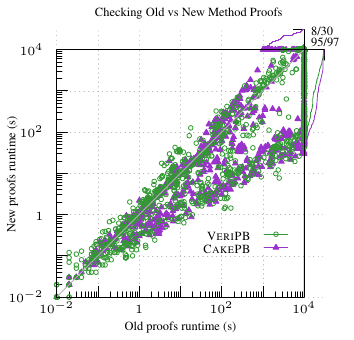}
    \end{center}
    \caption{On the left, the cumulative number of ``interesting'' SAT
    competition instances broken, logged, and checked over time. The centre
    plot compares the added cost of writing proofs with the \old and \new
    methods, compared to not writing proofs; 52 instances reached time or
    memory limits with the \old method. The
    right-hand plot compares the cost of checking proofs using the \old and \new
    methods, with points below the diagonal line showing where both methods
    succeeded but the \new method was faster; additionally, 8 and 30 instances reach limits with \veripb and \cakepb
    respectively with the \new method where the \old method succeeded,
    compared to 95 and 97 respectively with the \old method where the \new
    method succeeded.}\label{figure:satcomp}
\end{figure*}

From the analysis in Section~\ref{sec:satsumalogging}, we know
that our \new proof system is asymptotically better in theory. We now show that it indeed enables much faster proof logging and checking in practice,
on both crafted and real-world problem instances. The experiments in this section are performed on machines with dual AMD EPYC 7643 processors, 2TBytes of
RAM, and local solid state hard drives, running Ubuntu 22.04.2. We limit each individual process to 32GBytes RAM, and
run up to 16 processes in parallel (having checked that this does not make a measurable difference to runtimes). We give
\satsuma a time limit of 1,000s per instance, and \veripb and \cakepb 10,000s. We remark that runtimes involving
writing proofs are often bound by disk I/O performance; nevertheless, our general experimental trends are valid.
In each case, when we run \veripb, we run it in
elaboration mode. This means that, in addition to checking a proof, it also outputs a simplified proof that is suitable for
giving to \cakepb. This also means that any instance that fails for \veripb due to limits cannot be run through \cakepb.

The aim of our experiments is not to determine whether symmetry breaking is a good idea, or how it should be done.
Indeed, \satsuma produces the same CNF (modulo a potential sorting of the constraints) regardless of whether it is outputting proofs using the \old method or the \new method, or not outputting proofs at all. Thus, we
limit our experiments to checking that the proofs produced by \satsuma are in fact valid, rather than reporting
times for checking the entire solving process. This allows us to precisely measure the effects of our changes.

We perform experiments across two sets of instances, with different purposes. Our first set consists of five families of crafted benchmarks which have well-understood symmetries, generated using \cnfgen~\cite{LENV17CNFgen}. We note that the number of variables grows quadratically or cubically in the instance size.
\ifthenelse{\boolean{extended-version}}{
    Appendix~\ref{app:crafted} provides more details about the crafted benchmarks.
}{}

We show the results in Figure~\ref{figure:crafted-scalability}. For each of the five families, on the top row we plot the time
needed to run \satsuma to produce symmetry breaking constraints, without proof logging and with both kinds of proof
logging enabled. In each case, the \new method scales similarly to not doing proof logging, although there is a cost to
be paid to output the proofs to disk. However, particularly for the PHP and RPHP families, it is clear that even writing
the proofs is both asymptotically and practically much more expensive using the \old method. On the bottom row of the
figure, we plot the checking times. We see much better scaling from the \new method in all five cases. Formally verified
proof checking using \cakepb is slightly slower than with \veripb, which is not surprising---for the families where \cakepb's
curve stops on smaller instances, this is due to \cakepb hitting memory limits when \veripb did not.
The \new method is particularly helpful for \cakepb as its formally verified arbitrary precision arithmetic library is known to be less efficient than other (unverified) libraries~\cite{TMKFON19CakeML}.

Our second set of instances are taken from the SAT competition \cite{IJ24GlobalBenchmarkDatabase}. Because we are only interested in instances where we can
measure something interesting about symmetry breaking, we selected the 982 instances from the main competition tracks
from 2020 to 2024 where \satsuma was able to run to completion and identify at least one symmetry. In the left-hand plot
of Figure~\ref{figure:satcomp}, we show the cumulative number of instances successfully logged or checked over time. The
leftmost (``best'') curve is to run \satsuma with no proof logging, and this is closely followed by running \satsuma with
proof logging using the \new method, where we could produce proofs for all 982 instances. When producing proofs using the
\old method, in contrast, we were only able to produce proofs for 930 instances before limits were reached. We were able
to check the correctness of the symmetry breaking constraints for the \new method for 893 instances (799 with \cakepb), and
with the \old method for only 806 instances (732 with \cakepb). The two scatter plots in the figure give a more detailed
comparison of the added cost of running \satsuma with proof-logging enabled, and comparing the checking costs of the \old
and \new proof methods. In both cases it is clear that the \new method is never more than a small constant factor worse
than the \old method, and that it is often many orders of magnitude faster.

\section{Concluding Remarks}
\label{sec:conclusion}

We have presented a substantial redesign of the \veripb proof system~\cite{BGMN23Dominance,GN21CertifyingParity,Gocht22Thesis} in order to support faster certified symmetry breaking.
Central to our redesign is 
support for using auxiliary variables to encode ordering constraints over assignments.
Theoretically, the use of orders with auxiliary variables allows us to avoid encoding lexicographical orders using big integers, 
that are 
prohibitive for problems with large symmetries; this improves on the previous state-of-the-art proof logging approach \cite{BGMN23Dominance} by at least a linear factor.
To evaluate this in practice,
we implemented proof logging using our \new method in the state-of-the-art symmetry breaking tool \satsuma~\cite{AndersBR24}, and proof checking for our \extended system in \veripb and the formally verified checker \cakepb~\cite{GMMNOT24Subgraph};
our experimental evaluation shows orders-of-magnitude improvement for proof logging and checking compared to the \old approach.

Although proof logging is now asymptotically as fast as symmetry breaking, enabling proof logging can still incur a %
constant factor overhead. However, improving this would be mostly an engineering effort---we are already able to produce proofs for all of our benchmark instances within reasonable time.
Checking the proof can still be asymptotically slower than symmetry breaking in theory, which leaves room to significantly improve the performance of 
proof checking
for symmetry breaking.
The key challenge here is that the proof checker currently has to reason about all variables in the order for each symmetry broken, while the symmetry breaker does this once at a higher level.  %
Future work could investigate proof logging for conditional and dynamic symmetry breaking during search~\cite{GKLMMS05ConditionalSymmetryBreaking} or other dynamic methods of exploiting symmetries~\cite{DBB17SymmetricExplanationLearning}, in contrast to the static symmetry breaking we presented here.

 \section{Acknowledgments}
    
    We would like to thank anonymous reviewers of \emph{Pragmatics of SAT} and \emph{AAAI} for their useful comments. 
    
	This work is partially funded by the European Union (ERC, CertiFOX, 101122653). Views and opinions expressed are however those of the author(s) only and do not necessarily reflect those of the European Union or the European Research Council. Neither the European Union nor the granting authority can be held responsible for them.
	
	This work is also partially funded by the Fonds Wetenschappelijk Onderzoek -- Vlaanderen (project G064925N). 
 	
 	Benjamin Bogø and Jakob Nordström %
 	are funded by the Independent Research Fund Denmark grant \mbox{9040-00389B}. Jakob Nordström %
 	is also funded by the Swedish Research Council grant \mbox{2024-05801}.
 	  Wietze Koops and Andy Oertel %
 	are funded by the Wallenberg AI, Autonomous Systems and Software Program (WASP) funded by the Knut and Alice Wallenberg Foundation.
	Benjamin Bogø, Wietze Koops, Jakob Nordström and Andy Oertel
	also acknowledge to have benefited greatly from being part of the Basic Algorithms Research Copenhagen (BARC) environment
  financed by the Villum Investigator grant~54451.

  Ciaran McCreesh and Arthur Gontier %
  were supported by the Engineering and Physical Sciences Research Council grant number \mbox{EP/X030032/1}.
  Ciaran McCreesh was also supported by a Royal Academy of Engineering research fellowship.

  Magnus Myreen %
  is funded by the Swedish Research Council grant \mbox{2021-05165}.

  Adrian Rebola-Pardo is funded in whole or in part by the Austrian Science Fund
  (FWF) BILAI project \mbox{10.55776/COE12}.

  Yong Kiam Tan %
  was supported by the Singapore NRF Fellowship Programme \mbox{NRF-NRFF16-2024-0002}.

 	Our computational experiments used resources provided by
LUNARC at Lund University. %

  Different subsets of the authors wish to thank the participants of the Dagstuhl workshops 
  22411 \emph{Theory and Practice of SAT and Combinatorial Solving} and 
  25231 \emph{Certifying Algorithms for Automated Reasoning} 
  and of the \emph{1st International Workshop on Highlights in Organizing and Optimizing Proof-logging Systems (WHOOPS~'24)}
  at the University of Copenhagen
  for several stimulating discussions.

\ifthenelse{\boolean{extended-version}}{
\clearpage

\appendix
\section{The Cutting Planes Proof System} \label{app:cuttingplanes}

In this appendix, the cutting planes proof system is explained in sufficient detail to understand the rest of the appendix.
We will also present the syntax in the \veripb format used to write the proofs to a file and for the example in \refapp{app:example}. The cutting planes proof system was first introduced by \citet{CCT87ComplexityCP} and more details about recent results about cutting planes can be found in \citet{BN21ProofCplxSATplusCrossref}.

As mentioned in \refsec{sec:prelim}, a \emph{Boolean variable} $x$ can take values $0$ (false) and $1$ (true). A \emph{literal} is either the variable $x$ itself or its negation $\olnot{x} = 1 - x$. A \emph{pseudo-Boolean constraint} is an integer linear inequality over literals $\ell_i$
\begin{align}
    C \synteq \sum\nolimits_i a_i \ell_i \geq A \eqcomma \label{eq:app:pb-constraint}
\end{align}
where \wolog the coefficients $a_i$ and the right-hand side $A$ are non-negative and the literals $\ell_i$ are over distinct variables.
A (partial) assignment is a (partial) function mapping variables to $0$ and $1$, which is extended to literals in the natural way.
The \emph{cutting planes} proof system provides rules to derive a constraint $C$ that is implied by a formula $\formf$, \ie $\formf$ and $\formf \union \set{C}$ have the same set of satisfying assignments.

In the \veripb syntax a constraint is referred to by a positive integer called its \emph{constraint ID}.
The constraint IDs are consecutively assigned to the constraints in the order in which they are derived and start with the original formula $\formf$ getting IDs $1, \dots, \setsize{\formf}$.
A negative constraint ID can be used as a relative reference, \egabbrev the ID \texttt{-1} means the previously derived constraint. 
A cutting planes step starts with the keyword \texttt{pol} and ends with a semicolon.
After each cutting planes step, the derived constraint is assigned the next ID and added to the derived constraints in the checker.
The derivation is written in reverse Polish notation (or postfix notation), hence the checker maintains a stack of operands.
Each of the following rule pops it operands from the stack and pushed on the stack its result.

Let $C \synteq \sum_i a_i \ell_i \geq A$ and $D \synteq \sum_i b_i \ell_i \geq B$ be constraints from the formula $\formf$ or derived in previous steps.
The constraint $C$ can always be derived by the \emph{constraint axiom} rule, which is just the ID of the constraint in the syntax.
\emph{Literal axioms} $x_1 \geq 0$ and $\olnot{x}_1 \geq 0$ are derivable by using the syntax \texttt{x1} and \texttt{$\sim$x1}, respectively.
The constraint $C$ can be \emph{multiplied} by an integer $m$ by multiplying each coefficient and the right-hand side by $m$, resulting in $\sum_i (m a_i) \ell_i \geq m A$. This is denoted by \texttt{*} in the syntax and assumes the last operand on the stack is an integer and the second to last is a constraint.
The constraints $C$ and $D$ can be \emph{added} to derive $\sum_i (a_i + b_i) \ell_i \geq A + B$, which is denoted by \texttt{+} in the syntax, which assumes that the last two operands on the stack are constraints.
The constraint $C$ can be \emph{saturated} to derive $\sum_i \min\set{a_i, A} \ell_i \geq A$, which is denoted by \texttt{s} in the syntax and expects that the last operand on the stack is a constraint.
\emph{Weakening} is syntactic sugar for adding literal axioms to the constraint $C$ to eliminate the term $a_j \ell_j$ in C resulting $\sum_{i \neq j} a_i \ell_i \geq A - a_j$. This is denoted in the syntax by \texttt{w}, which assumes that the last operand on the stack is the variable over which $\ell$ is over and the second to last operand is a constraint.  

\veripb also supports the cutting planes rule of \emph{division}, which we will mention for completeness, but we will not need it for the example in \refapp{app:example}. Considering the constraint $C$ and a positive integer $d$, each coefficient and the right-hand side is divided by $d$ and rounded up, resulting in $\sum_i \ceiling{a_i / d} \ell_i \geq \ceiling{A / d}$. This is denoted by \texttt{d} in the syntax and assumes that the last operand on the stack is an integer and the second to last is a constraint. 

To better illustrate how the syntax relates to rules, suppose we have the example constraint
\begin{align}
    2 \olnot{x}_1 + 3 x_2 + 2 x_3 \geq 5 \eqcomma \label{eq:app:example-cutting-planes}
\end{align}
with constraint ID \texttt{\ref{eq:app:example-cutting-planes}}. The cutting planes proof line
\begin{lstlisting}[language=veripb,numbers=none,frame=none,aboveskip=\smallskipamount,belowskip=\smallskipamount,xleftmargin=2ex,escapechar=\$]
pol $\ref{eq:app:example-cutting-planes}$ x2 2 * + x1 w s 2 d ;
\end{lstlisting}
first pushes the constraint with ID \texttt{\ref{eq:app:example-cutting-planes}} to the stack, then the literal axiom $x_2 \geq 0$ is pushed on the stack, which is immediately afterwards multiplied by $2$ resulting on $2 x_2 \geq 0$ being pushed on the stack. Then the last two constraints on the stack, which are $2 \olnot{x}_1 + 3 x_2 + 2 x_3 \geq 5$ and $2 x_2 \geq 0$, are added, resulting in $2 \olnot{x}_1 + 5 x_2 + 2 x_3 \geq 5$. This constraint is then weakened on the variable $x_1$ resulting in the constraint $5 x_2 + 2 x_3 \geq 3$, which is saturated to yield $3 x_2 + 2 x_3 \geq 3$. Finally, dividing this constraint by $2$ yields $2 x_2 + x_3 \geq 2$, which is the result of this step and added to the derived constraints.

The \emph{slack} of the pseudo-Boolean constraint~\eqref{eq:app:pb-constraint} with respect to the (partial) assignment $\rho$ is
\begin{align*}
    \slack{\sum\nolimits_i a_i \ell_i \geq A}{\rho} \coloneq \sum\nolimits_{i: \rho(\ell_i) \neq 0} a_i - A \eqperiod
\end{align*}
If the slack is negative, then the constraint can not be satisfied by any extension of $\rho$, as there are not enough literals assigned to $1$ or unassigned. If there is a literal $\ell_i$ in the constraint $C$ that is unassigned and $0 \leq a_i < \slack{C}{\rho}$, then $\ell_i$ has to be assigned to $1$, as assigning $\ell_i$ to $0$ would lead to a negative slack. This assignment is then added to $\rho$, and we say that $C$ \emph{propagated} $\ell_i$ to $1$ under $\rho$. This process is repeated until there are no further propagations by any constraint or a constraint is falsified, \ie negative slack.
We refer to the latter case as a \emph{conflict}.

The negation of the pseudo-Boolean constraint $C$ in~\eqref{eq:app:pb-constraint} is $\neg C \synteq \sum_i a_i \olnot{\ell}_i \geq \sum_i a_i - A + 1$.
A formula $\formf$ implies a constraint $C$ by \emph{reverse unit propagation (RUP)}~\cite{GN03Verification} if $\formf \union \set{\neg C}$ propagtes to a conflict. This shows that $\neg C$ is falsified by all satisfying assignments to $\formf$, thus any assignment satisfying $\formf$ also satisfies $\formf \union \set{C}$.  

In the \veripb syntax, a RUP step contains the constraint to be derived in a modified OPB format~\cite{RM16OPB}. Additionally, the syntax also allows for hints similar to LRAT~\cite{CHHKS17EfficientCertified} that specify which constraints have to be propagated to reach a conflict. \Eg, the RUP step
\begin{lstlisting}[language=veripb,numbers=none,frame=none,aboveskip=\smallskipamount,belowskip=\smallskipamount,xleftmargin=2ex,escapechar=\$]
rup 2 x2 1 x3 >= 2 : $\ref{eq:app:example-cutting-planes}$ ;
\end{lstlisting}
derives the constraint $2 x_2 + x_3 \geq 2$ with the hint that only the constraint with ID~\texttt{\ref{eq:app:example-cutting-planes}} above ($2 \olnot{x}_1 + 3 x_2 + 2 x_3 \geq 5$) is required in addition to the negation ($2 \olnot{x}_2 + \olnot{x}_3 \geq 2$) of the constraint to propagate from the empty assignment to a conflict.

\section{Full Description of Extended Proof System} \label{app:proof_system_full}

In this appendix, we provide %
further details on the extended proof system introduced in
Sections~\ref{ssc:specifications} up to~\ref{ssc:redrule}, 
including the statement of all results in the more general context 
of optimization problems, and all proofs. For a discussion on the intuition
behind the proof system, we refer the start of Section~\ref{sec:extending-proof-system}.
To make it easier to read the appendix as its own section, 
we also include all content that was already in 
Sections~\ref{ssc:specifications} up to~\ref{ssc:redrule} in the main paper.

\subsection{Specifications}
\label{appssc:specifications}

Let $\varset$ be a \listt of variables. 
A pseudo-Boolean formula $\specification(\vx, \varset)$ is a \emph{specification over the variables $\varset$}, if it is derivable 
from the empty formula $\emptyset$ by the redundance rule, where each application only witnesses 
over variables in~$\varset$.

\begin{definition}
A formula $\specification(\vx, \varset) = \set{\constrc_1, \constrc_2, \dots{}, \constrc_n}$ is a \emph{specification over the variables $\varset$}, if there is a list 
\begin{equation*}
(\constrc_1, \witness_1), (\constrc_2, \witness_2), \dots{}, (\constrc_n, \witness_n)
\end{equation*}
which satisfies the following: 
\begin{enumerate}
    \item The constraint $\constrc_1$ can be obtained from the empty formula $\emptyset$ using the redundance rule with witness $\witness_1$.
    \item For each $i \in \{2, \dots{}, n\}$ we have that 
$C_i$ can be added by the redundance rule to
$\bigcup_{j=1}^{i-1} \{\constrc_j\} $ with the witness $\witness_i$. 
In other words, it should hold that 
\begin{equation}
\bigcup_{j=1}^{i-1} \{\constrc_j\}  \cup \{\lnot \constrc_i\} \vdash \bigcup_{j=1}^{i} \{\subst{\constrc_j}{\witness_i}\}.
\label{appeq:redspec}
\end{equation}
    \item For every witness $\witness_i$,  $\supp(\witness_i) \subseteq \varset$ holds.  
\end{enumerate}
\end{definition}

In terms of configurations, we can say that a formula $\specification(\vx, \varset) = \set{\constrc_1, \constrc_2, \dots{}, \constrc_n}$ is a specification over $\varset$ if we can transition from the configuration $(\emptyset, \emptyset, \emptyset, \emptyset)$ to the configuration $(\emptyset, \specification(\vx, \varset), \emptyset, \emptyset)$ using only the redundance rule with witnesses $\witness_i$ such that $\supp(\witness_i) \subseteq \varset$ holds. 

\begin{remark}
Note that the definition of the specification uses the redundance rule, while the redundance rule in our extended proof system is only defined later. This is not a problem since the applications of the redundance rule in a specification require that the loaded order is the trivial preorder $\trivialorder \synteq \emptyset$ relating all assignments, for which it is irrelevant whether we use auxiliary variables or not.
\end{remark}

A crucial property of specifications is that we can recover an assignment of the auxiliary variables from 
the assignment of the non-auxiliary variables.
We state this property below.
\begin{lemma} Let $\specification(\vx, \va)$ 
    be a specification over $\va$.
    Let $\assmntone$ be any assignment of the variables $\vx$. 
    Then, $\assmntone$ can be extended to an assignment $\assmnttwo$, such that
    \begin{enumerate} 
        \item $\assmnttwo$ satisfies $\specification$, and 
        \item $\assmntone(x) = \assmnttwo(x)$ holds for every $x \in \vx$. 
    \end{enumerate}
    \label{applem:spec_behaves_nicely}
\end{lemma}

To explain why Lemma~\ref{applem:spec_behaves_nicely} holds, recall from Section~\ref{ssc:original} that the redundance rule %
satisfies the following: if a constraint $\constrc$ is added by redundance with witness $\witness$, then given an assignment $\assmntone$ satisfying all other constraints, either  $\assmntone$ or  $\assmntone \circ \witness$ also satisfies $\constrc$. Hence, defining $\assmnttwo$ by composing $\assmntone$ with the witnesses~$\witness_i$ corresponding to the constraints
that do not already hold
ensures that $\assmnttwo$ satisfies $\specification$, and the fact that each witness~$\witness_i$ is the identity on $\vx$ (since $\supp(\witness_i) \subseteq \varset$) ensures that $\assmntone(x) = \assmnttwo(x)$ for $x \in \vx$. 

We now give a complete proof of Lemma~\ref{applem:spec_behaves_nicely}.

\begin{proof}[Proof of Lemma \ref{applem:spec_behaves_nicely}]	
Write $\specification(\vx, \varset) = \set{\constrc_1, \constrc_2, \dots{}, \constrc_n}$.
We show by induction on $i \geq 0$ that we can extend $\assmnta$ to an assignment $\assmnta_i$ such that 
\begin{enumerate} 
    \item $\assmnta_i$ satisfies $\formf_i = \set{\constrc_1, \constrc_2, \dots{}, \constrc_i}$, and
    \label{appeq:lem1proof1}
    \item $\assmntone(x) = \assmnta_i(x)$ holds for every $x \in \vx$. 
    \label{appeq:lem1proof2}
\end{enumerate}
For the base case $i=0$, we extend $\assmnta$ to an assignment $\assmnta_0$ by setting all variables of $\va$ to false. Then  $\assmnta_0$ trivially satisfies $\formf_0 = \emptyset$, and $\assmntone(x) = \assmnta_0(x)$ for every $x \in \vx$ since we only change the assignment of $\va$.

We proceed with the induction step. 
By the induction hypothesis, we can extend $\assmnta$ to an assignment $\assmnta_i$ such that 
$\assmnta_i$ satisfies $\formf_i$ and $\assmntone(x) = \assmnta_i(x)$ holds for every $x \in \vx$. 

If $\assmnta_i$ also happens to satisfy $\constrc_{i+1}$, we set $\assmnta_{i+1} = \assmnta_i$. Then $\assmnta_{i+1}$ indeed satisfies $\formf_{i+1} = \formf_i \cup \{\constrc_{i+1}\}$ and we have $\assmntone(x) = \assmnta_i(x) = \assmnta_{i+1}(x)$ for every $x \in \vx$.

    If $\assmnta_i$ does not satisfy $\constrc_{i+1}$, we set $\assmnta_{i+1} = \assmnta_i \circ \witness_{i+1}$. Since then $\assmnta_i$ satisfies $\bigcup_{j=1}^{i} \{\constrc_j\}  \cup \{\lnot \constrc_{i+1}\}$, Equation~\eqref{appeq:redspec} then implies that $\assmnta_i$ satisfies $\subst{\formf_{i+1}\!}{\witness_{i+1}} = \bigcup_{j=1}^{i+1} \{\subst{\constrc_j}{\witness_{i+1}}\}$, which in turn implies that $\assmnta_{i+1} = \assmnta_i \circ \witness_{i+1}$ satisfies $\formf_{i+1}$. Since  $\supp(\witness_{i+1}) \subseteq \varset$, we have $\witness_{i+1}(x) = x$ for every $x \in \vx$, so $\assmntone(x) = \assmnta_i(x) = \assmnta_i(\witness_{i+1}(x)) = (\assmnta_i \circ \witness_{i+1})(x) = \assmnta_{i+1}(x)$ for every $x \in \vx$.
    This completes the induction. 
    
    Since $\formf_n = \specification$ holds by definition, the claim follows.
\end{proof}

\subsection{Orders with Auxiliary Variables} \label{appsubsec:orderswithauxvars}
Next, we explain how two pseudo-Boolean formulas $\ord(\vu, \vv, \va)$ and $\spec(\vu, \vv, \va)$, together with 
the two disjoint \lists of variables $\vz$ and $\va$, define a preorder.

\begin{definition} \label{appdef:induced-preorder}
Let $\vu$ and $\vv$ be disjoint \lists of variables of size~$\nvar$ and let 
$\va$ be a \listt of auxiliary variables. 
Let $\ord(\vu, \vv, \va)$ and $\spec(\vu, \vv, \va)$ be two pseudo-Boolean formulas such that $\spec$ is a specification over $\va$. 

Then we define the relation $\preceq$  over the domain of total assignments to a 
\listt of variables $\vz$ of size~$\nvar$ as follows:
For assignments $\assmnta, \assmntb$ we let $\assmnta \preceq \assmntb$ hold, 
if and only if there exists an assignment $\assmntp$ to the variables $\va$, such that 
\begin{equation*}
\spec(\vz\uh_{\assmnta},\vz\uh_{\assmntb}, \va\uh_{\assmntp}) \land
\ord(\vz\uh_{\assmnta},\vz\uh_{\assmntb}, \va\uh_{\assmntp})
\end{equation*}
evaluates to true.
\end{definition}

To ensure that $\ord$ and $\spec$ 
actually define a preorder,
we require cutting planes proofs that
show reflexivity, \ie 
$\emptyset \derives \assmnta\preceq\assmnta$, 
and transitivity, \ie 
$\assmnta\preceq\assmntb \land 
\assmntb\preceq\assmntc\derives 
\assmnta\preceq \assmntc$. 
To write these proof obligations using the cutting planes proof system, 
which cannot handle an existentially quantified conclusion,
we can use the specification as a premise. 
The specification premise essentially tells us which auxiliary variables 
the existential quantifier should pick. 
In particular, for reflexivity, the proof obligation is 
\begin{equation}
\spec(\vx, \vx, \va) \vdash{} \ord(\vx, \vx, \va) \eqperiod
  \label{appeq:preorder1}
\end{equation}
For transitivity, the proof obligation is
\begin{equation}
\begin{split}
\spec(\vx, \vy, \va) &\cup \ord(\vx, \vy, \va) \cup \spec(\vy, \vz, \vb) \\ &\cup 
  \ord(\vy, \vz, \vb) \cup \spec(\vx, \vz, \vc)   \vdash{} \ord(\vx, \vz, \vc) \eqperiod
  \label{appeq:preorder2}
\end{split}
\end{equation}

Intuitively, \eqref{appeq:preorder2} says that if the circuits defining the auxiliary variables are correctly evaluated, which is encoded by the premises $\spec(\vx, \vy, \va) \cup \spec(\vy, \vz, \vb) \cup \spec(\vx, \vz, \vc)$, then transitivity should hold, \ie 
$\ord(\vx, \vy, \va) \cup 
  \ord(\vy, \vz, \vb)  \vdash{} \ord(\vx, \vz, \vc)$. 
However, if the auxiliary variables are not correctly set, then no claims are made.

These proof obligations 
ensure that $\preceq$ is a preorder:

\begin{lemma} If $\ord$ and $\spec$ satisfy Equations~\eqref{appeq:preorder1} and \eqref{appeq:preorder2}, then 
    $\preceq$ as defined by $\ord$ and $\spec$ is a preorder.
\end{lemma}
\begin{proof}
We need to show that $\preceq$ is reflexive and transitive.

\textit{(Reflexivity.)} Let $\assmnta$ be a total assignment to the variables $\vz$. We have to show that $\assmnta \preceq \assmnta$ holds, that is, 
there is an assignment $\assmntp$ to the variables $\va$ such that 
\begin{equation}
    \spec(\vz\uh_\assmnta, \vz\uh_\assmnta, \va\uh_\assmntp) \wedge \ord(\vz\uh_\assmnta, \vz\uh_\assmnta, \va\uh_\assmntp)
\end{equation}
is satisfied. From Lemma~\ref{lem:spec_behaves_nicely}, it follows that we can choose an extension $\assmntp$ to $\assmnta$, such that $\spec(\vz\uh_\assmnta, \vz\uh_\assmnta, \va\uh_\assmntp)$ is satisfied.
    Since this suffices to satisfy the preconditions of Equation~\ref{appeq:preorder1}, it follows that also $\ord(\vz\uh_\assmnta, \vz\uh_\assmnta, \va\uh_\assmntp)$ is satisfied, which suffices to conclude that $\assmnta \preceq \assmnta$ holds.

\textit{(Transitivity.)}
Let $\assmnta, \assmntb, \assmntc$ be total assignments to the variables $\vz$. We have to show that if $\assmnta \preceq \assmntb$ and $\assmntb \preceq \assmntc$ hold, then $\assmnta \preceq \assmntc$ holds.
That is we can assume that there are assignments $\assmntp_1, \assmntp_2$ to $\va$ such that 
\begin{equation} 
\begin{split}   
    \spec&(\vz\uh_\assmnta, \vz\uh_\assmntb, \va\uh_{\assmntp_1}) \cup 
    \ord(\vz\uh_\assmnta, \vz\uh_\assmntb, \va\uh_{\assmntp_1})\\ 
    &\cup \spec(\vz\uh_\assmntb, \vz\uh_\assmntc, \va\uh_{\assmntp_2}) \cup 
    \ord(\vz\uh_\assmntb, \vz\uh_\assmntc, \va\uh_{\assmntp_2})
\end{split}
\end{equation}
is satisfied.
From Lemma~\ref{lem:spec_behaves_nicely}, it follows that we can choose an extension $\assmntp_3$,
such that $\spec(\vz\uh_\assmnta, \vz\uh_\assmntc, \va\uh_{\assmntp_3})$ is satisfied.
Since this suffices to satisfy the preconditions of Equation~\eqref{appeq:preorder2}, it follows that also $\ord(\vz\uh_\assmnta, \vz\uh_\assmntc, \va\uh_{\assmntp_3})$ is satisfied, from which we conclude that $\assmnta \preceq \assmntc$ holds.
\end{proof}

\subsection{Validity} \label{appsec:validity}

We extend the configurations of the proof system to $\conf$,
where, as explained in \citet{BGMN23Dominance}, $v$ denotes the current upper bound on the objective $f$. 
In particular, we extend the notion of \emph{weak-$(\formf, f)$-validity} 
from \citet{BGMN23Dominance} to our new configurations. 
For this, we first define the relation $\preceq_f$ as follows:
\begin{equation}
\assmnta \preceq_f \assmntb \quad \text{iff} \quad \assmnta \preceq \assmntb \; \wedge \; f\uh_\assmnta \,\leq\, f\uh_\assmntb\eqperiod
\end{equation}
We define the strict order $\prec_f$ by defining $\assmnta \prec_f \assmntb$ to mean that both $\assmnta \preceq_f \assmntb$ and $\assmntb \not\preceq_f \assmnta$ hold.

\begin{definition} \label{appdef:validity}
  A configuration $\conf{}$ is \emph{weakly-$(\formf, f)$-valid} if the following conditions hold:
  \begin{enumerate}
    \item 
        For every $v' < v$, it holds that if
        $F \cup \{f \leq v'\}$
        is satisfiable, then $\core \cup \{f \leq v'\}$
        is satisfiable.
    \item 
        For every total assignment~$\assmnta$ satisfying the constraints
        $\core \cup \{f \leq v-1\}$, there exists a total assignment
        $\assmntaprime \preceq_f \assmnta$ satisfying $\core \cup \der \cup \{f \leq v-1\}$. 
  \end{enumerate}
  \label{appdef:valid}
\end{definition}
In the following, we furthermore assume that
for any configuration $\conf{}$, the following hold:

\begin{enumerate}
\item $\ord$ and $\spec$ refer to formulas for which Equations~\eqref{appeq:preorder1} and~\eqref{appeq:preorder2} have been successfully proven.
\item $\spec$ is a specification over $\va$.
\item The variables $\va$ only occur in $\ord$ and $\spec$, 
      and these variables are disjoint %
      from $\vz$. 
\end{enumerate}

These invariants only concern $\ord$, $\spec$, $\va$ and $\vz$. Whenever the order is changed, it is ensured that these invariants hold, and other rules do not change these components of the configuration.

Observe that due to these invariants, satisfying assignments for $\core \cup \der$ 
do not need to assign the 
variables $\va$. 
In the following, we assume that such assignments are indeed defined only over the domain 
$\var(\core \cup \der)$.

\subsection{Dominance-Based Strengthening Rule}

As in the \original system, the dominance rule allows adding a constraint $\constrc$ to the derived set  using witness $\witness$ if from the premises $\core \cup \der \cup \{\neg \constrc\}$ we can derive $\core \uh_\witness$ and show that $\assmnta \circ \witness \prec \assmnta$ holds for all assignments $\assmnta$ satisfying $\core \cup \der \cup \{\neg \constrc\}$. To show  $\assmnta  \circ \witness \prec \assmnta$, we separately show that  $\assmnta  \circ \witness \preceq \assmnta$ and $\assmnta \not\preceq \assmnta \circ \witness$. To show that  $\assmnta  \circ \witness \preceq \assmnta$, we have to show that $\ord(\vz\uh_\witness, \vz, \va)$, assuming that the circuit defining the auxiliary variables $\va$ has been evaluated correctly, which is encoded by the specification $\spec(\vz\uh_\witness, \vz, \va)$.  This leads to the proof obligation
\begin{equation}
\begin{split}
\core \cup \der &\cup \{\neg \constrc\}  \cup \{f \leq v-1\} \\ 
 &\cup \spec(\vz\uh_\witness, \vz, \va) \vdash \ord(\vz\uh_\witness, \vz, \va)\eqperiod
\end{split}
\end{equation}
To show that $\assmnta \not\preceq \assmnta \circ \witness$, we have to show that $\neg\ord(\vz, \vz\uh_\witness, \va)$ assuming $ \spec(\vz, \vz\uh_\witness,\va)$. However, since  $\neg\ord(\vz, \vz\uh_\witness, \va)$ is not necessarily a pseudo-Boolean formula (due to the negation), we instead show that we can derive contradiction from $\ord(\vz, \vz\uh_\witness, \va)$, leading to the proof obligation:
\begin{equation}
\begin{split}
\core \cup \der &\cup \{\neg \constrc\}  \cup \{f \leq v-1\} \\ 
 &\cup \spec(\vz, \vz\uh_\witness,\va) \cup \ord(\vz, \vz\uh_\witness, \va) \vdash \bot \eqperiod
\end{split}
\end{equation}
The following lemma shows that these proof obligations indeed imply $\assmnta \circ\witness \preceq \assmnta$ and $\assmnta \not\preceq \assmnta\circ\witness$, respectively: 
 
\begin{lemma} 
Let $\formg$ be a formula and $\witness$ a witness with 
$\supp(\witness) \subseteq \var(\formg)$. 
Furthermore, let $\va \cap \var(\formg) = \emptyset$ and $\spec$ be a specification over $\va$.
              Also, let $\ord$ and $\spec$ define a pre-order $\preceq$. Then the following hold:
\begin{enumerate}
    \item If $\formg \cup \spec(\vz\uh_\witness, \vz, \va) 
    \vdash \{f\uh_\witness \leq f\} \cup \ord(\vz\uh_\witness, \vz, \va)$, 
    then for each assignment $\assmnta$ satisfying $\formg$, 
    $\assmnta \circ\witness \preceq_f \assmnta$ holds.
    \item If $\formg \cup \spec(\vz, \vz\uh_\witness, \va) 
    \cup \ord(\vz, \vz\uh_\witness, \va) \vdash \bot$
    holds, then for each assignment~$\assmnta$ satisfying~$\formg$, 
    $\assmnta \not\preceq_f \assmnta\circ\witness$ holds. 
\end{enumerate} 
\label{applem:order_reasonable_both}
\end{lemma}

\begin{proof}
1. Let $\assmnta$ be an assignment satisfying $\formg$. 
By Lemma~\ref{lem:spec_behaves_nicely}, we can extend $\assmnta$ to $\assmntp$ 
such that $\spec(\vz\uh_{\assmnta \circ \witness}, \vz \uh_{\assmnta}, \va \uh_{\assmntp})$
is satisfied. This means that $\assmntp$ satisfies the preconditions of 
\begin{equation}
\formg \cup \spec(\vz\uh_\witness, \vz, \va)  \vdash \{f\uh_\witness \leq f\} \cup \ord(\vz\uh_\witness, \vz, \va)\eqcomma
\end{equation}
so also $\ord(\vz\uh_{\assmntp \circ \witness}, \vz \uh_{\assmntp}, \va \uh_{\assmntp})$ holds. 
Since $\assmnta$ and $\assmntp$ coincide on~$\vz$ (which holds since $\va$ and $\vz$ are disjoint), it follows that
$\ord(\vz\uh_{\assmnta \circ \witness}, \vz \uh_{\assmnta}, \va \uh_{\assmntp})$ holds,
which is sufficient to conclude that $\assmnta\circ\witness \preceq \assmnta$ holds. In addition,  $\assmntp$ satisfies $\{f\uh_\witness \leq f\}$. Since $\assmntp$ and $\assmnta$ coincide on the variables of $f$, this implies that $f\uh_{\assmnta \circ \witness} \leq f_\assmnta$, so together with $\assmnta\circ\witness \preceq \assmnta$ we conclude that $\assmnta\circ\witness \preceq_f \assmnta$ holds, as required.

2. Let $\assmnta$ be an assignment satisfying $\formg$. 
We argue by contradiction. Suppose that $\assmnta \preceq \assmnta\circ\witness$ holds. 
Then there exists an assignment $\assmntp$ to $\va$ such that
\[\spec(\vz_{\assmnta}, \vz\uh_{\assmnta \circ \witness}, \va\uh_{\assmntp}) 
    \cup \ord(\vz_{\assmnta}, \vz\uh_{\assmnta \circ \witness}, \va\uh_{\assmntp})\eqperiod\]
Define  $\assmntaprime$ to coincide with $\assmntp$ on $\va$, and with $\assmnta$ otherwise. 
Since $\va$ and $\vz$ are disjoint, $\assmntaprime$ then coincides with $\assmnta$ on $\vz$. 
Since $\va \cap \var(\formg) = \emptyset$,  $\assmntaprime$ also coincides with $\assmnta$ on $\var(\formg)$. 
Hence, $\assmntaprime$ satisfies $\formg$ and 
\[\spec(\vz_{\assmntaprime}, \vz\uh_{\assmntaprime \circ \witness}, \va\uh_{\assmntaprime}) 
    \cup \ord(\vz_{\assmntaprime}, \vz\uh_{\assmntaprime \circ \witness}, \va\uh_{\assmntaprime})\eqcomma\]
so $\assmntaprime$ satisfies the preconditions of $\formg \cup \spec(\vz, \vz\uh_\witness, \va) 
    \cup \ord(\vz, \vz\uh_\witness, \va) \vdash \bot$, which is clearly contradictory. Hence, we have $\assmnta \not\preceq \assmnta\circ\witness$, which implies $\assmnta \not\preceq_f \assmnta\circ\witness$.
\end{proof}

Hence, we define the dominance rule as follows:

\begin{definition}[Dominance-based strengthening with specification]
    We can transition from the configuration $\conf$ to $(\core, \der \cup \{\constrc\}, \ord, \spec, \vz, \va, v)$ using the dominance rule 
    if the following conditions are met: 
    \begin{enumerate}
        \item The constraint $\constrc$ does not contain variables in $\va$.
        \item There is a witness $\witness$ for which $\img(\witness) \cap \va = \emptyset$ holds.
        \item We have cutting planes proofs for the following:
\end{enumerate}
\begin{equation}
\begin{split}
\core \cup \der &\cup \{\neg \constrc\}  \cup \{f \leq v-1\} \cup \spec(\vz\uh_\witness, \vz, \va) \\ 
     & \vdash \core\uh_\witness \cup \{f\uh_\witness \leq f\} \cup  \ord(\vz\uh_\witness, \vz, \va)\eqcomma \label{appeq:dom1} \\
\end{split}
\end{equation}
\begin{equation}
\begin{split}
\core \cup \der &\cup \{\neg \constrc\}  \cup \{f \leq v-1\} \\ 
 &\cup \spec(\vz, \vz\uh_\witness,\va) \cup \ord(\vz, \vz\uh_\witness, \va) \vdash \bot \eqperiod \label{appeq:dom2}
\end{split}
\end{equation}
\end{definition}

Using Lemma~\ref{applem:order_reasonable_both}, we can show that the dominance rule preserves the invariants required by weak validity: 

\begin{lemma}
    If we can transition from $\conf$ to 
    $(\core, \der \cup \{\constrc\}, \ord, \spec, \vz, \va, v)$ 
    by the dominance rule, and the configuration $\conf$ is weakly-$(\formf, f)$-valid, 
    then the configuration~$(\core, \der \cup \{\constrc\}, \ord, \spec, \vz, \va, v)$ is also weakly-$(\formf, f)$-valid.
\end{lemma}

\begin{proof}
    Since~$\formf$, $f$, $\core$, and $v$ are not affected, item 1. in Definition~\ref{appdef:valid} still holds.
    Hence, it remains to show that item 2. holds, \ie that for every total assignment~$\assmnta$ that 
    satisfies~$\core \cup \{f \leq v-1\}$ there exists a total assignment $\assmntaprime \preceq_f \assmnta$ 
    that satisfies~$\core \cup \der \cup \{\constrc\} \cup \{f \leq v-1\}$.

    Assume towards a contradiction that it does not hold. 
    Let $S$ denote the set of total assignments $\assmnta$ that 
    \begin{enumerate}
        \item satisfy $\core \cup \{f \leq v-1\}$ and 
        \item admit no $\assmntaprime \preceq_f \assmnta$ satisfying 
        $\core \cup \der \cup \{\constrc\} \cup \{f \leq v-1\}$.
    \end{enumerate}
    By our assumption, $S$ is non-empty.
	Since $\prec_f$ is a strict order and $S$ is finite and non-empty, 
	$S$ contains a $\prec_f$-minimal element. 
	Let $\assmnta$ be some $\prec_f$-minimal assignment in $S$. 
	Since $\conf$ is weakly-$(F, f)$-valid, 
	there exists some $\assmnta_1 \preceq_f \assmnta$ that satisfies 
	$\core \cup \der \cup \{f \leq v-1\}$.
	We know that $\assmnta_1$ cannot satisfy $\constrc$ since $\assmnta \in S$. 
	Hence, $\assmnta_1$ satisfies 
	\begin{equation*}
	\formg = \core \cup \der \cup \{\neg \constrc\} \cup \{f \leq v-1\} \eqperiod
	\end{equation*}
	
	Since $\formg$ does not contain variables of $\va{}$, both points of Lemma~\ref{applem:order_reasonable_both}
	apply by \eqref{appeq:dom1} and \eqref{appeq:dom2}, respectively. Hence, it follows that 
	$\assmnta_1\circ\witness \preceq_f \assmnta_1$ and  
	$\assmnta_1 \not\preceq_f \assmnta_1\circ\witness$ hold, 
	which together imply that 
	$\assmnta_1\circ\witness \prec_f \assmnta_1$.
	Note that \eqref{appeq:dom1} also implies that $\assmnta_1$ satisfies $\core\uh_\witness$, 
	so $\assmnta_1\circ\witness$ satisfies $\core$. Moreover, $\assmnta_1$ satisfies 
	$f \leq v-1$ and $f\uh_\witness\, \leq f$, which together imply $f\uh_\witness\, \leq v-1$.
	Hence, $\assmnta_1\circ\witness$ satisfies $f \leq v-1$.
	
	Let  $\assmnta_2 = \assmnta_1\circ\witness$. Then $\assmnta_2$ satisfies  $\core \cup \{f \leq v-1\}$.
	Since $\assmnta_2 \prec_f \assmnta_1$ and $\assmnta_1 \preceq_f \assmnta$, we have $\assmnta_2\prec_f \assmnta$.
	Since $\assmnta$ is a $\prec_f$-minimal element of $S$, this implies that $\assmnta_2 \not\in S$. 
	Hence, since $\assmnta_2$ does satisfy $\core \cup \{f \leq v-1\}$, it follows that $\assmnta_2$ 
	does admit an $\assmntaprime \preceq_f \assmnta_2$ satisfying 
	$\core \cup \der \cup \{\constrc\} \cup \{f \leq v-1\}$.
	But it also holds that $\assmntaprime \preceq_f \assmnta_2 \prec_f \assmnta$, so $\assmntaprime \preceq_f \assmnta$. 
	This contradicts $\assmnta \in S$, finishing the proof.
\end{proof}

\subsection{Redundance-Based Strengthening Rule} \label{appssc:redrule}
We also modify the redundance rule to work in our \extended proof system.
Similarly to the dominance rule, we can use $\spec(\vz\uh_\witness, \vz, \va)$ as an 
extra premise in our proof obligations. 

\begin{definition}[Redundance-based strengthening with specification]
We can transition from $\conf$ to $(\core, \der \cup \{\constrc\}, \ord, \spec, \vz, \va, v)$ using 
the redundance rule if the following conditions are met: 
\begin{enumerate}
    \item The constraint $\constrc$ does not contain variables in $\va$.
    \item There is a witness $\witness$ for which $\img(\witness) \cap \va = \emptyset$ holds.
    \item We have cutting planes proof that the following holds:
\end{enumerate}
\begin{equation}
\begin{split}
    \core &\cup \der \cup \{\neg \constrc\} \cup \{f \leq v-1\} \cup \spec(\vz\uh_\witness, \vz, \va)  \\
          &\vdash (\core \union \der \union \{ \constrc \})\uh_{\witness} \cup \{f\uh_\witness \leq f\} 
           \cup \ord(\vz\uh_\witness, \vz, \va) \eqperiod \label{appeq:red}
\end{split}
\end{equation}
\end{definition}

The redundance rule preserves weak validity: 

\begin{lemma}\!\!
    \label{applem:red_preserves_validity}
    \mbox{If we can transition from $\conf$} to 
    $(\core, \der \cup \{\constrc\}, \ord, \spec, \vz, \va, v)$
    by the redundance rule, and the configuration $\conf$ is weakly-$(\formf, f)$-valid, 
    then the configuration~$(\core, \der \cup \{\constrc\}, \ord, \spec, \vz, \va, v)$ 
    is also weakly-$(\formf, f)$-valid.
\end{lemma}

\begin{proof}
    As for the dominance rule, $\formf$, $f$, $\core$, and $v$ are not affected, 
    so item 1. in Definition~\ref{appdef:valid} still holds.
    Hence, it remains to show that item 2. holds, \ie that for every total assignment~$\assmnta$ that 
    satisfies~$\core \cup \{f \leq v-1\}$ there exists a total assignment $\assmntaprime \preceq_f \assmnta$ 
    that satisfies
    \begin{equation*}
        \core \cup \der \cup \{\constrc\} \cup \{f \leq v-1\}\eqperiod
    \end{equation*}
    Let $\assmnta$ be an assignment that satisfies~$\core \cup \{f \leq v-1\}$.
	Since $\conf$ is weakly-$(F, f)$-valid, 
	there exists some $\assmnta_1 \preceq_f \assmnta$ that satisfies 
	$\core \cup \der \cup \{f \leq v-1\}$.
	If $\assmnta_1$ also satisfies $\constrc$, then we are done. Otherwise, $\assmnta_1$ satisfies 
	\begin{equation*}
	\formg = \core \cup \der \cup \{\neg \constrc\} \cup \{f \leq v-1\} \eqperiod
	\end{equation*}
	Since $\formg$ does not contain variables of $\va{}$, 
	point 1 of Lemma~\ref{applem:order_reasonable_both}
	apply by \eqref{appeq:red}. Hence, it follows that 
	$\assmnta_1\circ\witness \preceq_f \assmnta_1$. 
	Equation \eqref{appeq:red}  also implies that $\assmnta_1$ satisfies 
	$(\core \union \der \union \{ \constrc \})\uh_{\witness}$, so 
	$\assmnta_1\circ\witness$ satisfies $\core \union \der \union \{ \constrc \}$.
	Finally, $\assmnta_1$ satisfies 
	$f \leq v-1$ and $f\uh_\witness\, \leq f$, which together imply $f\uh_\witness\, \leq v-1$.
	Hence, $\assmnta_1\circ\witness$ satisfies $f \leq v-1$.
	Then $\assmnta_1\circ\witness \preceq_f \assmnta_1 \preceq_f \assmnta$ and 
	$\assmnta_1\circ\witness$ satisfies $\core \union \der \union \{ \constrc \} \cup \{f \leq v-1\}$,
	so $\assmnta_1\circ\witness$ is an assignment showing that 
	item 2. in Definition~\ref{appdef:valid} holds for $\assmnta$.
\end{proof}

\section{Proof Logging in Satsuma} \label{app:satsumalogging}

In this appendix, we show how our proof system using auxiliary variables is 
used to log the symmetry-breaking clauses generates by \satsuma. 
This is done in two steps:
\begin{enumerate}
\item Defining the lexicographical order using auxiliary variables 
and proving its transitivity and reflexivity.
\item Deriving the symmetry-breaking clauses using the dominance rule 
and the loaded lexicographical order.
\end{enumerate}

\subsection{Defining the Lexicographical Order}
\label{ssc:satsuma_order}

We first discuss how to define the lexicographical order on sequences of 
length $\nvar$ using auxiliary variables. 
In Lemma~\ref{lem:lex_using_aux}, we provide an encoding of the 
lexicographical order using auxiliary variables, 
and show that this indeed encodes the lexicographical order. 
In Lemma~\ref{lem:lex_using_aux_PB}, we explain how 
to translate this encoding to pseudo-Boolean constraints. 
In Lemma~\ref{lem:lex_using_aux_reflexivity}, we then show that the 
reflexivity follows by reverse unit propagation (RUP), 
while Lemma~\ref{lem:lex_using_aux_trans} shows that transitivity 
can be shown using a cutting planes derivation of size $\bigoh{\nvar}$. 
Finally, this implies Theorem~\ref{thm:lexorderdef}, which states that we 
can define the lexicographical order in \veripb using a derivation 
of size $\bigoh{\nvar}$. 

Throughout this section, $\nvar \in \Nplus$ denotes the length of the sequences
over which we define lexicographical order, $\varx_i$ and $\vary_i$ are Boolean
variables, and $\auxa_i$ and $\auxd_i$ are (auxiliary) Boolean variables.

\begin{lemma} \label{lem:lex_using_aux}
Let $\varx_1, \ldots,  \varx_\nvar$ and $\vary_1, \ldots, \vary_\nvar$ be given 
sequences of Boolean variables. 
Define the Boolean variables $\auxa_1, \ldots, \auxa_{\nvar-1}$ and 
$\auxd_1, \ldots, \auxd_\nvar$ inductively by 
\begin{align}
\auxa_1 &\iff (\varx_1 \geq \vary_1)\eqcomma \label{eq:a1}\\
\auxa_{i+1} &\Leftrightarrow (\auxa_i \;\wedge\; \varx_{i+1} \geq \vary_{i+1})\eqcomma \label{eq:ai} \\
\auxd_1 &\Leftrightarrow (\vary_1 \geq \varx_1)\eqcomma \label{eq:d1} \\
\auxd_{i+1} &\Leftrightarrow (\auxd_i \;\wedge\; (\olnot{\auxa_i} \,\vee\, \vary_{i+1} \geq \varx_{i+1}))\eqperiod \label{eq:di}
\end{align}
Then $\auxd_\nvar$ holds if and only if $(\varx_1, \ldots,  \varx_\nvar) \lexorder (\vary_1, \ldots, \vary_\nvar)$.
\end{lemma}

\begin{proof}
We show by induction on $i \geq 1$ that $\auxd_i$ holds if and only if 
$(\varx_1, \ldots,  \varx_i) \lexorder (\vary_1, \ldots, \vary_i)$ (for $i \leq \nvar$) 
and that $\auxd_i \wedge \auxa_i$ holds if and only if 
$(\varx_1, \ldots,  \varx_i) = (\vary_1, \ldots, \vary_i)$ (for $i \leq \nvar-1$). 
The base case is immediate from the definitions of $\auxa_1$ and~$\auxd_1$. 

We proceed with the induction step. Note that we have to show two equivalences, so four implications.

\paragraph{$\auxd_{i+1}$ implies $(\varx_1, \ldots,  \varx_{i+1}) \lexorder (\vary_1, \ldots, \vary_{i+1})$.} 
Suppose that  $\auxd_{i+1}$ holds. Then~\eqref{eq:di} implies that $\auxd_i$ 
and $\olnot{\auxa_i} \,\vee\, \vary_{i+1} \geq \varx_{i+1}$ hold. 
Since $\auxd_i$ holds, the induction hypothesis implies that 
\begin{equation}
(\varx_1, \ldots,  \varx_i) \lexorder (\vary_1, \ldots, \vary_i).\label{eq:lexi}
\end{equation}
We split cases using $\olnot{\auxa_i} \,\vee\, \vary_{i+1} \geq \varx_{i+1}$:
\begin{itemize}
\item If $\olnot{\auxa_i}$, then $\auxd_i \wedge \auxa_i$ does not hold, so the 
induction hypothesis implies $(\varx_1, \ldots,  \varx_i) \neq (\vary_1, \ldots, \vary_i)$, 
which combined with \eqref{eq:lexi} implies 
$(\varx_1, \ldots,  \varx_{i+1}) \lexorder (\vary_1, \ldots, \vary_{i+1})$. 
\item If $\vary_{i+1} \geq \varx_{i+1}$, then  \eqref{eq:lexi} also implies
$(\varx_1, \ldots,  \varx_{i+1}) \lexorder (\vary_1, \ldots, \vary_{i+1})$.
\end{itemize}
Hence, we have $(\varx_1, \ldots,  \varx_{i+1}) \lexorder (\vary_1, \ldots, \vary_{i+1})$.

\paragraph{$(\varx_1, \ldots,  \varx_{i+1}) \lexorder (\vary_1, \ldots, \vary_{i+1})$  implies $\auxd_{i+1}$.} 
Conversely, suppose that  $(\varx_1, \ldots,  \varx_{i+1}) \lexorder (\vary_1, \ldots, \vary_{i+1})$. 
Then also $(\varx_1, \ldots,  \varx_i) \lexorder (\vary_1, \ldots, \vary_i)$, so the induction hypothesis 
implies that  $\auxd_i$ holds. We consider two cases: 
\begin{itemize}
\item If $(\varx_1, \ldots,  \varx_i) \neq (\vary_1, \ldots, \vary_i)$, 
then $\auxd_i \wedge \auxa_i$ does not hold (by the induction hypothesis), 
so $\olnot{\auxa_i}$ holds (since  $\auxd_i$ holds). 
\item If $(\varx_1, \ldots,  \varx_i) = (\vary_1, \ldots, \vary_i)$, then 
$(\varx_1, \ldots,  \varx_{i+1}) \lexorder (\vary_1, \ldots, \vary_{i+1})$ 
implies $\vary_{i+1} \geq \varx_{i+1}$. 
\end{itemize}
Hence, $\olnot{\auxa_i} \,\vee\, \vary_{i+1} \geq \varx_{i+1}$ holds, 
which implies that $\auxd_{i+1}$ holds.

\paragraph{$\auxd_{i+1} \wedge \auxa_{i+1}$ implies $(\varx_1, \ldots,  \varx_{i+1}) = (\vary_1, \ldots, \vary_{i+1})$.} 
Suppose that  $\auxd_{i+1} \wedge \auxa_{i+1}$ holds. 
Then also $\auxd_i \wedge \auxa_i$ holds (by \eqref{eq:ai} and \eqref{eq:di}),
so the induction hypothesis implies that  $(\varx_1, \ldots, \varx_i) = (\vary_1, \ldots, \vary_i)$.
It remains to show that $\varx_{i+1} = \vary_{i+1}$, which we show by showing both inequalities:
\begin{itemize}
\item $\auxa_{i+1}$ implies $\varx_{i+1} \geq \vary_{i+1}$ (by \eqref{eq:ai}). 
\item $\auxd_{i+1}$ implies $\olnot{\auxa_i} \,\vee\, \vary_{i+1} \geq \varx_{i+1}$ (by \eqref{eq:di}), 
which using that $\auxa_i$ holds implies $\vary_{i+1} \geq \varx_{i+1}$. 
\end{itemize}
Hence, we conclude that $(\varx_1, \ldots,  \varx_{i+1}) = (\vary_1, \ldots, \vary_{i+1})$.  

\paragraph{$(\varx_1, \ldots,  \varx_{i+1}) = (\vary_1, \ldots, \vary_{i+1})$ implies $\auxd_{i+1} \wedge \auxa_{i+1}$.} 
Conversely, suppose that $(\varx_1, \ldots,  \varx_{i+1}) = (\vary_1, \ldots, \vary_{i+1})$. 
Then also $(\varx_1, \ldots,  \varx_i) = (\vary_1, \ldots, \vary_i)$, 
so the induction hypothesis implies that  $\auxd_i \wedge \auxa_i$  holds. 
Moreover, we have  $\varx_{i+1} = \vary_{i+1}$. 
Then $\vary_{i+1} \geq \varx_{i+1}$ together with $\auxd_i$ implies $\auxd_{i+1}$ (by \eqref{eq:di}),
while $\varx_{i+1} \geq \vary_{i+1}$ together with $\auxa_i$ implies $\auxa_{i+1}$ (by \eqref{eq:ai}). 
We conclude that $\auxd_{i+1} \wedge \auxa_{i+1}$ holds.

This completes the induction and hence the proof.
\end{proof}

The constraints given in Lemma~\ref{lem:lex_using_aux} are not written as
pseudo-Boolean constraints, but they can be expressed using pseudo-Boolean
constraints. For this, note that the constraint $r \Leftrightarrow \constrc$, where 
$r$ is a Boolean variable and $\constrc \synteq \sum_i \coeffapp_i \varx_i \geq \degapp$
is a pseudo-Boolean constraint in normalized form, can be expressed using
the pair of pseudo-Boolean constraints
$\degapp \olnot{r} + \sum_i \coeffapp_i \varx_i \geq \degapp$ and 
$( \sum_i \coeffapp_i - \degapp + 1) r + \sum_i \coeffapp_i \olnot{\varx_i} \geq  \sum_i \coeffapp_i - \degapp + 1$. 
We~call a constraint of the form  $r \Leftrightarrow \constrc$ a \emph{reification}, 
and the process of rewriting a reification to
 a pair of pseudo-Boolean constraints \emph{expanding the reification}.

\begin{lemma}  \label{lem:lex_using_aux_PB}
We can write the constraints~\eqref{eq:a1} up to~\eqref{eq:di} as follows as pseudo-Boolean constraints:
\begin{align}
\olnot{\auxa_1} + \varx_1 + \olnot{\vary_1} &\geq 1\eqcomma \label{eq:a1.1}\\
2\auxa_1 + \olnot{\varx_1} + \vary_1 &\geq 2\eqcomma \label{eq:a1.2}\\
3\olnot{\auxa_{i+1}} + 2\auxa_i + \varx_{i+1} + \olnot{\vary_{i+1}} &\geq 3\eqcomma \label{eq:ai.1}   \\
2\auxa_{i+1} + 2\olnot{\auxa_i} + \olnot{\varx_{i+1}} + \vary_{i+1} &\geq 2\eqcomma \label{eq:ai.2}   \\
\olnot{\auxd_1} + \vary_1 + \olnot{\varx_1} &\geq 1\eqcomma \label{eq:d1.1}\\
2\auxd_1 + \olnot{\vary_1} + \varx_1 &\geq 2\eqcomma \label{eq:d1.2}\\
4\olnot{\auxd_{i+1}} + 3\auxd_i + \olnot{\auxa_i} + \vary_{i+1} + \olnot{\varx_{i+1}} &\geq 4\eqcomma \label{eq:di.1} \\
4 \auxd_{i+1} + 3 \olnot{\auxd_i} + \auxa_i + \olnot{\vary_{i+1}} + \varx_{i+1} &\geq 3\eqperiod \label{eq:di.2}
\end{align}
\end{lemma}

\begin{proof}
For constraint \eqref{eq:a1}, we rewrite $\varx_1 \geq \vary_1$ as 
$\varx_1 + \olnot{\vary_1} \geq 1$, and then expand the reification.

For constraint \eqref{eq:ai}, we note that
$\auxa_i \;\wedge\; \varx_{i+1} \geq \vary_{i+1}$ is equivalent to 
$\auxa_i \;\wedge\; \varx_{i+1} + \olnot{\vary_{i+1}} \geq 1$, 
which in turn is equivalent to 
$2\auxa_i + \varx_{i+1} + \olnot{\vary_{i+1}} \geq 3$. 
Then we expand the reification.

For constraint \eqref{eq:d1}, we rewrite $\vary_1 \geq \varx_1$ as 
$\vary_1 + \olnot{\varx_1} \geq 1$, and then expand the reification.

For constraint \eqref{eq:di}, we first note that 
$\olnot{\auxa_i} \,\vee\, \vary_{i+1} \geq \varx_{i+1}$ 
is equivalent to $\olnot{\auxa_i} \,\vee\, \vary_{i+1} + \olnot{\varx_{i+1}} \geq 1$, 
which is in turn equivalent to 
$\olnot{\auxa_i} + \vary_{i+1} + \olnot{\varx_{i+1}} \geq 1$. 
Hence, $\auxd_i \;\wedge\; (\olnot{\auxa_i} \,\vee\, \vary_{i+1} \geq \varx_{i+1})$ 
is equivalent to the pseudo-Boolean constraint 
$3\auxd_i + \olnot{\auxa_i} + \vary_{i+1} + \olnot{\varx_{i+1}} \geq 4$. 
Finally, we expand the reification.
\end{proof}

We write $\spec(\vx, \vy, \va, \vd)$ for the constraints~\eqref{eq:a1.1}
up to~\eqref{eq:di.2} in Lemma~\ref{lem:lex_using_aux_PB}. Note that $\spec$ 
has four arguments, rather than three as in Section~\ref{subsec:orderswithauxvars}, since we 
use two different symbols for the auxiliary variables to emphasize their different 
functions. We write $\ord(\vd)$ for the single constraint $\auxd_\nvar \geq 1$. 
Throughout the remainder of this appendix, we write $\preceq$ for the preorder 
with the specification $\spec(\vx, \vy, \va, \vd)$ and definition $\ord(\vd)$. 
Then Lemma~\ref{lem:lex_using_aux} shows that the preorder  $\preceq$ 
defined by $\spec$ and $\ord$ is indeed the lexicographical order. 

Next, we show how we can prove reflexivity and transitivity of this order in \veripb.

\begin{lemma} \label{lem:lex_using_aux_reflexivity}
We can prove reflexivity of the order $\preceq$ using a single RUP step, 
requiring $\bigoh{\nvar}$ propagations. 
\end{lemma}
\begin{proof}
Since proof goals are always shown by contradiction in \veripb, we assume that
$\auxd_\nvar \geq 1$ does not hold, \ie that $\olnot{\auxd_\nvar} \geq 1$ holds.
For the reflexivity proof, we have access to the specification 
$\spec(\vx, \vx, \va, \vd)$. 
Under the substitution $y_i = x_i$, the constraints \eqref{eq:d1.2} and 
\eqref{eq:di.2} simplify to  
\begin{align}
2\auxd_1 &\geq 1\eqcomma \label{eq:d1.2ref}\\
4 \auxd_{i+1} + 3 \olnot{\auxd_i} + \auxa_i &\geq 2\eqperiod \label{eq:di.2ref}
\end{align}
Then \eqref{eq:d1.2ref} propagates $\auxd_1$ to 1, after which~\eqref{eq:di.2ref}
inductively propagates $\auxd_{i+1}$ to 1 (for $1 \leq i \leq \nvar-1$). 
After deriving $\auxd_{\nvar}$, we obtain a contradiction with $\olnot{\auxd_\nvar} \geq 1$.
\end{proof}

\begin{lemma} \label{lem:lex_using_aux_trans}
We can prove transitivity of the order $\preceq$ 
using a cutting planes derivation of size $\bigoh{\nvar}$.
\end{lemma}

\begin{proof}
Let $\spec(\vx, \vy, \va, \vd)$, $\spec(\vy, \vz, \vb, \ve)$
and $\spec(\vx, \vz, \vc, \vf)$ be the specifications corresponding
to the relations 
$(\varx_1, \ldots,  \varx_\nvar) \preceq (\vary_1, \ldots, \vary_\nvar)$, 
$(\vary_1, \ldots,  \vary_\nvar) \preceq (\varz_1, \ldots, \varz_\nvar)$, and 
$(\varx_1, \ldots,  \varx_\nvar) \preceq (\varz_1, \ldots, \varz_\nvar)$, respectively. 
We also assume that $\ord(\vd)$ and $\ord(\ve)$ hold, which are the constraints
$\auxd_\nvar \geq 1$ and $\auxe_\nvar \geq 1$. 
From this, we need to show that $\ord(\vf)$ holds, 
which is the constraint $\auxf_\nvar \geq 1$. 

\paragraph{Overview.} The proof consists of two main steps: 

\begin{enumerate}
\item Deriving  inductively that $\auxd_i \geq 1$ for $1 \leq i \leq \nvar$, 
and simplifying constraints~\eqref{eq:d1.1} and~\eqref{eq:di.1} to account 
for the known values of the $\auxd_i$-variables. Similarly, deriving that 
$\auxe_i \geq 1$ for $1 \leq i \leq \nvar$, and simplifying constraints 
using known values of the $\auxe_i$-variables.
\item Deriving inductively that  $\auxf_i \geq 1$ for $1 \leq i \leq \nvar$,
for which we also inductively derive that $\auxc_i$ implies $\auxa_i$ and $\auxb_i$. 

The intuition for why $\auxc_i$ implies  $\auxa_i$ and  $\auxb_i$ is as follows. 
We know that 
$(\varx_1, \ldots,  \varx_i) \preceq (\vary_1, \ldots, \vary_i) \preceq (\varz_1, \ldots, \varz_i)$,
and $\auxc_i$ encodes that $(\varx_1, \ldots,  \varx_i) = (\varz_1, \ldots, \varz_i)$. 
However, since $ (\vary_1, \ldots, \vary_i)$ is `squeezed' in between 
$(\varx_1, \ldots,  \varx_i)$ and $(\varz_1, \ldots, \varz_i)$,
this equality means that also the equalities
$(\varx_1, \ldots,  \varx_i) = (\vary_1, \ldots, \vary_i)$ and 
$(\vary_1, \ldots,  \vary_i) = (\varz_1, \ldots, \varz_i)$ hold, 
which is in turn encoded by $\auxa_i$ and  $\auxb_i$. 
\end{enumerate}

Since we use a constant number of cutting planes steps in the inductive step 
of these inductive derivations, the total number of cutting planes steps is $\bigoh{\nvar}$. Since each constraint has bounded size, the size of the derivation is also $\bigoh{\nvar}$.

\paragraph{Deriving $\auxd_i \geq 1$ and simplifying constraints.} 
We start by deriving that all $\auxd_i$ hold and simplify the constraints
containing the  $\auxd_i$. Weakening constraint~\eqref{eq:di.1} on 
$\auxa_i$, $\vary_{i+1}$ and $\varx_i$ yields 
$4 \olnot{\auxd_{i+1}} + 3 \auxd_i \geq 1$. 
Hence, from  $\auxd_\nvar \geq 1$ we can inductively prove $\auxd_i \geq 1$
for $1 \leq i \leq n$. This in turn allows us to derive the constraints 
$\olnot{\auxa_i} + \vary_{i+1} + \olnot{\varx_{i+1}} \geq 1$
(by adding $4 \cdot (\auxd_{i+1} \geq 1)$ to~\eqref{eq:di.1} and weakening on $\auxd_i$),
and the constraint $\vary_1 + \olnot{\varx_1} \geq 1$. 
Similarly, we derive $\auxe_i \geq 1$, the constraints 
$\olnot{\auxb_i} + \varz_{i+1} + \olnot{\vary_{i+1}} \geq 1$, 
and the constraint $\varz_1 + \olnot{\vary_1} \geq 1$.

\paragraph{Deriving $\auxf_i \geq 1$ and $\auxc_i$ implies  $\auxa_i$ and  $\auxb_i$.} 
We inductively show that $\auxf_i \geq 1$,
and that $\auxc_i$ implies $\auxa_i$ and  $\auxb_i$,
\ie that $\auxa_i + \olnot{\auxc_i} \geq 1$ and $\auxb_i + \olnot{\auxc_i} \geq 1$.
For the base case, adding $\vary_1 + \olnot{\varx_1} \geq 1$ and 
$\varz_1 + \olnot{\vary_1} \geq 1$ yields $\varz_1 + \olnot{\varx_1} \geq 1$,
which in turn implies that $\auxf_1 \geq 1$. Adding the three constraints
\begin{align*}
2\auxa_1 + \olnot{\varx_1} + \vary_1 &\geq 2\eqcomma  \\
\olnot{\auxc_1} + \varx_1 + \olnot{\varz_1} &\geq 1\eqcomma \\
\varz_1 + \olnot{\vary_1} &\geq 1\eqcomma
\end{align*}
and saturating yields $\auxa_1 + \olnot{\auxc_1} \geq 1$. 
Similarly, we derive the implication $\auxb_1 + \olnot{\auxc_1} \geq 1$, 
which completes the base case. 
 
We proceed with the inductive step. We first derive that $\auxf_{i+1}$ holds. 
Adding the constraints
\begin{align*}
\olnot{\auxa_i} + \vary_{i+1} + \olnot{\varx_{i+1}} &\geq 1 \eqcomma \\ 
\olnot{\auxb_i} + \varz_{i+1} + \olnot{\vary_{i+1}} &\geq 1 \eqcomma \\
\auxa_i + \olnot{\auxc_i} &\geq 1 \eqcomma \\
\auxb_i + \olnot{\auxc_i} &\geq 1 \eqcomma
\end{align*} 
and saturating, we get $\olnot{\auxc_i} + \varz_{i+1} + \olnot{\varx_{i+1}} \geq 1$. 
Adding this constraint, $3 \cdot (\auxf_i \geq 1)$ and 
$3 \auxf_{i+1} + 3 \olnot{\auxf_i} + \auxc_i + \olnot{\varz_{i+1}} + \varx_{i+1} \geq 3$
and saturating then allows us to derive $\auxf_{i+1} \geq 1$.
 
Finally, we have to show that $\auxc_{i+1}$ implies  $\auxa_{i+1}$ and  $\auxb_{i+1}$.
First note that
$\auxc_{i+1} \Rightarrow (\auxc_i \wedge \varx_{i+1} \geq \varz_{i+1})$ 
implies the two constraints $\auxc_{i+1} \Rightarrow \auxc_i$ and 
$\auxc_{i+1} \Rightarrow (\varx_{i+1} \geq \varz_{i+1})$. Indeed, weakening 
$3\olnot{\auxc_{i+1}} + 2\auxc_i + \varx_{i+1} + \olnot{\varz_{i+1}} \geq 3$
on $\varx_{i+1}$ and $\varz_{i+1}$ and saturating yields 
$\olnot{\auxc_{i+1}} + \auxc_i \geq 1$, 
while weakening on $\auxc_i$ and saturating yields 
$\olnot{\auxc_{i+1}} + \varx_{i+1} + \olnot{\varz_{i+1}} \geq 1$. 
Adding $\olnot{\auxc_{i+1}} + \auxc_i \geq 1$ to $\auxa_i + \olnot{\auxc_i} \geq 1$ 
and $\auxb_i + \olnot{\auxc_i} \geq 1$ respectively yields 
$\auxa_i + \olnot{\auxc_{i+1}} \geq 1$ and $\auxb_i + \olnot{\auxc_{i+1}} \geq 1$.
Adding the constraints 
\begin{align*}
\olnot{\auxc_{i+1}} + \varx_{i+1} + \olnot{\varz_{i+1}} &\geq 1 \eqcomma  \\
2\auxa_{i+1} + 2\olnot{\auxa_i} + \olnot{\varx_{i+1}} + \vary_{i+1} &\geq 2 \eqcomma \\
\olnot{\auxb_i} + \varz_{i+1} + \olnot{\vary_{i+1}} &\geq 1 \eqcomma
\end{align*}
and saturating yields $\olnot{\auxc_{i+1}} + \auxa_{i+1} + \olnot{\auxa_i}  + \olnot{\auxb_i} \geq 1$. Then adding $\auxa_i + \olnot{\auxc_{i+1}} \geq 1$ and $\auxb_i + \olnot{\auxc_{i+1}} \geq 1$ and saturating yields $\olnot{\auxc_{i+1}} + \auxa_{i+1} \geq 1$. We similarly derive $\olnot{\auxc_{i+1}} + \auxb_{i+1} \geq 1$, which completes the inductive step and hence the proof.
\end{proof}

Together, this shows the following result:

\begin{theorem} \label{thm:lexorderdef}
We can define the lexicographical order over $\nvar$ variables in \veripb using
auxiliary variables and prove its transitivity and reflexivity using a \veripb 
proof of size $\bigoh{\nvar}$ that can be checked in time $\bigoh{\nvar}$.
\end{theorem}

\subsection{Deriving the Symmetry-Breaking Clauses}

Next, we show how to use our proof system in applications of the dominance
rule where the witness is a symmetry $\symmetry$ of the input formula. 
Let $\suppsymmetry = \{\varx_{\idxi_1}, \ldots, \varx_{\idxi_\suppsize}\}$ 
be the support of $\symmetry$ (where $\idxi_1 \leq \ldots \leq \idxi_\suppsize$). 
We want to derive symmetry breaking clauses encoding the lex-leader constraint
\begin{equation*}
(\varx_{\idxi_1}, \ldots, \varx_{\idxi_\suppsize}) \lexorder (\symmetry(\varx_{\idxi_1}), \ldots, \symmetry(\varx_{\idxi_\suppsize}))\eqperiod
\end{equation*}

Throughout this appendix, we assume that the lexicographical order $\preceq$ over 
$\nvar$ variables is loaded. We first show how to define a circuit 
over variables $\auxs_\subidx$ and $\auxt_\subidx$ defining 
the lexicographical order over just $\suppsymmetry$. 
Then we show how to do the actual symmetry breaking using the dominance rule. 
In particular, Lemma~\ref{lem:first_proof_goal} and 
Lemma~\ref{lem:second_proof_goal} show how we can prove 
the order proof goals 
corresponding to the dominance rule application.
Lemma~\ref{lem:symmetrybreakingclauses} then shows how to derive the 
symmetry breaking clauses. 
Together, this leads to Theorem~\ref{thm:symmetrybreakingclauses},
which shows that we can justify a symmetry~$\symmetry$ of the input formula with 
$\suppsize = |\suppsymmetry|$ using a \veripb proof of size~$\bigoh{\suppsize}$
that can be checked in time~$\bigoh{\nvar}$.

\paragraph{Defining a circuit.} 
We first define at top level (using the redundance rule) a
circuit similarly to the circuit defined in Lemma~\ref{lem:lex_using_aux}, 
using the following equations:
\begin{align}
\auxs_1 &\Leftrightarrow (\varx_{\idxi_1} \geq \symmetry(\varx_{\idxi_1}))\eqcomma \label{eq:s1}\\
\auxs_{\subidx+1} &\Leftrightarrow (\auxs_{\subidx} \;\wedge\; \varx_{\idxi_{\subidx+1}} 
  \geq \symmetry(\varx_{\idxi_{\subidx+1}}))\eqcomma \label{eq:si}   \\
\auxt_1 &\Leftrightarrow (\symmetry(\varx_{\idxi_1}) \geq \varx_{\idxi_1})\eqcomma \label{eq:t1} \\
\auxt_{\subidx+1} &\Leftrightarrow (\auxt_{\subidx} \;\wedge\; (\olnot{\auxs_{\subidx}} 
  \,\vee\, \symmetry(\varx_{\idxi_{\subidx+1}}) \geq \varx_{\idxi_{\subidx+1}}))\eqcomma \label{eq:ti}
\end{align}
where we have constraint \eqref{eq:si} for $1\leq\subidx\leq\suppsize-2$ 
and constraint \eqref{eq:ti} for $1\leq\subidx\leq\suppsize-1$.

Intuitively, the variables $\auxs_\subidx$ and $\auxt_\subidx$ have the
same role as the auxiliary variables $\auxa_{\idxi_\subidx}$ and
$\auxd_{\idxi_\subidx}$ in the specification, respectively.
As in the proof of Lemma~\ref{lem:lex_using_aux}, $\auxt_{\subidx}$ hold
 if and only if 
 $(\varx_{\idxi_1}, \ldots, \varx_{\idxi_\subidx}) \lexorder 
  (\symmetry(\varx_{\idxi_1}), \ldots, \symmetry(\varx_{\idxi_\subidx}))$, 
 while $\auxs_{\subidx} \wedge \auxt_{\subidx}$ holds if and only if 
 $(\varx_{\idxi_1}, \ldots, \varx_{\idxi_\subidx}) = 
  (\symmetry(\varx_{\idxi_1}), \ldots, \symmetry(\varx_{\idxi_\subidx}))$.

As in Lemma~\ref{lem:lex_using_aux_PB}, we can write these constraints as 
follows as pseudo-Boolean constraints:
\begin{align}
\olnot{\auxs_1} + \varx_{\idxi_1} + \olnot{\symmetry(\varx_{\idxi_1})} &\geq 1\eqcomma \label{eq:s1.1}\\
2\auxs_1 + \olnot{\varx_{\idxi_1}} + \symmetry(\varx_{\idxi_1}) &\geq 2\eqcomma \label{eq:s1.2}\\
3\olnot{\auxs_{\subidx+1}} + 2\auxs_\subidx + \varx_{\idxi_{\subidx+1}} 
 + \olnot{\symmetry(\varx_{\idxi_{\subidx+1}})} &\geq 3\eqcomma \label{eq:si.1}   \\
2\auxs_{\subidx+1} + 2\olnot{\auxs_\subidx} + \olnot{\varx_{\idxi_{\subidx+1}}} 
 + \symmetry(\varx_{\idxi_{\subidx+1}}) &\geq 2\eqcomma \label{eq:si.2}   \\
\olnot{\auxt_1} + \symmetry(\varx_{\idxi_1}) + \olnot{\varx_{\idxi_1}} &\geq 1\eqcomma \label{eq:t1.1}\\
2\auxt_1 + \olnot{\symmetry(\varx_{\idxi_1})} + \varx_{\idxi_1} &\geq 2\eqcomma \label{eq:t1.2}\\
4\olnot{\auxt_{\subidx+1}} + 3\auxt_\subidx + \olnot{\auxs_\subidx} 
 + \symmetry(\varx_{\idxi_{\subidx+1}}) + \olnot{\varx_{\idxi_{\subidx+1}}} &\geq 4\eqcomma \label{eq:ti.1} \\
3 \auxt_{\subidx+1} + 3 \olnot{\auxt_\subidx} + \auxs_\subidx 
 + \olnot{\symmetry(\varx_{\idxi_{\subidx+1}})} + \varx_{\idxi_{\subidx+1}} &\geq 3\eqperiod\label{eq:ti.2}
\end{align}
In total, this circuit consists of $4 \suppsize - 2$ pseudo-Boolean constraints, 
and can be derived using the redundance rule in time $\bigoh{1}$ per constraint. 
Note that for this the autoproving of the order proof goal of the redundance 
rule and the explicit loading of the 
specification discussed in Section~\ref{sec:implementationdetails} are essential 
to avoid a factor $\nvar$ overhead. Namely, these applications of the redundance
rule witness over fresh variables only, so in particular not over the variables
over which the order is loaded. Hence, the order proof goal follows directly from
the reflexivity of the order, and these checker improvements allow us to detect 
this in time $\bigoh{1}$, instead of wasting time $\bigoh{\nvar}$ on loading 
the specification and repeating the reflexivity proof. 

\paragraph{Overview: the symmetry breaking clauses.} 
We first state the symmetry breaking clauses:
\begin{align}
\auxs_1 +   \olnot{\varx_{\idxi_1}}  &\geq 1\eqcomma \label{eq:symbreak_s1.1} \\
\auxs_{\subidx+1} + \olnot{\auxs_\subidx} + \olnot{\varx_{\idxi_{\subidx+1}}} &\geq 1\eqcomma \label{eq:symbreak_si.1} \\
\auxs_1 + \symmetry(\varx_{\idxi_1}) &\geq 1\eqcomma \label{eq:symbreak_s1.2} \\
\auxs_{\subidx+1} + \olnot{\auxs_\subidx} + \symmetry(\varx_{\idxi_{\subidx+1}}) &\geq 1\eqcomma \label{eq:symbreak_si.2} \\
\symmetry(\varx_{\idxi_1}) + \olnot{\varx_{\idxi_1}} &\geq 1 \eqcomma  \label{eq:symbreak_t1} \\
\olnot{\auxs_{\subidx}} + \symmetry(\varx_{\idxi_{\subidx+1}}) + \olnot{\varx_{\idxi_{\subidx+1}}}&\geq 1\eqperiod  \label{eq:symbreak_ti}
\end{align}
Note that a constraint whose coefficients and whose right hand side is 1
corresponds to a clause: for example, the constraint 
$\auxs_1 + \olnot{\varx_{\idxi_1}}  \geq 1$ is equivalent to the clause 
$\auxs_1 \vee \olnot{\varx_{\idxi_1}}$.  

In the symmetry breaking clauses, the variable $\auxs_\subidx$ corresponds
to the equality $(\varx_{\idxi_1}, \ldots, \varx_{\idxi_\subidx}) = 
 (\symmetry(\varx_{\idxi_1}), \ldots, \symmetry(\varx_{\idxi_\subidx}))$. 
Clauses \eqref{eq:symbreak_s1.1} up to \eqref{eq:symbreak_si.2} ensure that 
$\auxs_\subidx$ propagates to true if this equality holds, 
while \eqref{eq:symbreak_t1} and \eqref{eq:symbreak_ti} provide the symmetry
breaking constraint $\symmetry(\varx_{\idxi_1}) \geq \varx_{\idxi_1}$ and that 
$\symmetry(\varx_{\idxi_{\subidx+1}}) \geq \varx_{\idxi_{\subidx+1}}$
should hold if $(\varx_{\idxi_1}, \ldots, \varx_{\idxi_\subidx}) = 
 (\symmetry(\varx_{\idxi_1}), \ldots, \symmetry(\varx_{\idxi_\subidx}))$.

Intuitively, to break a symmetry we will show using dominance that 
$\auxt_\suppsize$ and hence all $\auxt_\subidx$ hold, and then we 
derive these clauses from the circuit.

\paragraph{Deriving the symmetry breaking clauses.} 
Using the dominance rule, we now derive that $\auxt_{\suppsize} \geq 1$. 
Let $\constrc \synteq \auxt_{\suppsize} \geq 1$. 
To apply the dominance rule, we have to show that
\begin{align}
\core &\cup \der \cup \{\neg \constrc\} \cup \spec(\vx\uh_\symmetry, \vx, \va, \vd)
        \vdash \core\uh_\symmetry  \cup \ord(\vd)\eqcomma \\
\core &\cup \der \cup \{\neg \constrc\} \cup \spec(\vx, \vx\uh_\symmetry, \va, \vd) 
      \cup \ord(\vd) \vdash \perp \eqcomma 
\end{align}

Since the witness $\symmetry$ is a symmetry of the formula, 
the proof goals corresponding to $\core \vdash \core\uh_\symmetry$ follow
directly from $\core\uh_\symmetry = \core$. 
The next two lemmas explain how to prove the other two proof goals.

Note that within the subproof of the dominance rule, we can use the negation
$\neg \constrc \synteq \olnot{\auxt_{\suppsize}} \geq 1$, which intuitively means that 
$(\varx_{\idxi_1}, \ldots, \varx_{\idxi_\suppsize}) \not\lexorder 
 (\symmetry(\varx_{\idxi_1}), \ldots, \symmetry(\varx_{\idxi_\suppsize}))$. 

\begin{lemma} \label{lem:first_proof_goal}
Assume that the lexicographical order $\preceq$ over 
$\nvar$ variables is loaded.
Let $\symmetry$ be a symmetry of the input formula and let $\suppsize = |\suppsymmetry|$. 
Then we can show the proof goal
\begin{equation*}
\core \cup \der \cup \{\neg \constrc\} \cup \spec(\vx\uh_\symmetry, \vx, \va, \vd)
        \vdash \ord(\vd)\eqcomma \\
\end{equation*}
using $\bigoh{\suppsize}$ RUP steps and cutting planes steps, where the RUP steps 
require $\bigoh{\nvar}$ propagations in total.
\end{lemma}

\begin{proof}
Note that $\ord(\vd)$ is the single constraint $\auxd_\nvar \geq 1$.
Since proof goals are always shown by contradiction, we assume that $\olnot{\auxd_\nvar} \geq 1$ holds. 
This means that $(\symmetry(\varx_{\idxi_1}), \ldots, \symmetry(\varx_{\idxi_\suppsize}))  
\not\lexorder   (\varx_{\idxi_1}, \ldots, \varx_{\idxi_\suppsize})$, 
but since $\lexorder$ is a complete order and we also have 
 $(\varx_{\idxi_1}, \ldots, \varx_{\idxi_\suppsize}) \not\lexorder 
  (\symmetry(\varx_{\idxi_1}), \ldots, \symmetry(\varx_{\idxi_\suppsize}))$ 
via the constraint $\neg \constrc \synteq \olnot{\auxt_{\suppsize}} \geq 1$, 
we should be able to derive contradiction.

\paragraph{Overview.} The proof consists of four main steps:
\begin{enumerate}
\item Rewriting the specification $\spec(\varx\uh_{\symmetry}, \varx, \auxa, \auxd)$
to only $\bigoh{\suppsize}$ constraints involving  auxiliary variables 
$\auxa_{\idxi_\subidx}$ and $\auxd_{\idxi_\subidx}$ and variables from $\suppsymmetry$ only.
\item Showing that $\auxs_{\subidx} \Rightarrow \auxd_{\idxi_{\subidx}}$ and 
$\auxa_{\idxi_{\subidx}} \Rightarrow \auxt_{\subidx}$.
\item Showing that $\auxt_{\subidx+1} + \olnot{\auxt_{\subidx}} + \auxd_{\idxi_{\subidx}} \geq 1$ 
and $\auxd_{\idxi_{\subidx+1}} + \olnot{\auxd_{\idxi_{\subidx}}} + \auxt_{\subidx} \geq 1$. 
\item Showing $\auxd_{\idxi_{\subidx}} + \auxt_{\subidx} \geq 1$ 
and  $\auxt_{\subidx+1} + \auxd_{\idxi_{\subidx}} \geq 1$ 
and  $\auxd_{\idxi_{\subidx+1}} + \auxt_{\subidx} \geq 1$.
\end{enumerate}

The first of these steps requires $\bigoh{\suppsize}$ RUP steps for which 
$\bigoh{\nvar}$ propagations are needed in total, and $\bigoh{\suppsize}$ cutting planes steps. 
The remaining steps require  $\bigoh{\suppsize}$ RUP steps for which 
$\bigoh{\suppsize}$ propagations are needed.

\paragraph{Rewriting the specification constraints.} 
Write $\suppsymmetryidxs = \{\idxi_1, \dots, \idxi_\suppsize\}$. 
There are two cases which are slightly different, depending on 
whether $1 \in \suppsymmetryidxs$ or not. 
Here we consider the case $1 \not\in \suppsymmetryidxs$. 
Under the substitution $\symmetry$, the specification constraints 
$\spec(\varx\uh_{\symmetry}, \varx, \auxa, \auxd)$ are (semantically) given by 
\begin{align}
\auxa_1 &\geq 1\eqcomma \label{eq:s1pf1}\\
\auxa_{\idxi} &\Leftrightarrow (\auxa_{\idxi-1} \;\wedge\;  \symmetry(\varx_{\idxi}) \geq \varx_{\idxi})\eqcomma && (\idxi \in \suppsymmetryidxs) \label{eq:sipf1}   \\
\auxa_{\idxi} &\Leftrightarrow \auxa_{\idxi-1}\eqcomma && (\idxi \not\in \suppsymmetryidxs) \label{eq:sipf1r}   \\
\auxd_1 &\geq 1\eqcomma \label{eq:t1pf1} \\
\auxd_{\idxi} &\Leftrightarrow (\auxd_{\idxi-1} \wedge (\olnot{\auxa_{\idxi-1}} \vee \varx_{\idxi} \!\geq\! \symmetry(\varx_{\idxi})))\eqcomma\!\!&& (\idxi \in \suppsymmetryidxs)\!\!  \label{eq:tipf1} \\ 
\auxd_{\idxi} &\Leftrightarrow \auxd_{\idxi-1}\eqperiod && (\idxi \not\in \suppsymmetryidxs)  \label{eq:tipf1r}
\end{align}
Using the constraints~\eqref{eq:s1pf1} and \eqref{eq:sipf1r}, it follows 
by RUP that $ \auxa_{\idxi_{1}-1} \geq 1$ and 
$\auxa_{\idxi_{\subidx}-1} \Leftrightarrow \auxa_{\idxi_{\subidx-1}}$ for $2 \leq \subidx \leq \suppsize$.
Similarly, using the constraints~\eqref{eq:t1pf1} and \eqref{eq:tipf1r}, 
it follows by RUP that $ \auxd_{\idxi_{1}-1} \geq 1$ and 
$\auxd_{\idxi_{\subidx}-1} \Leftrightarrow \auxd_{\idxi_{\subidx-1}}$ for $2 \leq \subidx \leq \suppsize$. 

Using these constraints, we can (using cutting planes steps) rewrite 
constraints~\eqref{eq:sipf1} and \eqref{eq:tipf1} as 
\begin{align}
\auxa_{\idxi_1} &\Leftrightarrow (\symmetry(\varx_{\idxi_1}) \geq \varx_{\idxi_1})\eqcomma \label{eq:s1pf1supp}\\
\auxa_{\idxi_{\subidx+1}} &\Leftrightarrow (\auxa_{\idxi_{\subidx}} \;\wedge\; \symmetry(\varx_{\idxi_{\subidx+1}})\geq \varx_{\idxi_{\subidx+1}})\eqcomma\label{eq:sipf1supp}   \\
\auxd_{\idxi_1} &\Leftrightarrow (\varx_{\idxi_1} \geq \symmetry(\varx_{\idxi_1}))\eqcomma \label{eq:t1pf1supp} \\
\auxd_{\idxi_{\subidx+1}} &\Leftrightarrow (\auxd_{\idxi_{\subidx}} \;\wedge\; (\olnot{\auxa_{\idxi_{\subidx}}} \,\vee\, \varx_{\idxi_{\subidx+1}} \geq \symmetry(\varx_{\idxi_{\subidx+1}})))\eqperiod \label{eq:tipf1supp}
\end{align}

\paragraph{Showing that $\auxs_{\subidx} \Rightarrow \auxd_{\idxi_{\subidx}}$ and $\auxa_{\idxi_{\subidx}} \Rightarrow \auxt_{\subidx}$.}
We inductively derive using RUP that $\auxs_{\subidx} \Rightarrow \auxd_{\idxi_{\subidx}}$, \ie  $\olnot{\auxs_{\subidx}} + \auxd_{\idxi_{\subidx}} \geq 1$. When showing this by RUP, the negation propagates $\auxs_{\subidx}$ and $\olnot{\auxd_{\idxi_{\subidx}}}$.

For the base case, note that $\olnot{\auxd_{\idxi_1}}$ propagates that the 
inequality $\varx_{\idxi_1} \geq \symmetry(\varx_{\idxi_1})$ is false 
(\ie that $\olnot{\varx_{\idxi_1}}$ and  $\symmetry(\varx_{\idxi_1})$ hold), 
which yields together with $\auxs_1$ a conflict in \eqref{eq:s1.1}.
 
For the inductive step, note that $\auxs_{\subidx+1}$ propagates 
$\auxs_{\subidx}$ using~\eqref{eq:si.1}, which by the induction hypothesis
propagates $\auxd_{\idxi_{\subidx}}$. Then $\olnot{\auxd_{\idxi_{\subidx+1}}}$
propagates that 
$\olnot{\auxa_{\idxi_{\subidx}}} \,\vee\, (\varx_{\idxi_{\subidx+1}} \geq \symmetry(\varx_{\idxi_{\subidx+1}}))$
is false, which in particular propagates that 
$\varx_{\idxi_{\subidx+1}} \geq \symmetry(\varx_{\idxi_{\subidx+1}})$
is false (\ie that $\olnot{\varx_{\idxi_{\subidx+1}}}$ and $\symmetry(\varx_{\idxi_{\subidx+1}})$ hold), 
which yields together with $\auxs_{\subidx+1}$ a conflict in \eqref{eq:si.1}.

In the same way, we derive using RUP that $\auxa_{\idxi_{\subidx}} \Rightarrow \auxt_{\subidx}$.

\paragraph{Showing that $\auxt_{\subidx+1} + \olnot{\auxt_{\subidx}} + \auxd_{\idxi_{\subidx}} \geq 1$ and $\auxd_{\idxi_{\subidx+1}} + \olnot{\auxd_{\idxi_{\subidx}}} + \auxt_{\subidx} \geq 1$.}
Next, we derive using RUP that $\auxt_{\subidx+1} + \olnot{\auxt_{\subidx}} + \auxd_{\idxi_{\subidx}} \geq 1$. 
Its negation propagates $\olnot{\auxt_{\subidx+1}}$, $\auxt_{\subidx}$, and $\olnot{\auxd_{\idxi_{\subidx}}}$. 
This propagates $ \auxs_\subidx$ using \eqref{eq:ti.2}. But then we have
 $\auxs_\subidx$ and $\olnot{\auxd_{\idxi_{\subidx}}}$, which is a 
 conflict in $\auxs_{\subidx} \Rightarrow \auxd_{\idxi_{\subidx}}$. 
In the same way, we derive using RUP that 
$\auxd_{\idxi_{\subidx+1}} + \olnot{\auxd_{\idxi_{\subidx}}} + \auxt_{\subidx} \geq 1$.

\paragraph{Showing $\auxd_{\idxi_{\subidx}} + \auxt_{\subidx} \geq 1$ 
and $\auxt_{\subidx+1} + \auxd_{\idxi_{\subidx}} \geq 1$ and  
$\auxd_{\idxi_{\subidx+1}} + \auxt_{\subidx} \geq 1$.}
Next, we derive using RUP inductively that $\auxd_{\idxi_{\subidx}} + \auxt_{\subidx} \geq 1$
and $\auxt_{\subidx+1} + \auxd_{\idxi_{\subidx}} \geq 1$ and 
$\auxd_{\idxi_{\subidx+1}} + \auxt_{\subidx} \geq 1$. 
The negated constraint propagates $\olnot{\auxd_{\idxi_{\subidx}}}$ and 
$\olnot{\auxt_{\subidx}}$. For the base case, note that $\olnot{\auxd_{\idxi_1}}$
propagates that the inequality $\varx_{\idxi_1} \geq \symmetry(\varx_{\idxi_1})$ 
is false, which yields together with $\olnot{\auxt_1}$ a conflict in \eqref{eq:t1.2}.

For the inductive step, we first show how to derive 
$\auxt_{\subidx+1} + \auxd_{\idxi_{\subidx}} \geq 1$ 
and  $\auxd_{\idxi_{\subidx+1}} + \auxt_{\subidx} \geq 1$ using 
$\auxd_{\idxi_{\subidx}} + \auxt_{\subidx} \geq 1$. 
The negation of  $\auxt_{\subidx+1} + \auxd_{\idxi_{\subidx}} \geq 1$ 
propagates $\olnot{\auxt_{\subidx+1}}$ and  $\olnot{\auxd_{\idxi_{\subidx}}}$,
which propagates $\auxt_{\subidx}$ using  
$\auxt_{\subidx} + \auxd_{\idxi_{\subidx}} \geq 1$,
which then yields a conflict in the constraint 
$\auxt_{\subidx+1} + \olnot{\auxt_{\subidx}} + \auxd_{\idxi_{\subidx}} \geq 1$. 
 
The derivation of $\auxd_{\idxi_{\subidx+1}} + \auxt_{\subidx} \geq 1$ is similar. 
 
Next, we show how to derive $\auxd_{\idxi_{\subidx+1}} + \auxt_{\subidx+1} \geq 1$.
The negation of $\auxd_{\idxi_{\subidx+1}} + \auxt_{\subidx+1} \geq 1$ propagates
$\olnot{\auxt_{\subidx+1}}$ and  $\olnot{\auxd_{\idxi_{\subidx+1}}}$, 
which propagates $\auxt_{\subidx}$ and  $\auxd_{\idxi_{\subidx}}$ using 
$\auxt_{\subidx+1} + \auxd_{\idxi_{\subidx}} \geq 1$ and 
$\auxd_{\idxi_{\subidx+1}} + \auxt_{\subidx} \geq 1$. 
Then $\olnot{\auxd_{\idxi_{\subidx+1}}}$ and $\auxd_{\idxi_{\subidx}}$
propagate that $\varx_{\idxi_{\subidx+1}} \geq \symmetry(\varx_{\idxi_{\subidx+1}})$
is false, but $\olnot{\auxt_{\subidx+1}}$ and $\auxt_{\subidx}$ propagate that
$\symmetry(\varx_{\idxi_{\subidx+1}}) \geq \varx_{\idxi_{\subidx+1}}$ is false,
which is a conflict. 
Hence, we derive $\auxd_{\idxi_{\subidx+1}} + \auxt_{\subidx+1} \geq 1$ by RUP.

Finally, using $\auxd_{\idxi_{\suppsize}} + \auxt_{\suppsize} \geq 1$, 
constraints \eqref{eq:tipf1r} which yield $\auxd_{\idxi_{\suppsize}} \Leftrightarrow \auxd_{\nvar}$, 
and the two constraints  $\olnot{\auxd_\nvar} \geq 1$ and $\olnot{\auxt_{\suppsize}} \geq 1$, 
we propagate to conflict, which completes the proof.
\end{proof}

\begin{proof}[Alternative derivation using cutting planes steps]
For step 2, 3 and 4 in the proof of Lemma~\ref{lem:first_proof_goal}, we can also provide a derivation using cutting planes steps instead.

\paragraph{Showing that $\auxs_{\subidx} \Rightarrow \auxd_{\idxi_{\subidx}}$ and $\auxa_{\idxi_{\subidx}} \Rightarrow \auxt_{\subidx}$.}
We show this by induction. For the base case, adding constraint~\eqref{eq:s1.1}, \ie 
$\olnot{\auxs_1} + \varx_{\idxi_1} + \olnot{\symmetry(\varx_{\idxi_1})}  \geq 1$,
and
$2\auxd_{\idxi_1} + \olnot{\varx_{\idxi_1}} + \symmetry(\varx_{\idxi_1}) \geq 2$ (from~\eqref{eq:t1pf1supp}),
and saturating yields
$\olnot{\auxs_1} + \auxd_{\idxi_1} \geq 1$. 

For the inductive step, we first add 
\begin{align*}
3 &\cdot \big(3\olnot{\auxs_{\subidx+1}} + 2\auxs_\subidx + \varx_{\idxi_{\subidx+1}} 
 + \olnot{\symmetry(\varx_{\idxi_{\subidx+1}})} \geq 3\big)\eqcomma \\
2 &\cdot \big(3 \auxd_{\idxi_{\subidx+1}} + 3 \olnot{\auxd_{\idxi_{\subidx}}} + \auxa_{\idxi_{\subidx}}
 + \olnot{\symmetry(\varx_{\idxi_{\subidx+1}})} + \varx_{\idxi_{\subidx+1}} \geq 3\big)\eqcomma
\end{align*}
and then weaken on $\auxa_{\idxi_{\subidx}}$,  $\symmetry(\varx_{\idxi_{\subidx+1}})$, 
and $\varx_{\idxi_{\subidx+1}}$, which yields 
\begin{equation*}
9\olnot{\auxs_{\subidx+1}} + 6\auxs_\subidx + 6 \auxd_{\idxi_{\subidx+1}} + 6 \olnot{\auxd_{\idxi_{\subidx}}} \geq 7 \eqperiod
\end{equation*}
Then we add $6 \cdot (\olnot{\auxs_\subidx} + \auxd_{\idxi_{\subidx}} \geq 1)$ and saturate,
which finally yields the constraint $\olnot{\auxs_{\subidx+1}} + \auxd_{\idxi_{\subidx+1}} \geq 1$

The derivation of $\olnot{\auxa_{\idxi_{\subidx}}} + \auxt_{\subidx} \geq 1$ is similar.

\paragraph{Showing that $\auxt_{\subidx+1} + \olnot{\auxt_{\subidx}} + \auxd_{\idxi_{\subidx}} \geq 1$ and $\auxd_{\idxi_{\subidx+1}} + \olnot{\auxd_{\idxi_{\subidx}}} + \auxt_{\subidx} \geq 1$.}

To show the constraint $\auxt_{\subidx+1} + \olnot{\auxt_{\subidx}} + \auxd_{\idxi_{\subidx}} \geq 1$, we start from constraint~\eqref{eq:ti.2}, \ie
\begin{equation*}
3 \auxt_{\subidx+1} + 3 \olnot{\auxt_\subidx} + \auxs_\subidx 
 + \olnot{\symmetry(\varx_{\idxi_{\subidx+1}})} + \varx_{\idxi_{\subidx+1}} \geq 3\eqcomma
\end{equation*}
add $\olnot{\auxs_\subidx} + \auxd_{\idxi_{\subidx}} \geq 1$, weaken on $\symmetry(\varx_{\idxi_{\subidx+1}})$ and $\varx_{\idxi_{\subidx+1}}$, and then saturate. 
The derivation of  $\auxd_{\idxi_{\subidx+1}} + \olnot{\auxd_{\idxi_{\subidx}}} + \auxt_{\subidx} \geq 1$ is similar.

\paragraph{Showing $\auxd_{\idxi_{\subidx}} + \auxt_{\subidx} \geq 1$ and 
$\auxt_{\subidx+1} + \auxd_{\idxi_{\subidx}} \geq 1$ and  $\auxd_{\idxi_{\subidx+1}} + \auxt_{\subidx} \geq 1$.}

We show this by induction.  For the base case, adding constraint~\eqref{eq:t1.2}, \ie 
$2\auxt_{\idxi_1} + \varx_{\idxi_1} + \olnot{\symmetry(\varx_{\idxi_1})} \geq 2$,
and
$2\auxd_{\idxi_1} + \olnot{\varx_{\idxi_1}} + \symmetry(\varx_{\idxi_1}) \geq 2$,
and dividing by 2 yields
$\olnot{\auxs_1} + \auxd_{\idxi_1} \geq 1$. 

For the inductive step, note that adding
$\auxd_{\idxi_{\subidx}} + \auxt_{\subidx} \geq 1$ 
to  $\auxt_{\subidx+1} + \olnot{\auxt_{\subidx}} + \auxd_{\idxi_{\subidx}} \geq 1$ 
and saturating yields $\auxt_{\subidx+1} + \auxd_{\idxi_{\subidx}} \geq 1$. 

Similarly, adding $\auxd_{\idxi_{\subidx}} + \auxt_{\subidx} \geq 1$ to 
$\auxd_{\idxi_{\subidx+1}} + \olnot{\auxd_{\idxi_{\subidx}}} + \auxt_{\subidx} \geq 1$
and saturating yields $\auxd_{\idxi_{\subidx+1}} + \auxt_{\subidx} \geq 1$.

To derive $\auxd_{\idxi_{\subidx+1}} + \auxt_{\subidx+1} \geq 1$, we start by adding
\begin{align*}
 3 \auxt_{\subidx+1} + 3 \olnot{\auxt_\subidx} + \auxs_\subidx 
 + \olnot{\symmetry(\varx_{\idxi_{\subidx+1}})} + \varx_{\idxi_{\subidx+1}} &\geq 3\eqcomma \\
 3 \auxd_{\idxi_{\subidx+1}} + 3 \olnot{\auxd_{\idxi_{\subidx}}} + \auxa_{\idxi_{\subidx}}
 + \symmetry(\varx_{\idxi_{\subidx+1}}) + \olnot{\varx_{\idxi_{\subidx+1}}} &\geq 3\eqcomma
\end{align*}
and weakening on $\auxa_{\idxi_{\subidx}}$ and $\auxs_\subidx$, which yields 
\begin{equation*}
3 \auxt_{\subidx+1} + 3 \olnot{\auxt_\subidx} +  3 \auxd_{\idxi_{\subidx+1}} 
  + 3 \olnot{\auxd_{\idxi_{\subidx}}} \geq 2 \eqperiod
\end{equation*}
Then adding $3 \cdot \big(\auxt_{\subidx+1} + \auxd_{\idxi_{\subidx}} \geq 1\big)$ and 
$3 \cdot \big(\auxd_{\idxi_{\subidx+1}} + \auxt_{\subidx} \geq 1\big)$
and saturating yields $\auxd_{\idxi_{\subidx+1}} + \auxt_{\subidx+1} \geq 1$. 
\end{proof}

We proceed by showing the second proof goal. 
\begin{lemma} \label{lem:second_proof_goal}
Assume that  the lexicographical order $\preceq$ over 
$\nvar$ variables is loaded.
Let $\symmetry$ be a symmetry of the input formula and let $\suppsize = |\suppsymmetry|$. 
Then we can show the proof goal
\begin{equation*}
\core \cup \der \cup \{\neg C\} \cup \spec(\vx, \vx\uh_\symmetry, \va, \vd) \cup \ord(\vd) \vdash \perp \eqcomma
\end{equation*}
using $\bigoh{\suppsize}$ RUP steps and cutting planes steps, 
where the RUP steps require $\bigoh{\nvar}$ propagations in total.
\end{lemma}

\begin{proof} 
Note that $\neg \constrc \synteq \olnot{\auxt_{\suppsize}} \geq 1$ still has 
the interpretation that 
$(\varx_{\idxi_1}, \ldots, \varx_{\idxi_\suppsize}) 
  \not\lexorder (\symmetry(\varx_{\idxi_1}), \ldots, \symmetry(\varx_{\idxi_\suppsize}))$, 
while the premises $\spec(\vx, \vx\uh_\symmetry, \va, \vd) \cup \ord(\vd)$ have
the interpretation that $(\varx_{\idxi_1}, \ldots, \varx_{\idxi_\suppsize}) 
  \lexorder (\symmetry(\varx_{\idxi_1}), \ldots, \symmetry(\varx_{\idxi_\suppsize}))$,
so we should be able to derive contradiction.

The idea of the proof is to show that $\auxa_{\idxi_\subidx}$ and 
$\auxs_\subidx$ have the same interpretation and that $\auxd_{\idxi_\subidx}$
and $\auxt_\subidx$ have the same interpretation. This then gives a 
contradiction since we have the premises $\auxd_{\nvar} \geq 1$ 
but $\olnot{\auxt_{\suppsize}} \geq 1$.

\paragraph{Overview.}  The proof consists of four main steps:
\begin{enumerate}
\item Rewriting the specification $\spec(\varx, \varx\uh_{\symmetry}, \auxa, \auxd)$
to only $\bigoh{\suppsize}$ constraints involving  auxiliary variables 
$\auxa_{\idxi_\subidx}$ and $\auxd_{\idxi_\subidx}$ and variables from $\suppsymmetry$ only.
\item Showing that $\auxd_{\idxi_{\subidx}} \geq 1$.
\item Showing that $\auxa_{\idxi_{\subidx}} + \olnot{\auxs_{\subidx}} \geq 1$.
\item Showing that $\auxt_{\subidx} \geq 1$.
\end{enumerate}

The first of these steps requires $\bigoh{\suppsize}$ RUP steps for which 
$\bigoh{\nvar}$ propagations are needed in total, and $\bigoh{\suppsize}$ cutting planes steps. 
The remaining steps require  $\bigoh{\suppsize}$ RUP steps for which 
$\bigoh{\suppsize}$ propagations are needed.

\paragraph{Rewriting the specification constraints.} 
Write $\suppsymmetryidxs = \{\idxi_1, \dots, \idxi_\suppsize\}$. 
We again consider the case $1 \not\in \suppsymmetryidxs$. 
Under the substitution $\symmetry$, the specification constraints $\spec(\varx, \varx\uh_{\symmetry}, \auxa, \auxd)$  are (semantically) given by 
\begin{align}
\auxa_1 &\geq 1\eqcomma \label{eq:s1pf2}\\
\auxa_{\idxi} &\Leftrightarrow (\auxa_{\idxi-1} \wedge \varx_{\idxi} \geq \symmetry(\varx_{\idxi}))\eqcomma && (\idxi \in \suppsymmetryidxs) \label{eq:sipf2}   \\
\auxa_{\idxi} &\Leftrightarrow \auxa_{\idxi-1}\eqcomma && (\idxi \not\in \suppsymmetryidxs) \label{eq:sipf2r}   \\
\auxd_1 &\geq 1\eqcomma \label{eq:t1pf2} \\
\auxd_{\idxi} &\Leftrightarrow (\auxd_{\idxi-1} \wedge (\olnot{\auxa_{\idxi-1}} \vee \symmetry(\varx_{\idxi})\!\geq\! \varx_{\idxi}))\eqcomma\!\! && (\idxi \in \suppsymmetryidxs) \!\!\label{eq:tipf2} \\ 
\auxd_{\idxi} &\Leftrightarrow \auxd_{\idxi-1}\eqperiod && (\idxi \not\in \suppsymmetryidxs)  \label{eq:tipf2r}
\end{align}
In the same way as in Lemma~\ref{lem:first_proof_goal}, we can rewrite this to 
\begin{align}
\auxa_{\idxi_1} &\Leftrightarrow (\varx_{\idxi_1} \geq \symmetry(\varx_{\idxi_1}))\eqcomma \label{eq:s1pf2supp}\\
\auxa_{\idxi_{\subidx+1}} &\Leftrightarrow (\auxa_{\idxi_{\subidx}} \;\wedge\; \varx_{\idxi_{\subidx+1}} \geq \symmetry(\varx_{\idxi_{\subidx+1}}) )\eqcomma\label{eq:sipf2supp}   \\
\auxd_{\idxi_1} &\Leftrightarrow (\symmetry(\varx_{\idxi_1}) \geq \varx_{\idxi_1})\eqcomma \label{eq:t1pf2supp} \\
\auxd_{\idxi_{\subidx+1}} &\Leftrightarrow (\auxd_{\idxi_{\subidx}} \;\wedge\; (\olnot{\auxa_{\idxi_{\subidx}}} \,\vee\, \symmetry(\varx_{\idxi_{\subidx+1}}) \geq \varx_{\idxi_{\subidx+1}} ))\eqperiod \label{eq:tipf2supp}
\end{align}
We also show by RUP using \eqref{eq:tipf2r} that $\auxd_{\idxi_{\suppsize}} \Leftrightarrow \auxd_{\nvar}$.

When replacing $\auxa_{\idxi_\subidx}$ with $\auxs_\subidx$ and  $\auxd_{\idxi_\subidx}$ with $\auxt_\subidx$, constraints  \eqref{eq:s1pf2supp} up to \eqref{eq:tipf2supp} exactly match constraints \eqref{eq:s1} up to \eqref{eq:ti}.

\paragraph{Showing that $\auxd_{\idxi_{\subidx}} \geq 1$.} 
Since we have the constraint \mbox{$\auxd_{\nvar} \geq 1$}, 
first $\auxd_{\idxi_{\suppsize}} \Leftrightarrow \auxd_{\nvar}$ propagates that 
$\auxd_{\idxi_{\suppsize}} \geq 1$, and then \eqref{eq:tipf2supp} 
inductively propagates $\auxd_{\idxi_{\subidx}} \geq 1$ 
for all $1 \leq \subidx \leq \suppsize$, starting with the high indices. 

\paragraph{Showing that $\auxa_{\idxi_{\subidx}} + \olnot{\auxs_{\subidx}} \geq 1$.} 
We show this inductively. The negation of 
$\auxa_{\idxi_{\subidx}} + \olnot{\auxs_{\subidx}} \geq 1$ 
propagates $\olnot{\auxa_{\idxi_{\subidx}}}$ and $\auxs_{\subidx}$. 

For the base case, we note that $\olnot{\auxa_{\idxi_{1}}}$ 
propagates that $\varx_{\idxi_1} \geq \symmetry(\varx_{\idxi_1})$
is false (\ie that $\olnot{\varx_{\idxi_1}}$ and $\symmetry(\varx_{\idxi_1})$ hold)
using \eqref{eq:s1pf2supp}, 
which then yields together with $\auxs_1$ a conflict in \eqref{eq:s1}.

For the inductive step, we note that  $\auxs_{\subidx+1}$ propagates 
$\auxs_{\subidx}$ using \eqref{eq:si}, which by the induction hypothesis 
propagates $\auxa_{\idxi_{\subidx}}$. Together with $\olnot{\auxa_{\idxi_{\subidx+1}}}$
and \eqref{eq:sipf2supp} this then propagates the inequality 
$\varx_{\idxi_{\subidx+1}} \geq \symmetry(\varx_{\idxi_{\subidx+1}})$ to false, 
which then yields together with $\auxs_{\subidx+1}$  a conflict in \eqref{eq:si}. 

\paragraph{Showing that $\auxt_{\subidx} \geq 1$.}   We show this inductively. 
For the base case, the negation  $\olnot{\auxt_1} \geq 1$ propagates the inequality $\symmetry(\varx_{\idxi_1}) \geq \varx_{\idxi_1}$ to false, which the propagates $\auxd_{\idxi_1}$ to false, which is a contradiction. 

For the inductive step, the negation  $\olnot{\auxt_{\subidx+1}} \geq 1$ propagates together with  $\auxt_{\subidx} \geq 1$ using \eqref{eq:ti} that $\olnot{\auxs_{\subidx}} \vee (\symmetry(\varx_{\idxi_{\subidx+1}}) \geq \varx_{\idxi_{\subidx+1}})$ is false, \ie that $\auxs_{\subidx}$, $\olnot{\symmetry(\varx_{\idxi_{\subidx+1}})}$ and $\varx_{\idxi_{\subidx+1}}$. Then $\auxs_{\subidx}$ propagates $\auxa_{\idxi_{\subidx}}$, which together means that $\olnot{\auxa_{\idxi_{\subidx}}} \,\vee\, (\symmetry(\varx_{\idxi_{\subidx+1}}) \geq \varx_{\idxi_{\subidx+1}})$ is false. This is a conflict in \eqref{eq:tipf2supp}, since we know that $\auxd_{\idxi_{\subidx+1}}$ holds.

We finish the proof by noting that this shows that $\auxt_{\suppsize} \geq 1$, which directly conflicts with the premise $\olnot{\auxt_{\suppsize}} \geq 1$.
\end{proof}

\begin{proof}[Alternative derivation using cutting planes steps]
For step 3 and 4 in the proof of Lemma~\ref{lem:second_proof_goal}, we can also provide a derivation using cutting planes steps instead.

\paragraph{Showing that $\auxa_{\idxi_{\subidx}} + \olnot{\auxs_{\subidx}} \geq 1$.} 
For the base case, we add the constraint 
$2 \auxa_{\idxi_1} + \olnot{\varx_{\idxi_1}} + \symmetry(\varx_{\idxi_1}) \geq 2$ 
(from \eqref{eq:s1pf2supp}) and the constraint 
$\olnot{\auxs_1} + \varx_{\idxi_1} + \olnot{\symmetry(\varx_{\idxi_1})} \geq 1$ 
(from \eqref{eq:s1}) and saturate.

For the inductive step, we add the constraints
\begin{align*}
& 2\auxa_{\idxi_{\subidx+1}} + 2\olnot{\auxa_{\idxi_\subidx}} + \olnot{\varx_{\idxi_{\subidx+1}}} 
 + \symmetry(\varx_{\idxi_{\subidx+1}}) \geq 2 \\
& 3\olnot{\auxs_{\subidx+1}} + 2\auxs_\subidx + \varx_{\idxi_{\subidx+1}} 
 + \olnot{\symmetry(\varx_{\idxi_{\subidx+1}})} \geq 3 \\
& 2 \cdot \big(\auxa_{\idxi_{\subidx}} + \olnot{\auxs_{\subidx}} \geq 1\big)
\end{align*}
from \eqref{eq:sipf2supp}, \eqref{eq:si} and the induction hypothesis respectively, and then saturate, which yields $\auxa_{\idxi_{\subidx+1}} + \olnot{\auxs_{\subidx+1}} \geq 1$. 

\paragraph{Showing that $\auxt_{\subidx} \geq 1$.} 
For the base case, we add the constraint 
$2 \auxt_{1} + \olnot{\symmetry(\varx_{\idxi_1})} + \varx_{\idxi_1} \geq 2$ 
(from \eqref{eq:t1}) and the constraint 
$\olnot{\auxd_{\idxi_1}} + \olnot{\varx_{\idxi_1}} + \symmetry(\varx_{\idxi_1})  \geq 1$ 
(from \eqref{eq:t1pf2supp}) and the constraint 
$\auxd_{\idxi_1} \geq 1$ and saturate, which yields  $\auxt_{1} \geq 1$.

For the inductive step, we add the constraints
\begin{align*}
& 4\olnot{\auxd_{\idxi_{\subidx+1}}} + 3\auxd_{\idxi_\subidx} + \olnot{\auxa_{\idxi_\subidx}} 
 + \symmetry(\varx_{\idxi_{\subidx+1}}) + \olnot{\varx_{\idxi_{\subidx+1}}} \geq 4 \\
& 3 \auxt_{\subidx+1} + 3 \olnot{\auxt_\subidx} + \auxs_\subidx 
 + \olnot{\symmetry(\varx_{\idxi_{\subidx+1}})} + \varx_{\idxi_{\subidx+1}} \geq 3 \\
 & \auxa_{\idxi_{\subidx}} + \olnot{\auxs_{\subidx}} \geq 1
\end{align*}
from \eqref{eq:ti}, \eqref{eq:tipf2supp}, and earlier derivations. 
This yields 
\begin{equation*}
4\olnot{\auxd_{\idxi_{\subidx+1}}} + 3\auxd_{\idxi_\subidx} 
 + 3 \auxt_{\subidx+1} + 3 \olnot{\auxt_\subidx} \geq 4
\end{equation*}
Then we add  the constraints $4 \cdot \big(\auxd_{\idxi_{\subidx+1}} \geq 1\big)$ and $3 \cdot \big(\auxt_{\subidx} \geq 1\big)$ (from the induction hypothesis) and weaken away $\auxd_{\idxi_\subidx}$ and saturate, which yields  $\auxt_{\subidx+1} \geq 1$.
\end{proof}

Finally, we show how to derive the symmetry breaking clauses \eqref{eq:symbreak_s1.1} 
up to \eqref{eq:symbreak_ti} from the constraint $\auxt_{\suppsize} \geq 1$ 
that we derived using dominance. 

\begin{lemma}  \label{lem:symmetrybreakingclauses}
Let $\symmetry$ be a symmetry of the input formula and let $\suppsize = |\suppsymmetry|$. 
Then we can derive the symmetry breaking clauses 
\eqref{eq:symbreak_s1.1} up to \eqref{eq:symbreak_ti} 
from the constraint $\auxt_{\suppsize} \geq 1$ 
using $\bigoh{\suppsize}$ RUP steps and cutting planes steps.
\end{lemma}

\begin{proof}
We can derive constraint~\eqref{eq:symbreak_s1.1}, \ie 
$\auxs_1 + \olnot{\varx_{\idxi_1}}  \geq 1$, 
by weakening constraint~\eqref{eq:s1.2}, \ie 
$2\auxs_1 + \olnot{\varx_{\idxi_1}} + \symmetry(\varx_{\idxi_1}) \geq 2$, 
on $\symmetry(\varx_{\idxi_1})$ and saturating. 
Similarly, we can derive~\eqref{eq:symbreak_s1.2}, \ie 
$\auxs_1 + \symmetry(\varx_{\idxi_1}) \geq 1$, 
by weakening~\eqref{eq:s1.2} on $\varx_{\idxi_1}$ and saturating.

For constraint~\eqref{eq:symbreak_si.1}, \ie 
$\auxs_{\subidx+1} + \olnot{\auxs_\subidx} + \olnot{\varx_{\idxi_{\subidx+1}}} \geq 1$, 
we weaken~\eqref{eq:si.2}, \ie 
$2\auxs_{\subidx+1} + 2\olnot{\auxs_\subidx} + \olnot{\varx_{\idxi_{\subidx+1}}} 
 + \symmetry(\varx_{\idxi_{\subidx+1}}) \geq 2$, 
on $\symmetry(\varx_{\idxi_{\subidx+1}})$ and saturate. 
Similarly, for \eqref{eq:symbreak_si.2}, \ie 
$\auxs_{\subidx+1} + \olnot{\auxs_\subidx} + \symmetry(\varx_{\idxi_{\subidx+1}}) \geq 1$,
we weaken~\eqref{eq:si.2} on $\varx_{\idxi_{\subidx+1}}$ and saturate. 

It remains to prove constraints \eqref{eq:symbreak_t1}, \ie 
\mbox{$\symmetry(\varx_{\idxi_1}) + \olnot{\varx_{\idxi_1}} \geq 1$}, 
and \eqref{eq:symbreak_ti}, \ie 
$\olnot{\auxs_{\subidx}} + \symmetry(\varx_{\idxi_{\subidx+1}}) + \olnot{\varx_{\idxi_{\subidx+1}}} \geq 1$.
For this, we first derive using RUP from $\auxt_{\suppsize} \geq 1$ and 
\eqref{eq:ti.1} that $\auxt_{\subidx} \geq 1$ for $1 \leq \subidx \leq \suppsize$ (starting with the high indices). 
Then adding $\auxt_1 \geq 1$ to \eqref{eq:t1.1} yields \eqref{eq:symbreak_t1}, 
while adding $4 \cdot \big(\auxt_{\subidx+1} \geq 1\big)$ to \eqref{eq:ti.1}, 
weakening on $\auxt_{\subidx}$ and saturating yields  \eqref{eq:symbreak_ti}.
\end{proof}

Together this shows

\begin{theorem} \label{thm:symmetrybreakingclauses}
Assume that  the lexicographical order $\preceq$ over 
$\nvar$ variables is loaded.
Let $\symmetry$ be a symmetry of the input formula and let $\suppsize = |\suppsymmetry|$. 
Then we can derive the symmetry breaking clauses  \eqref{eq:symbreak_s1.1} up to \eqref{eq:symbreak_ti} using a proof of size $\bigoh{\suppsize}$ that can be checked in time $\bigoh{\nvar}$.
\end{theorem}

\section{Example Of Symmetry Breaking Proof}
\label{app:example}

We present an example of a proof in the \textsc{VeriPB} format
using auxiliary preorder variables. %
In order to keep the size of the proof somewhat manageable, 
we consider a toy example:
the pigeonhole principle with 3 pigeons and 2 holes.
This formula contains variables $x_i$, where
$x_1$ and $x_2$ encode that pigeon 1 flies into hole 1 and 2, 
$x_3$ and $x_4$ encode that pigeon 2 flies into hole 1 and 2, and
$x_5$ and $x_6$ encode that pigeon 3 flies into hole 1 and 2, respectively.
This unsatisfiable formula claims that, first,
each pigeon is in some hole:
\begin{align*}
    x_1 \vee x_2&&
    x_3 \vee x_4&&
    x_5 \vee x_6\eqcomma
\end{align*}
and second, that each hole contains at most one pigeon:
\begin{align*}
    \overline{x_1} \vee {} & \overline{x_3} &
    \overline{x_1} \vee {} & \overline{x_5} &
    \overline{x_3} \vee {} & \overline{x_5} \\
    \overline{x_2} \vee {} & \overline{x_4} &
    \overline{x_2} \vee {} & \overline{x_6} &
    \overline{x_4} \vee {} & \overline{x_6}\eqperiod
\end{align*}
Our proof example breaks two symmetries of this formula:
\begin{align*}
\sigma = {} & \{ x_1 \mapsto x_3,\;x_2 \mapsto x_4,\;x_3\mapsto x_1,\;x_4\mapsto x_2\}\eqcomma\\
\tau = {} & \{ x_1\mapsto x_6,\;x_2\mapsto x_5,\;x_3 \mapsto x_2, \\
          & \qquad \;x_4\mapsto x_1,\;x_5\mapsto x_4,\;x_6 \mapsto x_3\}\eqcomma
\end{align*}
where $\sigma$ swaps pigeons 1 and 2, and $\tau$ moves every pigeon to the other hole,
and renames it to the previous pigeon.

The proof starts by defining a preorder \texttt{lex6}, corresponding to the lexicographic order over
a list of 6 variables.
\begin{lstlisting}[language=veripb,numbers=none,frame=none,aboveskip=\smallskipamount,belowskip=\smallskipamount,xleftmargin=2ex]
pseudo-Boolean proof version 3.0
def_order lex6
vars
left u1 u2 u3 u4 u5 u6;
right v1 v2 v3 v4 v5 v6;
aux $a1 $a2 $a3 $a4 $a5 $d1 $d2 $d3 $d4 $d5 $d6;
end vars;
\end{lstlisting}
Here, variables $u_i$ and $v_i$ are declared as left-hand and right-hand
variables for the preorder \texttt{lex6}. Furthermore, auxiliary preorder
variables $a_i$ and $d_i$ are declared.
Next, the specification is declared, following Lemma~\ref{lem:lex_using_aux_PB}.
\begin{lstlisting}[language=veripb,numbers=none,frame=none,aboveskip=\smallskipamount,belowskip=\smallskipamount,xleftmargin=1ex]
spec
red +1 ~$a1 +1 u1 +1 ~v1 >= 1 : $a1 -> 0;
red +2 $a1 +1 ~u1 +1 v1 >= 2 : $a1 -> 1;
red +3 ~$a2 +2 $a1 +1 u2 +1 ~v2 >= 3 : $a2 -> 0;
red +2 $a2 +2 ~$a1 +1 ~u2 +1 v2 >= 2 : $a2 -> 1;
\end{lstlisting}
Here, the constraint $\overline{a_1} + u_1 + \overline{v_1} \geq 1$ is first
introduced in the specification by redundance with witness
$\{a_1 \mapsto 0\}$. Since $a_1$ is a fresh variable at this point,
this is correct.
Next, the constraint $a_1 + \overline{u_1} + v_1 \geq 2$ is introduced
by redundance with witness $\{a_1 \mapsto 1\}$.
Note that this witness satisfies the previous constraint,
so this is correct too.
The specification then introduces constraints
$3\overline{a_2} + 2a_1 + u_2 + \overline{v_2} \geq 3$
and $2a_2 + \overline{a_1} + \overline{u_2} + v_2 \geq 2$
with witnesses $\{a_2 \mapsto 0\}$ and $\{a_2 \mapsto 1\}$,
which are correct for similar reasons.
The rest of the specification is introduced in a similar way:
\begin{lstlisting}[language=veripb,numbers=none,frame=none,aboveskip=\smallskipamount,belowskip=\smallskipamount,xleftmargin=1ex]
red +3 ~$a3 +2 $a2 +1 u3 +1 ~v3 >= 3 : $a3 -> 0;
red +2 $a3 +2 ~$a2 +1 ~u3 +1 v3 >= 2 : $a3 -> 1;
red +3 ~$a4 +2 $a3 +1 u4 +1 ~v4 >= 3 : $a4 -> 0;
red +2 $a4 +2 ~$a3 +1 ~u4 +1 v4 >= 2 : $a4 -> 1;
red +3 ~$a5 +2 $a4 +1 u5 +1 ~v5 >= 3 : $a5 -> 0;
red +2 $a5 +2 ~$a4 +1 ~u5 +1 v5 >= 2 : $a5 -> 1;
red +1 ~$d1 +1 ~u1 +1 v1 >= 1 : $d1 -> 0;
red +2 $d1 +1 u1 +1 ~v1 >= 2 : $d1 -> 1;
red +4 ~$d2 +3 $d1 +1 ~$a1 +1 ~u2 +1 v2 >= 4 : $d2 -> 0;
red +3 $d2 +3 ~$d1 +1 $a1 +1 u2 +1 ~v2 >= 3 : $d2 -> 1;
red +4 ~$d3 +3 $d2 +1 ~$a2 +1 ~u3 +1 v3 >= 4 : $d3 -> 0;
red +3 $d3 +3 ~$d2 +1 $a2 +1 u3 +1 ~v3 >= 3 : $d3 -> 1;
red +4 ~$d4 +3 $d3 +1 ~$a3 +1 ~u4 +1 v4 >= 4 : $d4 -> 0;
red +3 $d4 +3 ~$d3 +1 $a3 +1 u4 +1 ~v4 >= 3 : $d4 -> 1;
red +4 ~$d5 +3 $d4 +1 ~$a4 +1 ~u5 +1 v5 >= 4 : $d5 -> 0;
red +3 $d5 +3 ~$d4 +1 $a4 +1 u5 +1 ~v5 >= 3 : $d5 -> 1;
red +4 ~$d6 +3 $d5 +1 ~$a5 +1 ~u6 +1 v6 >= 4 : $d6 -> 0;
red +3 $d6 +3 ~$d5 +1 $a5 +1 u6 +1 ~v6 >= 3 : $d6 -> 1;
end spec;
\end{lstlisting}
Finally, the preorder constraints for \texttt{lex6} are introduced.
As per Section \ref{ssc:satsuma_order},
these are a single constraint $d_6 \geq 1$.
\begin{lstlisting}[language=veripb,numbers=none,frame=none,aboveskip=\smallskipamount,belowskip=\smallskipamount,xleftmargin=1ex]
def
+1 $d6 >= 1;
end def;
\end{lstlisting}
This section is followed by a transitivity proof.
The proof starts by declaring some new variables.
\begin{lstlisting}[language=veripb,numbers=none,frame=none,aboveskip=\smallskipamount,belowskip=\smallskipamount,xleftmargin=1ex]
transitivity
vars
fresh_right w1 w2 w3 w4 w5 w6 ;
fresh_aux_1 $b1 $b2 $b3 $b4 $b5 $e1 $e2 $e3 $e4 $e5 $e6;
fresh_aux_2 $c1 $c2 $c3 $c4 $c5 $f1 $f2 $f3 $f4 $f5 $f6;
end vars;
\end{lstlisting}
These are needed to account for a new set of preorder variables $\vec{w}$,
and two new sets of preorder auxiliary variables $\vec{b},\vec{e}$ and $\vec{c},\vec{f}$.
These are given by the lines \texttt{fresh\_right}, \texttt{fresh\_aux\_1} and \texttt{fresh\_aux\_2},
respectively.
This proof will show the following implication:
\begin{align*}
    \spec(\vec{u}, \vec{v}, \vec{a},\vec{d}) \cup \ord(\vec{u}, \vec{v}, \vec{a}, \vec{d}) & \\
    {} \cup \spec(\vec{v}, \vec{w}, \vec{b},\vec{e}) \cup \ord(\vec{v}, \vec{w}, \vec{b}, \vec{e}) & \\
    {} \cup \spec(\vec{u}, \vec{w}, \vec{c}, \vec{f}) & {} \vdash \ord(\vec{u}, \vec{w}, \vec{c}, \vec{f})
\end{align*}
Copies of the specification and preorder constraints are implicitly brought into scope
and assigned numerical identifiers.
In this particular case, the specification contains 22 constraints, and the preorder contains a single constraint.
Then, a total of 68 constraints are implicitly introduced in the proof scope,
and given identifiers in the same order as they were defined in the \texttt{spec} and \texttt{ord} sections:
\begin{itemize}
\item $\spec(\vec{u}, \vec{v}, \vec{a},\vec{d})$ is mapped into identifers \texttt{1} through \texttt{22}.
\item $\spec(\vec{v}, \vec{w}, \vec{b},\vec{e})$ is mapped into identifiers \texttt{23} through \texttt{44}.
\item $\spec(\vec{u}, \vec{w}, \vec{c},\vec{f})$ is mapped into identifiers \texttt{45} through \texttt{66}.
\item The single constraints in $\ord(\vec{u}, \vec{v}, \vec{a},\vec{d})$ and
$\ord(\vec{v}, \vec{w}, \vec{b},\vec{e})$ are mapped into identifiers \texttt{67} and \texttt{68}, respectively.
\end{itemize}
\begin{lstlisting}[language=veripb,numbers=none,frame=none,aboveskip=\smallskipamount,belowskip=\smallskipamount,xleftmargin=1ex]
proof
proofgoal #1
\end{lstlisting}
Next, the \texttt{proof} section shows the implication above.
Each \texttt{proofgoal} line provides a proof for one constraint in $\ord(\vec{u}, \vec{w}, \vec{c}, \vec{f})$,
in order of definition in the \texttt{def} section. In our case, we only must prove $f_6 \geq 1$,
so only one \texttt{proofgoal} appears in the proof.
Within this context, the negation of the goal, $\overline{f_6} \geq 1$,
is automatically assigned 
the next available identifier,
in this case \texttt{69}; to satisfy the proofgoal
we must reach a contradiction.
\begin{lstlisting}[language=veripb,numbers=none,frame=none,aboveskip=\smallskipamount,belowskip=\smallskipamount,xleftmargin=1ex]
pol 67 4 * 21 +;
rup +1 $d5 >= 1 : -1;
pol -2 $d5 w;
pol -2 4 * 19 +;
rup +1 $d4 >= 1 : -1;
pol -2 $d4 w;
\end{lstlisting}
The proof roughly follows Lemma~\ref{lem:lex_using_aux_trans}.
First, constraints
\begin{align}
\label{eqn:example-trans-rec}
3d_i + \overline{a_i} + u_{i+1} + \overline{v_{i + 1}} & {} \geq 4 \\
\label{eqn:example-trans-lit}
d_i & {} \geq 1 \\
\label{eqn:example-trans-bound}
\overline{a_i} + \overline{u_{i + 1}} + v_{i+1} & {} \geq 1
\end{align}
are derived for $i = 5,\dots,1$. We can see this in the displayed fragment for $i=5$.
Constraint \eqref{eqn:example-trans-rec} is derived through the line \lstinline[language=veripb]|pol 67 4 * 21 +|.
Lines starting with \texttt{pol} derive the result of a cutting planes proof by providing
explicit operations in reverse Polish notation. This line first fetches constraint 67
(which is $d_6 \geq 1$),
then scales it by $4$, then adds the result with constraint 21
(i.e.\ $4\overline{d_6} + 3d_5 + \overline{a_5} + u_6 + \overline{v_6} \geq 4$).
Constraint \eqref{eqn:example-trans-lit} is derived through the line \lstinline[language=veripb]|rup +1 $d5 >= 1 : -1|.
Lines starting with \texttt{rup} apply a form of cutting planes proof search called
\emph{reverse unit propagation}~(RUP) that can automatically identify some inferences.
After the colon, a list of constraint identifiers can be provided to use as premises for
proof search. Negative identifiers refer to identifiers relative to the current constraint,
so $-1$ in this case refers to the constraint we just derived, namely \eqref{eqn:example-trans-rec}.
For premises $F$ and a conclusion $C$, RUP tries to identify a contradiction in $F \cup \{ \overline{C} \}$
by constraint propagation. In this case, this formula contains the constraints
\begin{align*}
    3d_5 + \overline{a_5} + u_6 + \overline{v_6} & {} \geq 4 &
    \overline{d_5} & {} \geq 1
\end{align*}
which constraint propagation proves inconsistent.
Finally, constraint \eqref{eqn:example-trans-bound} is derived through the line
\lstinline[language=veripb]|pol -2 $d5 w|, which is reverse Polish notation for weakening
constraint \eqref{eqn:example-trans-rec} on variable $d_5$.
\begin{lstlisting}[language=veripb,numbers=none,frame=none,aboveskip=\smallskipamount,belowskip=\smallskipamount,xleftmargin=1ex]
pol -2 4 * 17 +;
rup +1 $d3 >= 1 : -1;
pol -2 $d3 w;
pol -2 4 * 15 +;
rup +1 $d2 >= 1 : -1;
pol -2 $d2 w;
pol -2 4 * 13 +;
rup +1 $d1 >= 1 : -1;
pol -2 $d1 w;
pol -2 11 +;
\end{lstlisting}
This process iterates on $i$ until the following constraints are derived
\begin{align*}
d_5 {} & \geq 1 \quad\text{(\texttt{71})}\eqcomma & \overline{a_5} + \overline{u_6} + v_6 & {} \geq 1 \quad\text{(\texttt{72})}\eqcomma\\
d_4 {} & \geq 1 \quad\text{(\texttt{74})}\eqcomma & \overline{a_4} + \overline{u_5} + v_5 & {} \geq 1 \quad\text{(\texttt{75})}\eqcomma\\
d_3 {} & \geq 1 \quad\text{(\texttt{77})}\eqcomma & \overline{a_3} + \overline{u_4} + v_4 & {} \geq 1 \quad\text{(\texttt{78})}\eqcomma\\
d_2 {} & \geq 1 \quad\text{(\texttt{80})}\eqcomma & \overline{a_2} + \overline{u_3} + v_3 & {} \geq 1 \quad\text{(\texttt{81})}\eqcomma\\
d_1 {} & \geq 1 \quad\text{(\texttt{83})}\eqcomma & \overline{a_1} + \overline{u_2} + v_2 & {} \geq 1 \quad\text{(\texttt{84})}\eqcomma\\
\overline{u_1} + v_1 {} & {} \geq 1 \quad\text{(\texttt{85})}\eqcomma
\end{align*}
where the last constraint is derived through the line \lstinline[language=veripb]|pol -2 11 +|.
The proof then repeats the process above for variables $b_i$, $e_i$.
\begin{lstlisting}[language=veripb,numbers=none,frame=none,aboveskip=\smallskipamount,belowskip=\smallskipamount,xleftmargin=1ex]
pol 68 4 * 43 +;
rup +1 $e5 >= 1 : -1;
pol -2 $e5 w;
pol -2 4 * 41 +;
rup +1 $e4 >= 1 : -1;
pol -2 $e4 w;
pol -2 4 * 39 +;
rup +1 $e3 >= 1 : -1;
pol -2 $e3 w;
pol -2 4 * 37 +;
rup +1 $e2 >= 1 : -1;
pol -2 $e2 w;
pol -2 4 * 35 +;
rup +1 $e1 >= 1 : -1;
pol -2 $e1 w;
pol -2 33 +;
\end{lstlisting}
This yields constraints
\begin{align*}
e_5 {} & \geq 1 \quad\text{(\texttt{87})}\eqcomma& \overline{b_5} + \overline{v_6} + w_6 & {} \geq 1 \text{\quad(\texttt{88})}\eqcomma  \\
e_4 {} & \geq 1 \quad\text{(\texttt{90})}\eqcomma& \overline{b_4} + \overline{v_5} + w_5 & {} \geq 1 \text{\quad(\texttt{91})}\eqcomma  \\
e_3 {} & \geq 1 \quad\text{(\texttt{93})}\eqcomma& \overline{b_3} + \overline{v_4} + w_4 & {} \geq 1 \text{\quad(\texttt{94})}\eqcomma  \\
e_2 {} & \geq 1 \quad\text{(\texttt{96})}\eqcomma& \overline{b_2} + \overline{v_3} + w_3 & {} \geq 1 \text{\quad(\texttt{97})}\eqcomma  \\
e_1 {} & \geq 1 \quad\text{(\texttt{99})}\eqcomma& \overline{b_1} + \overline{v_2} + w_2 & {} \geq 1 \text{\quad(\texttt{100})}\eqcomma \\
\overline{v_1} + w_1 {} & {} \geq 1 \text{\quad (\texttt{101})}\eqperiod
\end{align*}
Next, the proof derives constraints $A_i$ and $B_i$ given by
$a_i + \overline{c_i} \geq 1$ and $b_i + \overline{c_i} \geq 1$
for $i = 1,\dots,6$, respectively.
This follows the inductive process detailed in Lemma~\ref{lem:lex_using_aux_trans}.
\begin{lstlisting}[language=veripb,numbers=none,frame=none,aboveskip=\smallskipamount,belowskip=\smallskipamount,xleftmargin=1ex]
pol 24 45 + 85 + s;
pol 47 u2 w w2 w s;
\end{lstlisting}
In the base case, the line \lstinline[language=veripb]|pol 2 45 + -1 + s| derives
$a_1 + \overline{c_1} \geq 1$ by adding and saturating constraints
\begin{align*}
\overline{u_1} + v_1 + 2 a_1 & {} \geq 2\\
u_1 + \overline{w_1} + \overline{c_1} \geq 1\\
\overline{v_1} + w_1 {} & {} \geq 1
\end{align*}
where the two constraints above come from the specifications,
and the one below is 101. Similarly, the line
\lstinline[language=veripb]|pol 24 45 + 85 + s| derives $b_1 + \overline{c_1} \geq 1$.
\begin{lstlisting}[language=veripb,numbers=none,frame=none,aboveskip=\smallskipamount,belowskip=\smallskipamount,xleftmargin=1ex]
pol -1 -3 +;
pol -2 -3 +;
pol 47 $c1 w s;
pol -2 100 + -1 + -3 2 * + 4 + s;
pol -4 84 + -2 + -3 2 * + 26 + s;
\end{lstlisting}
The induction case needs to derive some intermediate lemmas at each step.
\begin{itemize}
    \item First, the constraint $C_i$ given by $c_i + c_{i+1} \geq 1$ is obtained by
    weakening the specification constraint $u_{i+1} + \overline{w_{i + 1}} + 2c_i + 3\overline{c_{i+1}} \geq 3$
    on variables $u_{i+1}$ and $w_{i+1}$ and then saturating; in the fragment displayed,
    this is done for $i = 1$ by the line \lstinline[language=veripb]|pol 47 u2 w w2 w s|.
    \item Then, the constraints $A^\prime_i$ and $B^\prime_i$ given by
    $a_i + \overline{c_{i+1}} \geq 1$ and $b_i + \overline{c_{i+1}} \geq 1$
    are derived by adding $C_i$ with $A_i$ and with $B_i$, respectively.
    In the displayed fragment,
    this is done for $i = 1$ by the lines \lstinline[language=veripb]|pol -1 -3 +| and \lstinline[language=veripb]|pol -2 -3 +|.
    \item The constraint $L_i$ given by $u_{i+1} + \overline{w_{i+1}} + \overline{c_{i+1}} \geq 1$
    is derived by weakening the specification constraint
    $u_{i+1} + \overline{w_{i+1}} + 2 c_i + 3 \overline{c_{i+1}} \geq 3$ on variable $c_i$
    and then saturating. In the displayed fragment,
    this is done for $i = 1$ by the line \lstinline[language=veripb]|pol 47 $c1 w s|.
    \item Finally, $A_{i+1}$ is derived by adding the constraints $2 \cdot A^\prime_i + B^\prime_i + L_i$
    with $\overline{b_i} + \overline{v_{i+1}} + w_{i + 1} \geq 1$
    (which was derived in the previous proof fragment) and
    $\overline{u_{i+1}} + \overline{v_{i+1}} + 2\overline{a_i} + 2a_{i+1} \geq 2$
    (which is a specification constraint), and then saturating.
    The constraint $B_{i+1}$ is analogously derived as well.
    In the displayed fragment,
    this is done for $i = 1$ by the lines \lstinline[language=veripb]|pol -2 100 + -1 + -3 2 * + 4 + s|
    and \lstinline[language=veripb]|pol -4 84 + -2 + -3 2 * + 26 + s|.
\end{itemize}
\begin{lstlisting}[language=veripb,numbers=none,frame=none,aboveskip=\smallskipamount,belowskip=\smallskipamount,xleftmargin=1ex]
pol 49 u3 w w3 w s;
pol -1 -3 +;
pol -2 -3 +;
pol 49 $c2 w s;
pol -2 97 + -1 + -3 2 * + 6 + s;
pol -4 81 + -2 + -3 2 * + 28 + s;
pol 51 u4 w w4 w s;
pol -1 -3 +;
pol -2 -3 +;
pol 51 $c3 w s;
pol -2 94 + -1 + -3 2 * + 8 + s;
pol -4 78 + -2 + -3 2 * + 30 + s;
pol 53 u5 w w5 w s;
pol -1 -3 +;
pol -2 -3 +;
pol 53 $c4 w s;
pol -2 91 + -1 + -3 2 * + 10 + s;
pol -4 75 + -2 + -3 2 * + 32 + s;
pol 85 101 +;
\end{lstlisting}
Last, the proof derives the constraint $F_i$
given by $f_i \geq 1$ for $i = 1,\dots,6$.
\begin{lstlisting}[language=veripb,numbers=none,frame=none,aboveskip=\smallskipamount,belowskip=\smallskipamount,xleftmargin=1ex]
pol -1 56 + s;
pol 84 100 + 102 + 103 + s -1 3 * +;
pol -1 58 + s;
pol 81 97 + 108 + 109 + s -1 3 * +;
pol -1 60 + s;
pol 78 94 + 114 + 115 + s -1 3 * +;
pol -1 62 + s;
pol 75 91 + 120 + 121 + s -1 3 * +;
pol -1 64 + s;
pol 72 88 + 126 + 127 + s -1 3 * +;
pol -1 66 + s;
\end{lstlisting}
For the base case, can derive $F_1$ by adding constraints
$\overline{u_1} + v_1 \geq 1$ and $\overline{v_1} + w_1 \geq 1$,
which we derived as 85 and 101; this is the line \lstinline[language=veripb]|pol 85 101 +|.
In the induction case, we first derive a lemma $F^\prime_i$ given by
$\overline{u_{i+1}} + w_{i+1} + \overline{c_i} + 3f_i \geq 4$.
This is obtained by first and then saturating the constraints
\begin{align*}
\overline{a_i} + \overline{u_{i+1}} + v_{i+1} & {} \geq 1\eqcomma\\
\overline{b_i} + \overline{u_{i+1}} + v_{i+1} & {} \geq 1\eqcomma\\
a_i + \overline{c_i} & {} \geq 1\eqcomma\\
b_i + \overline{c_i} & {} \geq 1\eqcomma
\end{align*}
resulting in $\overline{c_i} + \overline{u_{i+1}} + \overline{w_{i+1}}$;
the two constraints above come from the first transitivity proof fragment;
the two constraints below are $A_i$ and $B_i$.
The result is added to $3 \cdot F_i$, which yields $F^\prime_i$.
This is done for $i=1$ by line
\lstinline[language=veripb]|pol 84 100 + 102 + 103 + s -1 3 * +|.
The constraint $F_{i+1}$ can then be derived by adding $F^\prime_i$
to the specification constraint
$u_{i+1} + \overline{w_{i+1}} + c_i + 3\overline{f_i} + 3\overline{f_{i+1}} \geq 3$
and saturating; this is done for $i=1$ by the line
\lstinline[language=veripb]|pol -1 58 + s|.
\begin{lstlisting}[language=veripb,numbers=none,frame=none,aboveskip=\smallskipamount,belowskip=\smallskipamount,xleftmargin=1ex]
pol -1 69 +;
qed #1 : -1;
qed proof;
end transitivity;
\end{lstlisting}
Once $F_6$ has been derived, the line \lstinline[language=veripb]|pol -1 69 +| adds it
to the negated proof goal $\overline{f_6} \geq 1$, creating a contradiction $0 \geq 1$.
The line \lstinline[language=veripb]|qed #1 : -1| communicates this contradiction to the proof checker,
finishing the transitivity proof. The reflexivity proof follows.
\begin{lstlisting}[language=veripb,numbers=none,frame=none,aboveskip=\smallskipamount,belowskip=\smallskipamount,xleftmargin=1ex]
reflexivity
proof
proofgoal #1
rup >= 1;
qed #1 : -1;
qed proof;
end reflexivity;
\end{lstlisting}
The reflexivity section automatically identifies the constraints in
$\spec(\vec{u}, \vec{u}, \vec{a},\vec{d})$ as constraints \texttt{1} through \texttt{22}.
We must provide a proof of $\ord(\vec{u}, \vec{u}, \vec{a},\vec{d})$,
which contains a single constraint $d_6 \geq 1$, identified by the proofgoal \#1.
Similar to the transitivity proof, entering the \lstinline[language=veripb]|proofgoal #1| section
negates this goal as $\overline{d_6} \geq 1$ and adds it to the premises
as constraint \texttt{23}.
As shown by Lemma~\ref{lem:lex_using_aux_reflexivity}, this proof derives
a contradiction in a single RUP step in line \lstinline[language=veripb]|rup >= 1|.
\begin{lstlisting}[language=veripb,numbers=none,frame=none,aboveskip=\smallskipamount,belowskip=\smallskipamount,xleftmargin=1ex]
end def_order;
load_order lex6 x5 x6 x1 x2 x3 x4;
\end{lstlisting}
After this, the preorder \texttt{lex6} has been correctly defined,
and the refutation of the pigeonhole principle starts,
assigning its constraints identifiers 1 through 9.
The line \lstinline[language=veripb]|load_order lex6 x5 x6 x1 x2 x3 x4| applies an order change rule
with the preorder \texttt{lex6}. The list of variables there
tells the proof checker to map formula variables to preorder variables as:
\begin{align*}
x_5 \mapsto {} & u_1, v_1 & x_6 \mapsto {} & u_2, v_2 & x_1 \mapsto {} & u_3, v_3 \\
x_2 \mapsto {} & u_4, v_4 & x_3 \mapsto {} & u_5, v_5 & x_3 \mapsto {} & u_6, v_6
\end{align*}

The proof immediately proceeds to break the symmetry $\sigma$.
First, the symmetry breaker circuit is introduced.
\begin{lstlisting}[language=veripb,numbers=none,frame=none,aboveskip=\smallskipamount,belowskip=\smallskipamount,xleftmargin=1ex]
red +1 ~s1 +1 x1 +1 ~x3 >= 1 : s1 -> 0;
red +2 s1 +1 ~x1 +1 x3 >= 2 : s1 -> 1;
red +3 ~s2 +2 s1 +1 x2 +1 ~x4 >= 3 : s2 -> 0;
red +2 s2 +2 ~s1 +1 ~x2 +1 x4 >= 2 : s2 -> 1;
red +3 ~s3 +2 s2 +1 x3 +1 ~x1 >= 3 : s3 -> 0;
red +2 s3 +2 ~s2 +1 ~x3 +1 x1 >= 2 : s3 -> 1;
red +1 ~t1 +1 ~x1 +1 x3 >= 1 : t1 -> 0;
red +2 t1 +1 x1 +1 ~x3 >= 2 : t1 -> 1;
red +4 ~t2 +3 t1 +1 ~s1 +1 ~x2 +1 x4 >= 4 : t2 -> 0;
red +3 t2 +3 ~t1 +1 s1 +1 x2 +1 ~x4 >= 3 : t2 -> 1;
red +4 ~t3 +3 t2 +1 ~s2 +1 ~x3 +1 x1 >= 4 : t3 -> 0;
red +3 t3 +3 ~t2 +1 s2 +1 x3 +1 ~x1 >= 3 : t3 -> 1;
red +4 ~t4 +3 t3 +1 ~s3 +1 ~x4 +1 x2 >= 4 : t4 -> 0;
red +3 t4 +3 ~t3 +1 s3 +1 x4 +1 ~x2 >= 3 : t4 -> 1;
\end{lstlisting}
These correspond to constraints \eqref{eq:s1.1} through \eqref{eq:ti.2}.
Each constraint is introduced by redundance-based strengthening.
First, the constraint $C_1 = \overline{s_1} + x_1 + \overline{x_3} \geq 1$ is
derived with witness $\omega_1 = \{s_1 \mapsto 0\}$; this corresponds to constraint
\eqref{eq:s1.1}. The check
$\core \cup \der \cup \{ \overline{C_1} \} \vdash (\core \cup \der \cup \{ C_1 \})\uh_{\omega_1}$
succeeds trivially, since $s_1$ does not occur in $\core \cup \der$ and
$C_1\uh_{\omega_1}$ is a tautology.
For the same reason, the check
$\spec(\vec{x}\uh_{\omega_1},\vec{x},\vec{a},\vec{d}) \vdash \ord(\vec{x}\uh_{\omega_1},\vec{x},\vec{a},\vec{d})$
can be skipped, since this is the same as the reflexivity check.

Next, the constraint $C_2 = 2s_1 + \overline{x_1} + x_3 \geq 2$ is derived with
witness $\omega_2 = \{s_1 \mapsto 1\}$; this corresponds to constraint
\eqref{eq:s1.2}. The check
$\core \cup \der \cup \{ \overline{C_2} \} \vdash (\core \cup \der \cup \{ C_2 \})\uh_{\omega_2}$
succeeds trivially. The reasons are the same as above for all constraints except
the recently introduced $C_1$; in that case simply observe that $C_1$ is the result
of weakening $\overline{C_2}$ (i.e.\ $2\overline{s_1} + x_1 + \overline{x_3} \geq 3$)
on variable $x_3$.

A similar reasoning proves all remaining redundance-based strengthening steps.
They receive constraint identifiers \texttt{10} through \texttt{23}.
Once these circuit constraints have been introduced, the proof is ready to introduce
the symmetry breaker by dominance-based strengthening.
\begin{lstlisting}[language=veripb,numbers=none,frame=none,aboveskip=\smallskipamount,belowskip=\smallskipamount,xleftmargin=1ex]
dom +1 t4  >= 1 : x1 -> x3 x2 -> x4 x3 -> x1 x4 -> x2  : subproof
\end{lstlisting}
This line starts the dominance proof to derive the constraint $C$ given by $t_4 \geq 1$
with witness $\sigma$. Note that all the constraints derived above
by redundance-based strengthening have been derived into $\der$;
this becomes relevant for the checks elicited by dominance-based strengthening:
\begin{gather*}
\core \cup \der \cup \{\neg C\} \cup \spec(\vec{x}\uh_\sigma, \vec{x}, \vec{a},\vec{d}) \vdash \core\uh_\sigma \!\cup \ord(\vec{x}\uh_\sigma, \vec{x}, \vec{a},\vec{d}) \\
\core \cup \der \cup \{\neg C\} \cup \spec(\vec{x}, \vec{x}\uh_\sigma,\vec{a}, \vec{d}) \cup \ord(\vec{x}, \vec{x}\uh_\sigma, \vec{a}, \vec{d}) \vdash \bot
\end{gather*}
The single constraint in $\ord(\vz\uh_\sigma, \vz, \va)$ in the first check
is assigned proof goal identifier \texttt{\#1}, and the $\bot$ in the second check
is assigned proof goal identifier \texttt{\#2}. The goals $\core\uh_\sigma$ in the first
check are assigned as proof goal identifiers their constraint identifiers from $\core$,
without a ``\texttt{\#}''; we actually do not need to prove these,
since $\sigma$ is a symmetry of $\core$, so $\core\uh_\sigma = \core$ holds.
The \texttt{subproof} context above introduces the constraint $\neg C$
given by $\overline{t_4} \geq 1$ with identifier \texttt{24}.
\begin{lstlisting}[language=veripb,numbers=none,frame=none,aboveskip=\smallskipamount,belowskip=\smallskipamount,xleftmargin=1ex]
scope leq
proofgoal #1
\end{lstlisting}
The scope \lstinline[language=veripb]|leq| introduces the constraints
$\spec(\vec{x}\uh_\sigma, \vec{x}, \vec{a},\vec{d})$ from the first check
with constraint identifiers \texttt{25} through \texttt{46}:
and \texttt{proofgoal \#1} selects the only constraint in
$\ord(\vec{x}\uh_\sigma, \vec{x}, \vec{a},\vec{d})$, namely $d_6 \geq 1$,
as a proofgoal; its negation is introduced with constraint identifier \texttt{47}.
First the proof tries to derive the simplified specification constraints:
\begin{align*}
    \overline{a_3} + x_3 + \overline{x_1} & {} \geq 1 &
        \overline{d_3} + x_1 + \overline{x_3} & {} \geq 1 \\
    2a_3 + \overline{x_3} + x_1 & {} \geq 2 &
        2d_3 + \overline{x_1} + x_3 & {} \geq 2\\
    3\overline{a_4} + 2a_3 + x_4 + \overline{x_2} & {} \geq 3 &
        4\overline{d_4} + 3d_3 + \overline{a_3} + x_2 + \overline{x_4} & {} \geq 4\\
    2a_4 + 2\overline{a_3} + \overline{x_4} + x_2 & {} \geq 2 &
        3d_4 + 3\overline{d_3} + a_3 + \overline{x_2} + x_4 & {} \geq 3\\
    3\overline{a_5} + 2a_4 + x_1 + \overline{x_3} & {} \geq 3 &
        4\overline{d_5} + 3d_4 + \overline{a_4} + x_3 + \overline{x_1} & {} \geq 4\\
    2a_5+ 2\overline{a_4} + \overline{x_1} + x_3 & {} \geq 2 &
        3d_5 + 3\overline{d_4} + a_4 + \overline{x_3} + x_1 & {} \geq 3\\
    &&  4\overline{d_6} + 3d_5 + \overline{a_5} + x_4 + \overline{x_2} & {} \geq 4\\
    &&  3d_6 + 3\overline{d_5} + a_5 + \overline{x_4} + x_2 & {} \geq 3\eqperiod
\end{align*}
First the reification 
$a_3 \Longleftrightarrow x_3 \geq x_1$ is obtained.
\begin{lstlisting}[language=veripb,numbers=none,frame=none,aboveskip=\smallskipamount,belowskip=\smallskipamount,xleftmargin=1ex]
rup +1 ~$d6 >= 1;
rup +1 $a2 >= 1;
pol 29 ~$a2 2 * + s;
pol 30 -2 2 * +;
\end{lstlisting}
Since $\sigma(x_5) = x_5$ and $\sigma(x_6) = x_6$,
$\spec(\vec{x}\uh_\sigma, \vec{x}, \vec{a},\vec{d})$ contains the constraints
$a_1 \geq 1$ and $2a_2 + 2\overline{a_i} \geq 1$,
from which we can derive the lemma $a_2 \geq 1$ by RUP.
Then from the literal axiom $\overline{a_2} \geq 2$ and
specification constraint $3\overline{a_3} + 2a_2  + x_3 + \overline{x_1} \geq 3$ we can derive
$\overline{a_3} + x_3 + \overline{x_1} \geq 1$.
Similarly, from the lemma $a_1 \geq 1$ and the specification constraint
$2a_3 + 2\overline{a_2} + \overline{x_3} + x_1 \geq 2$ we
derive $2a_3 + \overline{x_3} + x_1 \geq 2$.

Next, we obtain the reification 
$a_4 \Longleftrightarrow a_3 \wedge (x_4 \geq x_2)$.
\begin{lstlisting}[language=veripb,numbers=none,frame=none,aboveskip=\smallskipamount,belowskip=\smallskipamount,xleftmargin=1ex]
rup +1 ~$a3 +1 $a3 >= 1;
rup +1 $a3 +1 ~$a3 >= 1;
pol 31 -2 2 * +;
pol 32 -2 2 * +;
\end{lstlisting}
The proof first tries to derive $a_i \Longleftrightarrow a_j$ where
$i$ is the index for $a$ right after the previous support variable, and $j$ is
the index for $a$ at the next support variable.
Since in this case $i = j = 3$, the proof actually degenerates into
deriving two tautologies by RUP.
Next, the proof derives constraints
$3\overline{a_4} + 2a_3 + x_4 + \overline{x_2} \geq 3$ and
$2a_4 + 2\overline{a_3} + \overline{x_4} + x_2 \geq 2$ by operating
over some specification constraints and the reification $a_i \Longleftrightarrow a_j$.
These constraints actually already appear in the specification
$\spec(\vec{x}\uh_\sigma, \vec{x}, \vec{a},\vec{d})$,
so we do not really derive anything new.
The proof then repeats this last step for the variable $a_5$.
\begin{lstlisting}[language=veripb,numbers=none,frame=none,aboveskip=\smallskipamount,belowskip=\smallskipamount,xleftmargin=1ex]
rup +1 ~$a4 +1 $a4 >= 1;
rup +1 $a4 +1 ~$a4 >= 1;
pol 33 -2 2 * +;
pol 34 -2 2 * +;
rup +1 ~$a5 +1 $a5 >= 1;
rup +1 $a5 +1 ~$a5 >= 1;
\end{lstlisting}
A very similar approach derives the simplified reifications of the
$d_i$ variables.
\begin{lstlisting}[language=veripb,numbers=none,frame=none,aboveskip=\smallskipamount,belowskip=\smallskipamount,xleftmargin=1ex]
rup +1 $d2 >= 1;
pol 39 ~$d2 3 * + 49 + s;
pol 40 -2 3 * + ~$a2 +;
rup +1 ~$d3 +1 $d3 >= 1;
rup +1 $d3 +1 ~$d3 >= 1;
pol 41 -2 3 * + -14 +;
pol 42 -2 3 * + -16 +;
rup +1 ~$d4 +1 $d4 >= 1;
rup +1 $d4 +1 ~$d4 >= 1;
pol 43 -2 3 * + -14 +;
pol 44 -2 3 * + -16 +;
rup +1 ~$d5 +1 $d5 >= 1;
rup +1 $d5 +1 ~$d5 >= 1;
pol 45 -2 3 * + -14 +;
pol 46 -2 3 * + -16 +;
\end{lstlisting}
The proof now proves the following constraints
\begin{align*}
    d_3 + \overline{s_1} & {} \geq 1 & t_1 + \overline{a_3} & {} \geq 1 \\
    d_4 + \overline{s_2} & {} \geq 1 & t_2 + \overline{a_4} & {} \geq 1 \\
    d_5 + \overline{s_3} & {} \geq 1 & t_3 + \overline{a_5} & {} \geq 1
\end{align*}
by RUP.
\begin{lstlisting}[language=veripb,numbers=none,frame=none,aboveskip=\smallskipamount,belowskip=\smallskipamount,xleftmargin=1ex]
rup +1 $d3 +1 ~s1 >= 1;
rup +1 $d4 +1 ~s2 >= 1;
rup +1 $d5 +1 ~s3 >= 1;
rup +1 t1 +1 ~$a3 >= 1;
rup +1 t2 +1 ~$a4 >= 1;
rup +1 t3 +1 ~$a5 >= 1;
\end{lstlisting}
Here, checking that $d_3 + \overline{s_1} \geq 1$ is implied by RUP
can be done by assuming literals $\overline{d_3}$ and $s_1$ are true,
and then propagating. The simplified specification constraint $2d_3 + \overline{x_1} + x_3 \geq 2$
then propagates $\overline{x_1}$ and $x_3$. This in turn yields a conflict
with the symmetry breaker circuit constraint $\overline{s_1} + x_1 + \overline{x_3} \geq 1$.
To show that $d_4 + \overline{s_2} \geq 2$ is implied by RUP,
we similarly assume literals  $\overline{d_4}$ and $s_2$ are true.
Then, the symmetry breaker circuit constraint $3\overline{s_2} + 2s_1 + x_2 + \overline{x_4} \geq 3$
propagates $s_1$, and so the previously derived constraint propagates $d_3$.
Finally, the simplified specification constraint $3d_4 + 3\overline{d_3} + a_3 + x_2 + \overline{x_4} \geq 3$
propagates $\overline{x_2}$ and $x_4$,
which in turn conflict with the same symmetry breaker circuit constraint above.
Similar RUP steps derive the rest of the constraints above, as well as the rest of the proof.
\begin{lstlisting}[language=veripb,numbers=none,frame=none,aboveskip=\smallskipamount,belowskip=\smallskipamount,xleftmargin=1ex]
rup +1 $d3 +1 ~s1 >= 1;
rup +1 $d4 +1 ~s2 >= 1;
rup +1 $d5 +1 ~s3 >= 1;
rup +1 t1 +1 ~$a3 >= 1;
rup +1 t2 +1 ~$a4 >= 1;
rup +1 t3 +1 ~$a5 >= 1;
rup +1 t2 +1 ~t1 +1 $d3 >= 1;
rup +1 t3 +1 ~t2 +1 $d4 >= 1;
rup +1 t4 +1 ~t3 +1 $d5 >= 1;
rup +1 $d4 +1 ~$d3 +1 t1 >= 1;
rup +1 $d5 +1 ~$d4 +1 t2 >= 1;
rup +1 $d6 +1 ~$d5 +1 t3 >= 1;
rup +1 $d3 +1 t1 >= 1;
rup +1 $d3 +1 t2 >= 1;
rup +1 $d4 +1 t1 >= 1;
rup +1 $d4 +1 t2 >= 1;
rup +1 $d4 +1 t3 >= 1;
rup +1 $d5 +1 t2 >= 1;
rup +1 $d5 +1 t3 >= 1;
rup +1 $d5 +1 t4 >= 1;
rup +1 $d6 +1 t3 >= 1;
rup +1 $d6 +1 t4 >= 1;
rup >= 1;
qed #1 : -1;
end scope;
\end{lstlisting}
Finally, the proof derives $d_6 + t_4 \geq 1$, which immediately yields a contradiction:
$\overline{t_4} \geq 1$ is assumed as the negation of the current dominance constraint,
and $\overline{d_6} \geq 1$ is assumed as the negation of the current proof goal.
\begin{lstlisting}[language=veripb,numbers=none,frame=none,aboveskip=\smallskipamount,belowskip=\smallskipamount,xleftmargin=1ex]
rup >= 1;
qed #1 : -1;
end scope;
\end{lstlisting}
Next, the proof shows the second dominance check.
\begin{lstlisting}[language=veripb,numbers=none,frame=none,aboveskip=\smallskipamount,belowskip=\smallskipamount,xleftmargin=1ex]
scope geq
proofgoal #2
\end{lstlisting}
The scope \lstinline[language=veripb]|geq| reintroduces the constraints
$\spec(\vec{x}\uh_\sigma, \vec{x}, \vec{a},\vec{d})$ with constraint
identifiers \texttt{100} through \texttt{121}, just like \lstinline[language=veripb]|leq|.
Additionally, the single constraint in $\ord(\vec{x}\uh_\sigma, \vec{x}, \vec{a},\vec{d})$
is introduced as well with identifier \texttt{122}.
The proof itself proceeds similarly to the \lstinline[language=veripb]|leq| check,
by deriving constraints that connect $\vec{a}$, $\vec{d}$, $\vec{s}$ and $\vec{t}$.
The RUP steps follow closely Lemma~\ref{lem:second_proof_goal}.
\begin{lstlisting}[language=veripb,numbers=none,frame=none,aboveskip=\smallskipamount,belowskip=\smallskipamount,xleftmargin=1ex]
rup +1 $d5 >= 1;
rup +1 $d4 >= 1;
rup +1 $d3 >= 1;
rup +1 ~s1 +1 $a3 >= 1;
rup +1 ~s2 +1 $a4 >= 1;
rup +1 ~s3 +1 $a5 >= 1;
rup +1 t1 >= 1;
rup +1 t2 >= 1;
rup +1 t3 >= 1;
\end{lstlisting}
The proof eventually derives $t_3 \geq 1$ and $s_3 + \overline{a_5} \geq 1$.
This leads to a contradiction,
since the assumptions $\overline{t_4} \geq 1$ and $d_6 \geq 1$ then
propagate $x_4$, $\overline{x_2}$ and $s_3$ from the symmetry breaker circuit constraint
$3t_4 + 3\overline{t_3} + s_3 + x_4 + \overline{x_2} \geq 3$.
The constraint $s_3 + \overline{a_5} \geq 1$ then propagates $a_5$,
and then the constraint $4\overline{d_6} + 3d_5 + \overline{a_5} + x_2 + \overline{x_4} \geq 4$
yields a conflict. This concludes the dominance proof.
\begin{lstlisting}[language=veripb,numbers=none,frame=none,aboveskip=\smallskipamount,belowskip=\smallskipamount,xleftmargin=1ex]
rup >= 1;
qed #2 : -1;
end scope;
qed dom;
\end{lstlisting}
The proof next simplifies the constraints and deletes the circuit,
leaving only some constraints that use the $\vec{s}$ variables
to break the symmetry $\sigma$.
\begin{lstlisting}[language=veripb,numbers=none,frame=none,aboveskip=\smallskipamount,belowskip=\smallskipamount,xleftmargin=1ex]
rup +1 t3 >= 1 : -1 22;
rup +1 t2 >= 1 : -1 20;
rup +1 t1 >= 1 : -1 18;
pol 11 x1 + s;
pol 13 x2 + s;
pol 15 x3 + s;
pol 11 ~x3 + s;
pol 13 ~x4 + s;
pol 15 ~x1 + s;
pol 16 136 + s;
pol 18 ~t1 3 * + 135 4 * + s;
pol 20 ~t2 3 * + 134 4 * + s;
pol 22 ~t3 3 * + 133 4 * + s;
del range 10 26;
del range 133 137;
\end{lstlisting}
By the end of this fragment, the net effect of symmetry breaking
has been adding the following constraints to the derived set:
\begin{align*}
x_3 + s_1 & {} \geq 1 &
x_4 + \overline{s_1} + s_2 & {} \geq 1 \\
x_1 + \overline{s_2} + s_3 & {} \geq 1 &
\overline{x_1} + s_1 & {} \geq 1 \\
\overline{x_2} + \overline{s_1} + s_2 & {} \geq 1 &
\overline{x_3} + \overline{s_2} + s_3 & {} \geq 1 \\
\overline{x_1} + x_3 & {} \geq 1 &
\overline{x_2} + x_4 + \overline{s_1} & {} \geq 1 \\
x_1 + \overline{x_3} + \overline{s_2} & {} \geq 1 & 
x_2 + \overline{x_4} + \overline{s_3} & {} \geq 1
\end{align*}
Since these constraints do not affect the core set,
we can still break another symmetry. The symmetry $\tau$ is larger than $\sigma$,
but the proof proceeds through similar strokes.
First, a symmetry breaker circuit is introduced through redundance-based strengthening.
\begin{lstlisting}[language=veripb,numbers=none,frame=none,aboveskip=\smallskipamount,belowskip=\smallskipamount,xleftmargin=1ex]
red +1 ~s4 +1 x5 +1 ~x4 >= 1 : s4 -> 0;
red +2 s4 +1 ~x5 +1 x4 >= 2 : s4 -> 1;
red +3 ~s5 +2 s4 +1 x6 +1 ~x3 >= 3 : s5 -> 0;
red +2 s5 +2 ~s4 +1 ~x6 +1 x3 >= 2 : s5 -> 1;
red +3 ~s6 +2 s5 +1 x1 +1 ~x6 >= 3 : s6 -> 0;
red +2 s6 +2 ~s5 +1 ~x1 +1 x6 >= 2 : s6 -> 1;
red +3 ~s7 +2 s6 +1 x2 +1 ~x5 >= 3 : s7 -> 0;
red +2 s7 +2 ~s6 +1 ~x2 +1 x5 >= 2 : s7 -> 1;
red +3 ~s8 +2 s7 +1 x3 +1 ~x2 >= 3 : s8 -> 0;
red +2 s8 +2 ~s7 +1 ~x3 +1 x2 >= 2 : s8 -> 1;
red +1 ~t1 +1 ~x5 +1 x4 >= 1 : t1 -> 0;
red +2 t1 +1 x5 +1 ~x4 >= 2 : t1 -> 1;
red +4 ~t2 +3 t1 +1 ~s4 +1 ~x6 +1 x3 >= 4 : t2 -> 0;
red +3 t2 +3 ~t1 +1 s4 +1 x6 +1 ~x3 >= 3 : t2 -> 1;
red +4 ~t3 +3 t2 +1 ~s5 +1 ~x1 +1 x6 >= 4 : t3 -> 0;
red +3 t3 +3 ~t2 +1 s5 +1 x1 +1 ~x6 >= 3 : t3 -> 1;
red +4 ~t4 +3 t3 +1 ~s6 +1 ~x2 +1 x5 >= 4 : t4 -> 0;
red +3 t4 +3 ~t3 +1 s6 +1 x2 +1 ~x5 >= 3 : t4 -> 1;
red +4 ~t5 +3 t4 +1 ~s7 +1 ~x3 +1 x2 >= 4 : t5 -> 0;
red +3 t5 +3 ~t4 +1 s7 +1 x3 +1 ~x2 >= 3 : t5 -> 1;
red +4 ~t6 +3 t5 +1 ~s8 +1 ~x4 +1 x1 >= 4 : t6 -> 0;
red +3 t6 +3 ~t5 +1 s8 +1 x4 +1 ~x1 >= 3 : t6 -> 1;
\end{lstlisting}
Then, a symmetry breaking constraint is introduced through dominance-based strengthening.
\begin{lstlisting}[language=veripb,numbers=none,frame=none,aboveskip=\smallskipamount,belowskip=\smallskipamount,xleftmargin=1ex]
dom +1 t6  >= 1 : x5 x4 x6 x3 x1 x6 x2 x5 x3 x2 x4 x1  : subproof
\end{lstlisting}
Then the first dominance proof goal is resolved.
\begin{lstlisting}[language=veripb,numbers=none,frame=none,aboveskip=\smallskipamount,belowskip=\smallskipamount,xleftmargin=1ex]
scope leq
proofgoal #1
rup +1 ~$d6 >= 1;
rup >= 0;
pol 170;
pol 171;
rup +1 ~$a1 +1 $a1 >= 1;
rup +1 $a1 +1 ~$a1 >= 1;
pol 172 -2 2 * +;
pol 173 -2 2 * +;
rup +1 ~$a2 +1 $a2 >= 1;
rup +1 $a2 +1 ~$a2 >= 1;
pol 174 -2 2 * +;
pol 175 -2 2 * +;
rup +1 ~$a3 +1 $a3 >= 1;
rup +1 $a3 +1 ~$a3 >= 1;
pol 176 -2 2 * +;
pol 177 -2 2 * +;
rup +1 ~$a4 +1 $a4 >= 1;
rup +1 $a4 +1 ~$a4 >= 1;
pol 178 -2 2 * +;
pol 179 -2 2 * +;
rup +1 ~$a5 +1 $a5 >= 1;
rup +1 $a5 +1 ~$a5 >= 1;
rup >= 0;
pol 180;
pol 181;
rup +1 ~$d1 +1 $d1 >= 1;
rup +1 $d1 +1 ~$d1 >= 1;
pol 182 -2 3 * + -22 +;
pol 183 -2 3 * + -24 +;
rup +1 ~$d2 +1 $d2 >= 1;
rup +1 $d2 +1 ~$d2 >= 1;
pol 184 -2 3 * + -22 +;
pol 185 -2 3 * + -24 +;
rup +1 ~$d3 +1 $d3 >= 1;
rup +1 $d3 +1 ~$d3 >= 1;
pol 186 -2 3 * + -22 +;
pol 187 -2 3 * + -24 +;
rup +1 ~$d4 +1 $d4 >= 1;
rup +1 $d4 +1 ~$d4 >= 1;
pol 188 -2 3 * + -22 +;
pol 189 -2 3 * + -24 +;
rup +1 ~$d5 +1 $d5 >= 1;
rup +1 $d5 +1 ~$d5 >= 1;
pol 190 -2 3 * + -22 +;
pol 191 -2 3 * + -24 +;
rup +1 $d1 +1 ~s4 >= 1;
rup +1 $d2 +1 ~s5 >= 1;
rup +1 $d3 +1 ~s6 >= 1;
rup +1 $d4 +1 ~s7 >= 1;
rup +1 $d5 +1 ~s8 >= 1;
rup +1 t1 +1 ~$a1 >= 1;
rup +1 t2 +1 ~$a2 >= 1;
rup +1 t3 +1 ~$a3 >= 1;
rup +1 t4 +1 ~$a4 >= 1;
rup +1 t5 +1 ~$a5 >= 1;
rup +1 t2 +1 ~t1 +1 $d1 >= 1;
rup +1 t3 +1 ~t2 +1 $d2 >= 1;
rup +1 t4 +1 ~t3 +1 $d3 >= 1;
rup +1 t5 +1 ~t4 +1 $d4 >= 1;
rup +1 t6 +1 ~t5 +1 $d5 >= 1;
rup +1 $d2 +1 ~$d1 +1 t1 >= 1;
rup +1 $d3 +1 ~$d2 +1 t2 >= 1;
rup +1 $d4 +1 ~$d3 +1 t3 >= 1;
rup +1 $d5 +1 ~$d4 +1 t4 >= 1;
rup +1 $d6 +1 ~$d5 +1 t5 >= 1;
rup +1 $d1 +1 t1 >= 1;
rup +1 $d1 +1 t2 >= 1;
rup +1 $d2 +1 t1 >= 1;
rup +1 $d2 +1 t2 >= 1;
rup +1 $d2 +1 t3 >= 1;
rup +1 $d3 +1 t2 >= 1;
rup +1 $d3 +1 t3 >= 1;
rup +1 $d3 +1 t4 >= 1;
rup +1 $d4 +1 t3 >= 1;
rup +1 $d4 +1 t4 >= 1;
rup +1 $d4 +1 t5 >= 1;
rup +1 $d5 +1 t4 >= 1;
rup +1 $d5 +1 t5 >= 1;
rup +1 $d5 +1 t6 >= 1;
rup +1 $d6 +1 t5 >= 1;
rup +1 $d6 +1 t6 >= 1;
rup >= 1;
qed #1 : -1;
end scope;
\end{lstlisting}
The second proof goal follows, and finishes the dominance proof:
\begin{lstlisting}[language=veripb,numbers=none,frame=none,aboveskip=\smallskipamount,belowskip=\smallskipamount,xleftmargin=1ex]
scope geq
proofgoal #2
rup +1 $d5 >= 1;
rup +1 $d4 >= 1;
rup +1 $d3 >= 1;
rup +1 $d2 >= 1;
rup +1 $d1 >= 1;
rup +1 ~s4 +1 $a1 >= 1;
rup +1 ~s5 +1 $a2 >= 1;
rup +1 ~s6 +1 $a3 >= 1;
rup +1 ~s7 +1 $a4 >= 1;
rup +1 ~s8 +1 $a5 >= 1;
rup +1 t1 >= 1;
rup +1 t2 >= 1;
rup +1 t3 >= 1;
rup +1 t4 >= 1;
rup +1 t5 >= 1;
rup >= 1;
qed #2 : -1;
end scope;
qed dom;
\end{lstlisting}
Finally, some constraints are simplified.
\begin{lstlisting}[language=veripb,numbers=none,frame=none,aboveskip=\smallskipamount,belowskip=\smallskipamount,xleftmargin=1ex]
rup +1 t5 >= 1 : -1 167;
rup +1 t4 >= 1 : -1 165;
rup +1 t3 >= 1 : -1 163;
rup +1 t2 >= 1 : -1 161;
rup +1 t1 >= 1 : -1 159;
pol 148 x5 + s;
pol 150 x6 + s;
pol 152 x1 + s;
pol 154 x2 + s;
pol 156 x3 + s;
pol 148 ~x4 + s;
pol 150 ~x3 + s;
pol 152 ~x6 + s;
pol 154 ~x5 + s;
pol 156 ~x2 + s;
pol 157 319 + s;
pol 159 ~t1 3 * + 318 4 * + s;
pol 161 ~t2 3 * + 317 4 * + s;
pol 163 ~t3 3 * + 316 4 * + s;
pol 165 ~t4 3 * + 315 4 * + s;
pol 167 ~t5 3 * + 314 4 * + s;
del range 147 171;
del range 314 320;
\end{lstlisting}
Through the second symmetry breaking proof,
the following constraints are ultimately introduced:
\begin{align*}
x_4 + s_4 & {} \geq 1 &
x_3 + \overline{s_4} + s_5 & {} \geq 1 \\
x_6 + \overline{s_5} + s_6 & {} \geq 1 &
x_5 + \overline{s_6} + s_7 & {} \geq 1 \\
x_2 + \overline{s_7} + s_8 & {} \geq 1 &
\overline{x_5} + s_4 & {} \geq 1 \\
\overline{x_6} + \overline{s_4} + s_5 & {} \geq 1 &
\overline{x_1} + \overline{s_5} + s_6 & {} \geq 1 \\
\overline{x_2} + \overline{s_6} + s_7 & {} \geq 1 &
\overline{x_3} + \overline{s_7} + s_8 & {} \geq 1 \\
x_4 + \overline{x_5} & {} \geq 1 &
x_3 + \overline{x_6} + \overline{s_4} & {} \geq 1 \\
\overline{x_1} + x_6 + \overline{s_5} & {} \geq 1 &
\overline{x_2} + x_5 + \overline{s_6} & {} \geq 1 \\
x_2 + \overline{x_3} + \overline{s_7} & {} \geq 1 &
x_1 + \overline{x_4} + \overline{s_8} & {} \geq 1  
\end{align*}

\begin{table*}
\centering
\begin{tabular}{c|cccc}
instance & variables & constraints & generators & support size \\ \hline
$\text{PHP}(n)$ & $\bigoh{n^2}$ & $\bigoh{n^3}$ & $\bigoh{n}$ & $\bigoh{n}$ \\
$\text{RPHP}(n)$ & $\bigoh{n^2}$ & $\bigoh{n^3}$ & $\bigoh{n}$ & $\bigoh{n}$ \\
$\text{ClqCl}(n,6,5)$ & $\bigoh{n^2}$ & $\bigoh{n^2}$ & $\bigoh{n}$ & $\bigoh{n}$ \\
$\text{TseitinGrid}(n)$ & $\bigoh{n^2}$ & $\bigoh{n^2}$ & $\bigoh{n^2}$ & varies \\
$\text{Count}(n,3)$ & $\bigoh{n^3}$ & $\bigoh{n^5}$ & $\bigoh{n}$ & varies
\end{tabular}
\caption{Some data on the families of crafted benchmarks that we used: The number of variables and constraints of the instance, 
the number of generators that \satsuma finds, and the support size of the symmetries that \satsuma finds.}
\end{table*}

\section{Details on Crafted Benchmarks}
\label{app:crafted}

In this appendix, we briefly describe the five crafted benchmark families that we used in our experimental evaluation. %
We use the standard notation $[n] = \{1, \ldots, n\}$. 

\paragraph{Pigeonhole principle (PHP).} The classical  \emph{pigeonhole principle} formula~\cite{Haken85Intractability} encodes the claim that $n$ pigeons can fly into $n-1$ holes such that no two pigeons fly into the same hole. Clearly, this formula is UNSAT. The encoding has $n(n-1)$ variables $x_{ij}$ for $i \in [n]$ and $j \in [n-1]$, where $x_{ij}$ encodes that pigeon $i$ flies into hole $j$. There are $(n-1)\binom{n}{2}$ binary clauses of the form $\olnot{x_{ik}} \vee \olnot{x_{jk}}$, for all $1 \leq i < j \leq n$ and $k \in [n-1]$, claiming that it is not the case that both pigeon $i$ and pigeon $j$ fly into hole $k$. Together, these clauses encode that no two pigeons fly into the same hole. In addition, there are $n$ clauses of the form $\bigvee_{1 \leq j \leq n-1} x_{ij}$ for $i \in [n]$, encoding that pigeon $i$ must fly into some hole. The symmetries of PHP are permuting the holes and permuting the pigeons. A set of generators for the corresponding symmetry group are $n-2$ symmetries that swap two holes and $n-1$ symmetries that swap two pigeons. Each such symmetry affects $\bigoh{n}$ variables.

\paragraph{Relativized pigeonhole principle (RPHP).} The \emph{relativized pigeonhole principle} formula~\cite{AMO13LowerBounds,ALN16NarrowProofs} encodes the claim that $n$ pigeons can fly into $m$ resting places and then onwards to $n-1$ holes such that no two pigeons fly into the same resting place or the same hole. Formally, the RPHP formula claims that there exist maps $p \colon [n] \rightarrow [m]$ and $q \colon [m] \rightarrow [n-1]$ such that $p$ is injective and $q$ is injective on the range of $p$. Similarly to PHP, this formula is UNSAT.  Here we choose $m = 2n$. Then the formula contains $4n^2$ variables and $\bigoh{n^3}$ constraints. The symmetries of RPHP are permuting the holes, permuting the pigeons, and permuting the resting places.

\paragraph{Clique-coloring formulas (ClqCl).} The \emph{clique-coloring} formula~\cite{Pudlak97LowerBounds} encodes the claim that there exists a graph on $n$ vertices that contains a clique on $k$ vertices and admits a coloring on $c$ vertices. This formula is UNSAT for $k > c$. The formula contains $\binom{n}{2}$ variables to encode the edges of the graph, $kn$ variables to encode a function from $[k]$ to $[n]$ representing the clique, and $nc$ variables to encode a function from $[n]$ to $[c]$ representing the coloring. 
 Here we choose $k=6$ and $c=5$. 
 Then the formula has  $\bigoh{n^2}$ constraints.
 The symmetries of the clique-coloring formula are the symmetries permuting the vertices of the graph. 
 
\paragraph{Counting formulas (Count).}  The \emph{counting} formula~\cite{PW85Counting}  encodes the claim that we can partition $n$ elements into sets of size $k$ each. The formula has $\binom{n}{k}$ variables, where each variable encodes that a particular set is chosen. There are $\bigoh{n^{2k-1}}$ constraints, encoding that any two chosen sets must be disjoint. In our experiments, we choose $k=3$. 
  The symmetries of the counting formula are permuting the elements.
  
\paragraph{Tseitin formulas on a grid (Tseitin Grid).} The \emph{Tseitin formulas}~\cite{Tseitin68ComplexityTranslated} consider a graph $G = (V, E)$ and a charge function $f \colon V \rightarrow \{0,1\}$. The formula has a variable for each edge, and the formula encodes that the XOR of the variables corresponding to the edges adjacent to some vertex $v$ equals the charge $f(v)$. Here we consider the Tseitin formula on an $n \times n$ grid with even charge, \ie $f(v) = 0$ for all $v \in V$.
For each cycle in the graph $G$ the Tseitin formula has a negation symmetry that flips all variables corresponding to edges in that cycle.

}{}

\bibliography{refArticles,refBooks,refOther,refLocal}

\end{document}